\documentclass[oneside,11pt]{article}

\usepackage[utf8]{inputenc}
\usepackage[T1]{fontenc}

\usepackage{amsthm}
\usepackage{amsmath}
\usepackage{amsfonts}
\usepackage{amssymb}

\usepackage{ifplatform}

\usepackage{bbm}
\usepackage{dsfont}

\usepackage{graphicx}
\usepackage{float}
\usepackage{booktabs}
\usepackage{multirow}
\usepackage{rotating}
\usepackage{array}

\usepackage{breakcites}

\usepackage{enumitem}

\usepackage[top=1.5cm,bottom=2cm,right=2.5cm,left=2.5cm]{geometry}
\usepackage{titling}
\usepackage{natbib}

\usepackage{ifplatform}

\ifwindows
\fi

\newcommand{\R}{\mathbb{R}}

\newcommand{\PR}{\mathbb{P}}

\newcommand{\N}{\mathbb{N}}

\newcommand{\lc}{\left[}
\newcommand{\rc}{\right]}
\newcommand{\lb}{\left\{}
\newcommand{\rb}{\right\}}

\newcommand{\prob}[1]{\PR\lb#1\rb}

\providecommand{\abs}[1]{\left|#1\right|}

\newcommand{\lrp}[1]{\left(#1\right)}
\newcommand{\lrc}[1]{\left[#1\right]}
\newcommand{\lrb}[1]{\left\{#1\right\}}
\newcommand{\E}[1]{\mathbb{E}\lc #1\rc}
\newcommand{\Ebig}[1]{\mathbb{E}\big[ #1\big]}

\newcommand{\V}[1]{\mathbb{V}\mathrm{ar}\lc #1\rc}

\newcommand{\norm}[1]{\left|\left| #1\right|\right|}

\newcommand{\OP}{\mathcal{O}_\mathbb{P}}
\newcommand{\op}{o_{\mathbb{P}}}
\newcommand{\Oh}{\mathcal{O}}
\DeclareFontFamily{OT1}{pzc}{}
\DeclareFontShape{OT1}{pzc}{m}{it}%
              {<-> s * [0.900] pzcmi7t}{}
\DeclareMathAlphabet{\mathpzc}{OT1}{pzc}%
                                 {m}{it}

\newcommand{\inprod}[2]{\left<#1,#2\right>}
\newcommand{\bX}{{\bf X}}
\newcommand{\bY}{{\bf Y}}
\newcommand{\bZ}{{\bf Z}}
\newcommand{\bL}{{\bf L}}
\newcommand{\bU}{{\bf U}}
\newcommand{\bV}{{\bf V}}
\newcommand{\bS}{{\bf S}}
\newcommand{\bR}{{\bf R}}
\newcommand{\bx}{{\bf x}}
\newcommand{\bz}{{\bf z}}
\newcommand{\by}{{\bf y}}
\newcommand{\be}{{\bf e}}
\newcommand{\bu}{{\bf u}}

\newcommand{\bt}{{\bf t}}
\newcommand{\bh}{{\bf h}}
\newcommand{\bXh}{\bX^\bh}

\newcommand{\brho}{\boldsymbol\rho}
\newcommand{\bpsi}{\boldsymbol\psi}
\newcommand{\bphi}{\boldsymbol\phi}

\newcommand{\btheta}{\boldsymbol\theta}

\newcommand{\Rho}{\mathcal P}
\newcommand{\inlaw}{\stackrel{\mathcal{L}}{\rightsquigarrow}}
\newcommand{\inprob}{\stackrel{p}{\rightarrow}}
\newcommand{\Medh}{\mbox{$\overline{\bX}_{x,\bh}$}}

\newcommand{\Hil}{\mbox{${\mathcal{H}}$}}
\newcommand{\conv}{\rightarrow}

\newtheorem{Prop}{Proposition} [section]
\newtheorem{Lemm}[Prop] {Lemma}
\newtheorem{Theo}[Prop]{Theorem}
\newtheorem{Coro}[Prop] {Corollary}
\newtheorem{Rem}{Remark}[Prop]
\newtheorem{Exam}[Prop] {Example}

\newtheorem{Algo}[Prop] {Algorithm}

\usepackage[pdftex,final,bookmarksnumbered,bookmarksopen=false,breaklinks,colorlinks]{hyperref}
\hypersetup{
	final,
    bookmarksnumbered=true,
    bookmarksopen=true,
    bookmarksopenlevel=0,
    unicode=false,          
    pdftoolbar=true,        
    pdfmenubar=true,        
    pdffitwindow=false,     
    pdftitle={Goodness-of-fit tests for the functional linear model based on randomly projected empirical processes},    
	pdfdisplaydoctitle=true, 
	pdftoolbar=true,
	pdfmenubar=true,
	pdflang={English},
    pdfauthor={Juan A. Cuesta-Albertos, Eduardo Garcia-Portugues, Manuel Febrero-Bande, Wenceslao Gonzalez-Manteiga},     
    pdfsubject={arXiv paper},   
    pdfcreator={Eduardo Garcia-Portugues},   
    pdfproducer={Eduardo Garcia-Portugues}, 
    pdfkeywords={Empirical process}{Functional data}{Functional linear model}{Functional principal components}{Goodness-of-fit}{Random projections}, 
    pdfnewwindow=true,      
    breaklinks=true,
    hidelinks,
    linkcolor=black,          
    citecolor=black,        
    filecolor=magenta,      
    urlcolor=cyan           
}

\allowdisplaybreaks

\newif\iffigures
\figurestrue

\newif\ifmain
\maintrue
\ifmain

\else

\fi
\newif\ifsupplement
\supplementtrue

\begin{document}

\ifmain

\title{Goodness-of-fit tests for the functional linear model based on randomly projected empirical processes}
\setlength{\droptitle}{-1cm}
\predate{}%
\postdate{}%

\date{}

\author{
Juan A. Cuesta-Albertos$^{1}$, Eduardo Garc\'ia-Portugu\'es$^{2,3,5}$,\\
Manuel Febrero-Bande$^{4}$, and Wenceslao Gonz\'alez-Manteiga$^{4}$}

\footnotetext[1]{
Department of Mathematics, Statistics and Computer Science, University of Cantabria (Spain).}
\footnotetext[2]{
Department of Statistics, Carlos III University of Madrid (Spain).}
\footnotetext[3]{
UC3M-BS Institute of Financial Big Data, Carlos III University of Madrid (Spain).}
\footnotetext[4]{
Department of Statistics, Mathematical Analysis and Optimization, University of Santiago de Compostela (Spain).}
\footnotetext[5]{Corresponding author. e-mail: \href{mailto:edgarcia@est-econ.uc3m.es}{edgarcia@est-econ.uc3m.es}.}
\maketitle

\begin{abstract}
We consider marked empirical processes indexed by a randomly projected functional covariate to construct goodness-of-fit tests for the functional linear model with scalar response. The test statistics are built from continuous functionals over the projected process, resulting in computationally efficient tests that exhibit root-$n$ convergence rates and circumvent the curse of dimensionality. The weak convergence of the empirical process is obtained conditionally on a random direction, whilst the almost surely equivalence between the testing for significance expressed on the original and on the projected functional covariate is proved. The computation of the test in practice involves calibration by wild bootstrap resampling and the combination of several $p$-values, arising from different projections, by means of the false discovery rate method. The finite sample properties of the tests are illustrated in a simulation study for a variety of linear models, underlying processes, and alternatives. The software provided implements the tests and allows the replication of simulations and data applications. 
\end{abstract}
\begin{flushleft}
\end{flushleft}

\section{Introduction}
\label{sec:intro}

The term ``goodness-of-fit'' was introduced at the beginning of the twentieth century by Karl Pearson, and, since then, there have been an enormous amount of papers devoted to this topic: first, concentrated on fitting a model for one distribution function, and, later, especially after the papers of \cite{Bickel1973} and \cite{Durbin1973}, on more general models related with the regression function. Considering a regression model with random design $Y=m(X)+\varepsilon$, the goal is to test the goodness-of-fit of a class of parametric regression functions $\mathcal{M}_\Theta:=\lrb{m_{\btheta}:\btheta\in\Theta\subset\R^q}$ to the data. This is the testing of 
\[ 
H_0: m\in \mathcal{M}_\Theta \quad\text{vs.}\quad H_1: m\notin \mathcal{M}_\Theta 
\] 
in an omnibus way from a sample $\lrb{\lrp{X_i,Y_i}}_{i=1}^n$ from $(X,Y)$. Here, $m(x)=\E{Y\vert X=x}$ is the regression function of $Y$ over $X$, and $\varepsilon$ is a random error centred such that $\E{\varepsilon|\bX}=0$. The literature of goodness-of-fit tests for the regression function is vast, and we refer to \cite{Gonzalez-Manteiga2013} for an updated review of the topic.\\ 

Following the ideas on smoothing for testing the density function \citep{Bickel1973}, the pilot estimators usually considered for $m$ were nonparametric, for example, the Nadaraya--Watson estimator (\cite{Nadaraya1964}, \cite{Watson1964}): $\hat m_h(x):=\sum_{i=1}^n W_{ni}(x)Y_i$, with $W_{ni}(x):=K\big((x-X_i)/h\big)\big/\allowbreak\sum_{j=1}^n K\big((x-X_j)/h\big)$, where $K$ is a kernel function and $h$ is a bandwidth parameter. Using these kinds of pilot estimators, statistical tests were given by $T_{n}=d\lrp{\hat m,m_{\hat\btheta}}$, with $d$ some functional distance and $\hat \btheta$ an estimator of $\btheta$ such that $\sqrt{n}(\hat \btheta-\btheta)=\OP(1)$ under $H_0$. Alternatively, following the paper by \cite{Durbin1973} for testing about the distribution, the pilot estimator in the regression case was given by $I_n(x):=n^{-1}\sum_{i=1}^n \mathds{1}_{\lrb{X_i\leq x}}Y_i $, and the empirical estimation of the integrated regression function was then $I(x):=\E{\mathds{1}_{\lrb{X\leq x}}Y}$. \cite{Hardle1993}, using $\hat m_h$, and \cite{Stute1997}, using $I_n$, are key references for these two approximations in the literature, and were only the beginning of more than two hundred papers published in the last two decades \citep{Gonzalez-Manteiga2013}.\\ 

More recently, there has been a growing interest in testing possible structures in a regression setting in the presence of functional covariates: 
\begin{align} 
Y=m(\bX)+\varepsilon,\label{funcmod} 
\end{align} 
with $\bX$ a random element in a functional space, for example, in the Hilbert space $\Hil=L^2[0,1]$, and $Y$ a scalar response. This is the context of ``Functional Data Analysis'', which has received increasing attention in the last decade (see, e.g., \cite{Ramsay2005}, \cite{Ferraty2006}, and \cite{Horvath2012}) due to the practical need to analyse data generated by high-resolution measuring devices. \\ 

A simple null hypothesis $H_0$ considered in the literature for model \eqref{funcmod} is $H_0: m(\bX)=c$, where $c\in\R$ is a fixed constant, that is, the testing of significance of the covariate $\bX$ over $Y$. Following some of the ideas from \cite{Ferraty2006} on considering pseudometrics for performing smoothing with functional data, the test by \cite{Hardle1993} was adapted by \cite{Delsol2011}\nolinebreak[4] as 
\begin{align*} 
T_{n,h}^\mathrm{D}:=&\;\int \lrp{\hat m_h(\bx)-\bar Y}^2\omega(\bx)\,\mathrm{d}P_{\bX}(\bx)=d(\hat m_h,\bar Y),\\ 
\hat m_h(\bx):=&\;\sum_{i=1}^n\bigg[K\lrp{\frac{\bar d(\bx,\bX_i)}{h}}Y_i\bigg/\sum_{j=1}^n K\lrp{\frac{\bar d(\bX_i,\bX_j)}{h}}\bigg], 
\end{align*} 
with $\bar d$ a functional pseudometric, $K$ a kernel function adapted to this situation, $h$ a bandwidth parameter, $\omega$ a weight function, and $P_{\bX}$ the probability measure induced by ${\bX}$ in $\Hil$. Testing $H_0$ has also been considered by \cite{Cardot2003} and \cite{Hilgert2014}, not in an omnibus way, but inside a Functional Linear Model (FLM): $m(\bX)=\inprod{\bX}{\brho}$, where $\langle \cdot,\cdot \rangle$ represents the inner product in $\Hil$ and $\brho\in\Hil$ is the FLM parameter. For both approximations, omnibus or not, there have also been other papers which consider the functional response case; see, for example, \cite{Chiou2007}, \cite{Kokoszka2008}, and \cite{Bucher2011}. \\ 

The generalization of the hypothesis $H_0: m(\bX)=c$ to the general case 
\begin{align} 
H_0: m\in\mathcal{M}_\Rho=\lrb{m_{\brho}:\brho\in\Rho}\quad\text{vs.}\quad H_1:m\notin\mathcal{M}_\Rho,\label{null} 
\end{align} 
where $\Rho$ can be of either finite or infinite dimension, has been the focus of very few papers, particularly in the context of omnibus goodness-of-fit tests. In \cite{Delsol2011a}, a discussion is given, without theoretical results, for the extension of the testing of a more complex null hypothesis, such as an FLM. Only one paper is known to us in which the FLM hypothesis is analysed with theoretical results. In \cite{Patilea2012}, motivated by the smoothing test statistic considered by \cite{Zheng1996} for finite dimensional covariates, a test based on 
\begin{align*} 
T_{n,h}^\mathrm{P}:=\frac{1}{n(n-1)}\sum_{1\leq i\neq j\leq n}&\lrp{Y_i-\hat m_{H_0}(\bX_i)}\lrp{Y_j-\hat m_{H_0}(\bX_j)}\\ 
\times&\frac{1}{h}K\lrp{\frac{F_{n,\bh}(\inprod{\bX_i}{\bh})-F_{n,\bh}(\inprod{\bX_j}{\bh})}{h}}, 
\end{align*} 
is employed for checking the null hypothesis of linearity with $\hat m_{H_0}(\bX):=\inprod{\bX}{\hat\brho}$, $\hat \brho$ a suitable estimator of $\brho$, and $F_{n,\bh}$ the empirical distribution function of $\{\inprod{\bX_i}{\bh}\}_{i=1}^n$. In the same spirit, \cite{Lavergne2008} developed a test for the finite dimensional context, and \cite{Patilea2014} provided a test for functional response. From a different perspective, and motivated by the test given by \cite{Escanciano2006} for finite dimensional predictors, \cite{Garcia-Portugues:flm} constructed a test from the marked empirical process $I_{n,\bh}(x):=\frac{1}{n}\sum_{i=1}^n\mathds{1}_{\lrb{\inprod{\bX_i}{\bh}\leq x}}Y_i$, with $x\in\R$, and $\bh \in \mathcal{H}$. The test statistic averages the Cram\'er--von Mises norm of $I_{n,\bh}$ over a finite-dimensional, estimation-driven space of random directions $\bh$. Although this approach circumvents the technical difficulties that a marked empirical process indexed by $\bx\in\Hil$ would represent (a possible functional extension of the process given in \cite{Stute1997}), no results on the convergence of the statistic are available. \\ 

In this paper, we consider marked empirical processes indexed by random projections of the functional covariate. The motivation stems from the almost surely characterization of the null hypothesis \eqref{null} via a \textit{projected hypothesis} that arises from the conditional expectation on the projected functional covariate. This allows, conditionally on a randomly chosen $\bh$, the study of the weak convergence of the process $I_{n,\bh}(x)$ for hypothesis testing with infinite-dimensional covariates and parameters. As a by-product, we obtain root-$n$ goodness-of-fit tests that evade the curse of dimensionality and, contrary to smoothing-based tests, do not rely on a tuning parameter. In particular, we focus on the testing of the aforementioned hypothesis of functional linearity where, contrary to the finite dimensional situation, the functional estimator has a nontrivial effect on the limiting process and requires careful regularization. The test statistics are built by a continuous functional (Kolmogorov--Smirnov or Cram\'er--von Mises) over the empirical process and are effectively calibrated by a wild bootstrap on the residuals. To account for a higher power and less influence from $\bh$, we consider a number $K$ (not to be confused with a kernel function) of different random directions and merge the resulting $p$-values into a final $p$-value by means of the False Discovery Rate (FDR) of \cite{Benjamini2001}. The empirical analysis reports a competitive performance of the test in practice, with a low impact of the choice of $K$ above a certain bound, and an expedient computational complexity of $\Oh(n)$ that yields notable speed improvements over \cite{Garcia-Portugues:flm}.\\ 

The rest of the paper is organized as follows. The characterization of the null hypothesis through the projected predictor is addressed in Section \ref{sec:hyproj}, together with an application for the testing of the null hypothesis $H_0: m=m_0$ (Subsection \ref{subsec:simple}). Section \ref{Sec.lin} is devoted to testing the composite hypothesis $H_0: m\in\{\inprod{\cdot}{\brho}:\brho\in\Hil\}$. To that aim, the regularized estimator for $\brho$ of \cite{Cardot2007}, $\hat{\brho}$, is reviewed in Subsection \ref{subsec:rho}. The pointwise asymptotic distribution of the projected process is studied in Subsection \ref{subsec:point}, whereas Subsection \ref{subsec:weak} gives its weak convergence. Section \ref{sec:testing} describes the implementation of the test and other practicalities. Section \ref{sec:simu} illustrates the finite sample properties of the test through a simulation study and with some real data applications. Some final comments and possible extensions are given in Section \ref{sec:final}. Appendix \ref{ap:proofs} presents the main proofs, whereas the supplementary material contains the auxiliary lemmas and further results from the simulation study.  

\subsection{General setting and notation}

Some of the general setting and notation considered in the paper are introduced now, while more specific notation will be introduced when required. The random variable (r.v.) $\bX$ belongs to a separable Hilbert space $\Hil$ endowed with the inner product $\langle \cdot,\cdot \rangle$ and associated norm $\|\cdot\|$. The space $\Hil$ is a general real Hilbert space, but, for simplicity, it can be regarded as $\Hil=L^2[0,1]$. $Y$ and $\bX$ are assumed to be centred r.v.'s providing an independent and identically distributed (i.i.d.) sample $\{(\bX_i,Y_i)\}_{i=1}^n\subset\mathcal{H}\times\mathbb{R}$. $\varepsilon$ is a centred r.v. with variance $\sigma^2_\varepsilon$ that is independent from $\bX$ (the independence between $\varepsilon$ and $\bX$ is a technical assumption required for proving Lemmas \ref{LemmTn2.Rn} and \ref{LemmTn2.Sn}, while for the rest of the paper it suffices that $\E{\varepsilon|\bX}=0$). Given the $\mathcal{H}$-valued r.v. $\bX$ and $\bh \in \Hil$, we denote by $\bX^\bh:= \langle \bX,\bh\rangle$ the projected $\bX$ in the direction $\bh$,  by $F_{\bh}$ the distribution function of $\bX^\bh$, and by $P_\bX$ the probability measure of $\mathbf{X}$ in $\Hil$. Bold letters are used for vectors in $\mathcal{H}$ (mainly) or column vectors in $\R^p$ (whose transposition is denoted by $'$), and the type is clearly determined by the context. Capital letters represent r.v.'s defined on the same probability space $(\Omega, \sigma, \nu)$ and $\sim$ denotes equality in distribution. Weak convergence is denoted by $\inlaw$ and $D(\R)$ represents the Skorohod's space of \textit{c\`adl\`ag} functions defined on $\R$. Finally, we shall implicitly assume that the null hypotheses stated hold almost surely (a.s.). 

\section{Hypothesis projection}
\label{sec:hyproj}

The pillar of the goodness-of-fit tests we present is the a.s. characterization of the null hypothesis \eqref{null}, re-expressed as $H_0: \mathbb{E}[Y-m_{\brho}(\bX) | \bX]=0$ for some $\brho\in\Rho$, by means of the associated \textit{projected hypothesis on $\bh\in\Hil$}, defined as $H_0^\bh: \E{Y-m_{\brho}(\bX) | \bX^{\bh}}=0$. In the following, we identify $Y-m_{\brho}(\bX)$ by $Y$ for the sake of simplicity in notation. In this section, we give two necessary and sufficient conditions based on the projections of $\bX$ such that $\E{Y | \bX}=0$ holds a.s. \\ 

The first condition only requires the integrability of $Y$, but the condition needs to be satisfied for every direction $\bh$. 
\begin{Prop}
\label{PropNoRandom}
Assume that $\E{|Y|}<\infty$. Then
\[
\E{Y|\bX}=0\text{ a.s.} \iff \mathbb{E}\big[Y | \bXh\big]=0\text{ a.s. for every }\bh \in \Hil.
\]
\end{Prop}

The second condition, more adequate for application, \textit{somehow} generalizes Proposition \ref{PropNoRandom}, as it only needs to be satisfied for a randomly chosen $\bh$. In exchange, it holds only under some additional conditions on the moments of $\bX$ and $Y$. Before stating it, we need some preliminary results, the first taken from \cite{Cuesta-Albertos2007a} and included here for the sake of completeness. 

\begin{Lemm}[Theorem 4.1 in \cite{Cuesta-Albertos2007a}]
\label{T:cwgauss}
Let $\mu$ be a nondegenerate Gaussian measure on $\Hil$ and $\bX_1,\bX_2$ be two $\Hil$-valued r.v.'s defined on $(\Omega, \sigma, \nu)$. Assume that:
\begin{enumerate}[label=(\alph{*}),ref=(\alph{*})]
\item $m_k:=\int\|\bX_1\|^k\, \mathrm{d}\nu<\infty$, for all $k\geq1$, and 
$\sum_{k=1}^\infty m_k^{-1/k}=\infty$.\label{T:cwgauss:a}
\item
The set $\{ \bh \in \Hil: \bX^\bh_1 \sim \bX^\bh_2\}$ is of positive $\mu$-measure.\label{T:cwgauss:b}
\end{enumerate}
Then $\bX_1 \sim \bX_2$.
\end{Lemm}

\begin{Rem}
The Gaussianity of $\mu$ in Lemma \ref{T:cwgauss} is not strictly required. It can be replaced by assuming a certain smoothness condition on $\mu$ (see, for instance, Theorem 2.5 and Example 2.6 in \cite{Cuesta-Albertos2007b}).
\end{Rem}

\begin{Rem}
Assumption \ref{T:cwgauss:a} in Lemma \ref{T:cwgauss} is not of a technical nature. According to Theorem 3.6 in \cite{Cuesta-Albertos2007a}, it becomes apparent that a similar condition is required. This assumption is satisfied if the tails of $P_{\bX_1}$ are light enough or if $\bX_1$ has a finite moment generating function in a neighbourhood of zero.
\end{Rem}
\begin{Lemm}
\label{Prop:Auxiliar}
If $\E{Y^2}<\infty$ and $\bX$ satisfies \ref{T:cwgauss:a} in Lemma \ref{T:cwgauss},
then $l_k:=\E{\|\bX\|^k  |Y|}\allowbreak < \infty$ for all $k\geq1$, and $\sum_{k=1}^\infty l_k^{-1/k}=\infty$.
\end{Lemm}
The second condition, and most important result in this section, is given as follows.
\begin{Theo}
\label{Th:basic}
Let $\mu$ be a nondegenerate Gaussian measure on $\Hil$. Assume that $\bX$ satisfies \ref{T:cwgauss:a} in Lemma \ref{T:cwgauss} and that $\E{Y^2} < \infty$. If we denote $\Hil_0:=\big\{\bh \in \Hil: \mathbb{E}\big[Y | \bXh\big]=0\text{ a.s.}\big\}$, then
\[
\E{Y | \bX}=0\text{ a.s.}\iff \Hil_0\text{ has positive $\mu$-measure.}
\]
\end{Theo}

\begin{Coro}
\label{Coro1}
Under the assumptions of the previous theorem,
\[
\E{Y | \bX}=0
\text{ a.s.}\iff
\mu\big(\Hil_0\big)=1.
\]
\end{Coro}
According to this corollary, if we are interested in testing the simple null hypothesis $H_0: \E{Y | \bX}=0$, then we can do so as follows: (\textit{i}) select, at random with $\mu$, a direction $\bh \in \Hil$; (\textit{ii}) conditionally on $\bh$, test the projected null hypothesis $H_0^\bh: \mathbb{E}\big[Y | \bX^\bh\big]=0$. The rationale is simple yet powerful: if $H_0$ holds, then $H_0^\bh$ also holds; if $H_0$ fails, then $H_0^\bh$ also fails $\mu$-a.s. In this case, with probability one, we have chosen a direction $\bh$ for which $H_0^\bh$ fails. Of course, the main advantage to testing $H_0^\bh$ over testing $H_0$ directly is that in $H_0^\bh$ the conditioning r.v. is real, which simplifies the problem substantially. 

\begin{Rem}
The set of directions for which $H_0$ is not congruent with $H_0^\mathbf{h}$ has measure zero. A concrete example of this set is given as follows. Suppose we are interested in testing if a random $p$-vector $\mathbf{X}$ is Gaussian. By the Cr\'amer--Wold device, $\mathbf{X}$ is Gaussian if and only if $\mathbf{a}'\mathbf{X}$ is Gaussian for any $\mathbf{a}\in\mathbb{R}^p$. However, by Theorem 3.6 in \cite{Cuesta-Albertos2007b}, it suffices that $\mathbf{h}'\mathbf{X}$ is Gaussian for a \emph{single}, \emph{randomly chosen} direction $\mathbf{h}\in\mathbb{R}^p$. Then the zero-measure set in which $H_0^\mathbf{h}$ and $H_0$ are incongruent (precisely, $H_0^\mathbf{h}$ holds, but $H_0$ does not) is the set of the \emph{projection counterexamples} $\{\mathbf{a}\in\mathbb{R}^p:\mathbf{a}'\mathbf{X}\text{ is Gaussian}, \mathbf{X}\text{ is not Gaussian}\}$. For example, for $\mathbf{X}\sim(\mathrm{Exp}(1),\,\mathcal{N}(0,1))$, the set is $\{(0,\lambda) : \lambda \in \R\}$. Obviously, if $\bh\in\R^2$ is chosen at random with a nondegenerate measure $\mu$, it is impossible that $\bh$ lies exactly on this line.
\end{Rem}

\subsection{Testing a simple null hypothesis}
\label{subsec:simple}

An immediate application of Corollary \ref{Coro1} is the testing of the simple null hypothesis $H_0: m=m_0$ via the empirical process $I_n$ of \cite{Stute1997}. Recall that other testing alternatives can be considered on the projected covariate due to the $\mu$-a.s. characterization. We refer to \cite{Gonzalez-Manteiga2013} for a review of alternatives. \\ 

For a random sample $\{(\bX_i,Y_i)\}_{i=1}^n$ from $(\bX,Y)$, we consider the empirical process of the regression conditioned on the direction $\bh$, 
\[
R_{n,\bh}(x):=n^{1/2}I_{n,\bh}(x)=n^{-1/2} \sum_{i=1}^n\mathds{1}_{\lrb{\bX_i^\bh\leq x}}Y_i,\quad x\in\R.
\]
Then the following result is trivially satisfied using Theorem 1.1 in \cite{Stute1997}. 
\begin{Coro} 
\label{coro:single} 
Under $H_0^\bh$ and $\E{Y^2}<\infty$, $R_{n,\bh}\inlaw \boldsymbol{\mathcal{G}}_1$ in $D(\R)$, with $\boldsymbol{\mathcal{G}}_1$ a Gaussian process with zero mean and covariance function $K_1(s,t):=\int_{-\infty}^{s\wedge t} \allowbreak\V{Y|\bXh=u}\allowbreak\, \mathrm{d}F_{\bh}(u)$.  
\end{Coro} 
Different statistics for the testing of $H_0^\bh$ can be built from continuous functionals on $R_{n,\bh}(x)$. We shall cover this in more detail in Section \ref{Sec.lin}. 

\begin{Exam}
Consider the FLM $Y=\inprod{\bX}{\brho}+\varepsilon$ in $\Hil=L^2[0,1]$, with $\bX$ a Gaussian process with associated Karhunen--Lo\'eve expansion \eqref{Eq.Kar_loeve} below, and $\varepsilon$ independent from $\bX$. Then $\bX^{\bh}$ and $\bX^{\brho}$ are centred Gaussians with variances $\sigma^2_{\bh}$ and $\sigma^2_{\brho}$, respectively, and $\mathbb{C}\mathrm{ov}[\bX^{\bh},\bX^{\brho}]=\sum_{j=1}^\infty h_j \rho_j\lambda_j$, with $h_j:=\inprod{\bh}{\be_j}$, and $\rho_j:=\inprod{\brho}{\be_j}$. Hence, 
\begin{align*}
K^1(s,t) 
=&\;\int_{-\infty}^{s\wedge t}\Bigg(\frac{\sigma^2_{\brho}\sigma^2_{\bh}-\big(\sum_{j=1}^\infty h_j\rho_j\lambda_j\big)^2}{\sigma^2_{\bh}}+\sigma_\varepsilon^2\Bigg)\phi(u/\sigma_{\bh})/\sigma_{\bh}\,\mathrm{d}u\\
=&\;\Bigg(\frac{\sigma^2_{\brho}\sigma^2_{\bh}-\big(\sum_{j=1}^\infty h_j\rho_j\lambda_j\big)^2}{\sigma^2_{\bh}}+\sigma_\varepsilon^2\Bigg)\Phi((s\wedge t)/\sigma_{\bh}),
\end{align*}
where $\phi$ and $\Phi$ are the density and distribution functions of a $\mathcal{N}(0,1)$, respectively.
\end{Exam}

\section{Testing the functional linear model}
\label{Sec.lin}

We focus now on testing the composite null hypothesis, expressed as 
\begin{align}
H_0:  m(\bX) = \langle \bX, \brho \rangle = \bX^{\brho} \text{ for some } \brho \in \Hil.\label{eq:H0m}
\end{align}
According to Corollary \ref{Coro1}, testing \eqref{eq:H0m} is $\mu$-a.s. equivalent to testing  
\[
H_0^\bh: \E{\left(Y - \bX^{\brho}\right)\big| \bX^\bh} =0 \text{ for some } \brho \in \Hil,
\]
where $\bh$ is sampled from a nondegenerate Gaussian law $\mu$. Again, we construct the associated empirical regression process indexed by the projected covariate following \cite{Stute1997}. Therefore, given an estimate $\hat \brho$ of $\brho$ under $H_0$, we consider
\begin{align}
T_{n,\bh}(x) :=a_n  \sum_{i=1}^n \mathds{1}_{\lrb{\bX_i^\bh \leq x}}\lrp{Y_i - \bX_i^{\hat \brho}}=
a_n\left(
T_{n,\bh}^1(x)+T_{n,\bh}^2(x)+T_{n,\bh}^3(x)\right),\label{Ec.Def.Stat}
\end{align}
where $a_n \conv 0$ is a normalizing positive sequence to be determined later and
\begin{align*}
T_{n,\bh}^1(x):=&\; \sum_{i=1}^n \mathds{1}_{\lrb{\bX_i^\bh \leq x}}\lrp{Y_i - \bX_i^{\brho}},\\
T_{n,\bh}^2(x):=&\;  \sum_{i=1}^n 
\inprod{\mathds{1}_{\lrb{\bX_i^\bh \leq x}}\bX_i
 -
\E{\mathds{1}_{\lrb{\bX^\bh \leq x}} \bX}}{\brho-\hat \brho},
\\
T_{n,\bh}^3(x):=&\;n  \inprod{\E{\mathds{1}_{\lrb{\bX^\bh \leq x}} \bX}}{\brho-\hat \brho}.
\end{align*}
The selection of the right estimator $\hat \brho$ has a crucial role in the weak convergence of $T^3_{n,\bh}$, which requires a substantially more complex proof than that for the simple hypothesis. We consider the regularized estimate proposed in Sections 2 and 3 of \cite{Cardot2007} (subsequently denoted by CMS), whose construction is sketched here for the sake of the exposition of our results. 

\subsection{\texorpdfstring{Construction of the estimator of $\brho$}{Construction of the estimator of rho}}
\label{subsec:rho}

Consider the so-called Karhunen--Lo\'eve expansion of $\bX$:
\begin{align}
\bX = \sum_{j=1}^\infty \lambda_j^{1/2} \xi_j \be_j. \label{Eq.Kar_loeve}
\end{align}
Here, $\lrb{\mathbf{e}_j}_{j=1}^\infty$ is a sequence of orthonormal eigenfunctions associated with the covariance operator of $\bX$, $\Gamma\bz:=\E{(\bX\otimes\bX)(\bz)}$, $\bz\in\Hil$, and the $\xi_j$'s are centred real r.v.'s (because $\bX$ is centred) such that $\E{\xi_j\xi_{j'}}= \delta_{j,j'}$, where $\delta_{j,j'}$ is the Kronecker's delta. The Kronecker operator $\otimes$ is such that $(\bx\otimes\by)\bz=\inprod{\bz}{\bx}\by$ for $\bx,\by,\bz\in\Hil$. We assume that the multiplicity of each eigenvalue is one, so $\lambda_1 > \lambda_2 > \ldots > 0$.\\ 

The functional coefficient $\brho$ is determined by the equation $\Delta = \Gamma \brho$, with $\Delta$ the cross-covariance operator of $\bX$ and $Y$, $\Delta\bz:=\E{(\bX\otimes Y)(\bz)}$, $\bz\in\Hil$. To ensure the existence and uniqueness of a solution to $\Delta = \Gamma \brho$, we require the next basic assumptions: 
\begin{enumerate}[label=\textbf{A\arabic{*}}.,ref=\textbf{A\arabic{*}}]
\item $\bX$ and $Y$ satisfy $\sum_{j=1}^\infty \frac{1}{\lambda_j^2} \langle \E{\bX Y}, \be_j\rangle^2 < \infty$. \label{Assump:B1}
\item The kernel of $\Gamma$ is $\{\mathbf{0}\}$. \label{Assump:B2}
\end{enumerate}

The estimation of $\brho$ requires the inversion of $\Gamma_n:=\frac{1}{n}\sum_{i=1}^{n}\bX_i\otimes\bX_i$, but, since $\Gamma_n$ is a.s. a finite rank operator, its inverse does not exist. CMS proposed a regularization yielding a family of continuous estimators for $\Gamma^{-1}$. Based on their Example 1, we define $\Gamma_n^\dag$, an empirical finite rank estimate of~$\Gamma^{-1}$: 
\[
\Gamma_n^\dag:=\sum_{i=1}^{k_n}\frac{1}{\hat{\lambda}_j}\hat{\be}_{j}\otimes\hat{\be}_{j}.
\]
The construction of $\Gamma_n^\dag$ (resp., the population version $\Gamma^\dag:=\sum_{i=1}^{k_n}\frac{1}{\lambda_j}\be_{j}\otimes\be_{j}$) is done by considering a sequence of thresholds $c_n \in (0,\lambda_1)$, $n \in \mathbb{N}$, with $c_n\conv 0$. The procedure is as follows: (\textit{i}) compute the Functional Principal Components (FPC) of $\bX_1,\ldots,\bX_n$, that is, calculate the eigenvalues $\{\hat \lambda_ j\}$ and eigenfunctions $\{\hat \be_j\}$ of $\Gamma_n$; (\textit{ii}) define the sequences $\lrb{\delta_j}$, with $\delta_1 :=\lambda_1- \lambda_2$ and $\delta_j := \min(\lambda_j - \lambda_{j+1}, \lambda_{j-1} - \lambda_j)$ for $j>1$, and set
\[
k_n:=\sup \{j\in\N:\lambda_j+\delta_j/2\geq c_n \};
\]
(\textit{iii}) compute $\Gamma_n^\dag$ (resp., $\Gamma^\dag$) as the finite rank operator with the same eigenfunctions as $\Gamma_n$ (resp., $\Gamma$) and associated eigenvalues equal to $\hat{\lambda}_j^{-1}$ (resp., $\lambda_j^{-1}$) if $j\leq k_n$ and $0$ otherwise. The regularized estimator of $\brho$ is then
\begin{align} \label{Eq.EstimRho}
\hat{\brho}:=\Gamma_n^\dagger\Delta_n=\frac{1}{n}\sum_{j=1}^{k_n}\sum_{i=1}^n\frac{\inprod{\bX_i\otimes Y_i}{\hat\be_j}}{\hat \lambda_j}\hat{\be_j}.
\end{align}
Note that \eqref{Eq.EstimRho} is not readily computable in practice, since $\{\lambda_j\}$ is usually unknown (and hence, $k_n$). As in CMS, we consider the (random) finite rank 
\[
d_n:=\sup \{j\in\N: \hat \lambda_ j \geq c_n \}
\]
as a replacement in practice for the deterministic $k_n$. As seen in Lemma \ref{Prop.Tn}, $\nu[k_n=d_n]\conv1$. Hence the estimator \eqref{Eq.EstimRho} has the same asymptotic behaviour with either $k_n$ or $d_n$. Therefore, we consider $k_n$ in \eqref{Eq.EstimRho} due to the enhanced probabilistic tractability. The consideration of $k_n$ in $\Gamma_n^\dag$, instead of $d_n$, is the main difference between our definition of $\Gamma_n^\dag$ and the proposal given from Example 1 in~CMS. \\

The following assumptions allow us to obtain meaningful asymptotic convergences involving $\hat{\brho}$:
\begin{enumerate}[label=\textbf{A\arabic{*}}.,ref=\textbf{A\arabic{*}}] 
\setcounter{enumi}{2}
\item $\E{\| \bX \| ^2} < \infty$.\label{Assump:C1}

\item $ \sum_{l=1}^\infty  | \langle \brho,\be_l\rangle | < \infty$.\label{Assump:C2}

\item For $j$ large, $\lambda_j=\lambda(j)$, with $\lambda(\cdot)$ a convex positive function.\label{Assump:C3} 
 
\item $\frac{\lambda_{n} n^4}{ \log n} = \Oh(1)$.\label{Assump:C4}

\item $\inf \left\{ |\langle \brho, \be_{k_n}\rangle | , \frac {\lambda_{k_n}} {\sqrt{k_n \log k_n}}\right\}= \Oh(n^{-1/2})$.\label{Assump:C5}
             
\item $\sup_j \left\{ \max \big(\mathbb{E}\big[\xi_j^4\big], \mathbb{E}\big[|\xi_j|^5\big]\big) \right\}\leq M < \infty$ for $M\geq1$.\label{Assump:C6} 

\item There exist $C_1,C_2>0$ such that $C_1 n^{-1/2} < c_n < C_2n^{-1/2}$ for every $n$.\label{Assump:C7}
\end{enumerate}

A brief summary of these assumptions is given as follows. \ref{Assump:C1} is standard for obtaining asymptotic distributions, allows decomposition \eqref{Eq.Kar_loeve}, and implies $\E{Y^2}<\infty$, which is required in Theorem 1.1 of \cite{Stute1997}. \ref{Assump:C2} and \ref{Assump:C3} are A.1 and A.2 in CMS. \ref{Assump:C4} is very similar to an assumption in the second part of Theorem 2 in CMS. \ref{Assump:C5} is the minimum requirement for controlling $\langle \bX, \bL_n\rangle$ (to be detailed in Section \ref{ap:linmod}) when Lemma 7 in CMS is used to prove Lemma \ref{Lemm.Tn3}. \ref{Assump:C6} is a reinforcement of A.3 in CMS, where only fourth-order moments are used. This is because we handle inner products of $\hat \brho$ times a nonindependent r.v., while in CMS the r.v. is not used to estimate $\brho$. \ref{Assump:C7} is useful, mainly (but also see the final part of Lemma \ref{LemmTn2.Yn}) to control the behaviour of $k_n$. We show this fact in Proposition \ref{Prop.1}, with a conclusion very close to assumption (8) in CMS and coinciding with one of the conditions of their Theorem 3 if $\lim_n t_{n,\mathbf{E}_{x,\bh}} < \infty$ (the term $t_{n,\mathbf{E}_{x,\bh}}$ is defined in \eqref{Eq.definiciones} below). Finally, we point out that in CMS the assumptions aim to control the behaviour of $k_n$ while here we have sought to control the threshold $c_n$, as this can be modified by the statistician. 

\subsection{\texorpdfstring{Pointwise asymptotic distribution of $T_{n,\bh}$}{Pointwise asymptotic distribution of Tnh}}
\label{subsec:point}

Corollary \ref{coro:single} shows the weak convergence of $n^{-1/2}T_{n,\bh}^1$. We analyse now the pointwise behaviour of $T_{n,\bh}^2(x)$ and $T_{n,\bh}^3(x)$ for a fixed $x\in\R$. We will show that $T_{n,\bh}^2(x)=\op(n^{1/2})$ and that the rate of $T_{n,\bh}^3(x)$ depends on the key normalizing sequence $\{t_{n,\mathbf{E}_{x,\bh}}\}$, where  
\begin{align} \label{Eq.definiciones}
t_{n,\bx} := \sqrt{\sum_{j=1}^{k_n} \frac{\langle \bx, \be_j\rangle^2}{\lambda_j}}\quad\text{and}\quad \mathbf{E}_{x,\bh}:=\E{\mathds{1}_{\lrb{\bX^\bh \leq x}} \bX}.
\end{align}

\begin{Theo} \label{Theo:DistPuntual}
Under $H_0^\bh$ and \ref{Assump:B1}--\ref{Assump:C7}, and for a fixed $x\in\R$, it follows that:
\begin{enumerate}[label=(\alph{*}),ref=(\alph{*})]
\item $n^{-1/2} t_{n,\mathbf{E}_{x,\bh}}^{-1} T_{n,\bh}^3(x) \inlaw  \mathcal{N}(0,  \sigma^2_{\varepsilon})$.\label{Theo:DistPuntual:a}

\item If $\lim_n t_{n,\mathbf{E}_{x,\bh}}= \infty$, then with $a_n = n^{-1/2} t_{n,\mathbf{E}_{x,\bh}}^{-1}$ in \eqref{Ec.Def.Stat}, the asymptotic distribution of $T_{n,\bh} (x)$ is $n^{-1/2} t_{n,\mathbf{E}_{x,\bh}}^{-1}T_{n,\bh}^3(x)$. \label{Theo:DistPuntual:b} 

\item If $\lim_n t_{n,\mathbf{E}_{x,\bh}} < \infty$, then with $a_n = n^{-1/2}$ in \eqref{Ec.Def.Stat}, the asymptotic distribution of $T_{n,\bh} (x)$ is $n^{-1/2}\big(T_{n,\bh}^1(x) + T_{n,\bh}^3(x)\big)$.  \label{Theo:DistPuntual:c} 
\end{enumerate} 
\end{Theo}

The behaviour of the sequence $\{t_{n,\mathbf{E}_{x,\bh}}\}$, indexed by $n\in\N$ and with arbitrary $\bh\in\Hil$ and $x \in\R$, is crucial for the convergence of $T_{n,\bh}$. Since $\{t_{n,\mathbf{E}_{x,\bh}}\}$ is nondecreasing, it has always a limit (finite or infinite). Its asymptotic behaviour is described next.
 
\begin{Prop}
\label{Prop.tnx}
The sequence $\{t_{n,\mathbf{E}_{x,\bh}}\}$ has asymptotic order between $\Oh(1)$ and $\Oh\big(k_n^{1/2}\big)$. In addition, if $\bX$ is Gaussian and satisfies \ref{Assump:C1}, then $\sigma_{\bh}^2:=\V{\bX^{\bh}}< \infty$ and $\lim_n t_{n,\mathbf{E}_{x,\bh}}=\phi(x/\sigma_{\bh})$. 
\end{Prop}

\subsection{\texorpdfstring{Weak convergence of $T_{n,\bh}$ and the test statistics}{Weak convergence of Tnh and the test statistics}} 
\label{subsec:weak}

The result given in Theorem \ref{Theo:DistPuntual} holds for every $x\in\R$. For case \ref{Theo:DistPuntual:c} of Theorem \ref{Theo:DistPuntual} (where the estimation of $\brho$ is not dominant) and under an additional assumption, the result can be generalized to functional weak convergence. 
\begin{Theo} \label{Theo:tighness}
Under $H_0^\bh$, \ref{Assump:B1}--\ref{Assump:C7}, and \ref{Theo:DistPuntual:c} in Theorem \ref{Theo:DistPuntual}, it follows that:
\begin{enumerate}[label=(\alph{*}),ref=(\alph{*})]
\item The finite dimensional distributions of $T_{n,\bh}$ converge to a multivariate Gaussian with covariance function\label{Theo:tighness:a} $K_2(s,t):=K_1(s,t)+C(s,t)+C(t,s)+V(s,t)$, where
\begin{align*}
C(s,t):=&\;\int_{\{\bu^\bh\leq s\}} \V{Y|\bX=\bu}\inprod{\mathbf{E}_{t,\bh}}{\Gamma^{-1}\bu}\,\mathrm{d}P_\bX(\bu),\\
V(s,t):=&\;\int \V{Y|\bX=\bu}\inprod{\mathbf{E}_{s,\bh}}{\Gamma^{-1}\bu}\inprod{\mathbf{E}_{t,\bh}}{\Gamma^{-1}\bu}\,\mathrm{d}P_\bX(\bu).
\end{align*}
\item If $\mathbb{E}\big[\norm{\hat{\brho}-\brho}^4\big]=\Oh(n^{-2})$, then
$T_{n,\bh}\inlaw\boldsymbol{\mathcal{G}}_2$ in $D(\R)$, with $\boldsymbol{\mathcal{G}}_2$ a Gaussian process with zero mean and covariance function $K_2$. \label{Theo:tighness:b}
\end{enumerate}
\end{Theo}

\begin{Rem}
According to Theorem 1 in CMS, it is impossible for $\hat \brho - \brho$ to converge to a nondegenerate random element in the topology of \Hil. To circumvent this issue and obtain the tightness of $T_{n,\bh}$, we assume $\mathbb{E}\big[\norm{\hat{\brho}-\brho}^4\big]=\Oh(n^{-2})$, which implies $\norm{\hat{\brho}-\brho}=\OP(n^{-1/2})$, and, thus, a finite-dimensional parametric convergence rate for $\hat\brho$. For instance, this happens when $\brho$ is a linear combination of a finite number of the eigenfunctions of $\Gamma$. Notice that this is not needed for the convergence of the finite dimensional distributions of $T_{n,\bh}$. 
\end{Rem}

The next result gives the convergence of the Kolmogorov--Smirnov (KS) and Cram\'er--von Mises (CvM) statistics for testing the FLM.
\begin{Coro} \label{Coro:KSCvM}
Under the assumptions in Theorem \ref{Theo:tighness} and  $\mathbb{E}\big[\norm{\hat{\brho}-\brho}^4\big]=\Oh(n^{-2})$, if $\|T_{n,\bh}\|_\mathrm{KS}:=\sup_{x\in\R}|T_{n,\bh}(x)|$ and $\|T_{n,\bh}\|_\mathrm{CvM}:=\int_\R T_{n,\bh}(x)^2\,\mathrm{d}F_{n,\bh}(x)$, then
\begin{align*}
\|T_{n,\bh}\|_\mathrm{KS}\inlaw\|\boldsymbol{\mathcal{G}}_2\|_\mathrm{KS}\text{ and }\|T_{n,\bh}\|_\mathrm{CvM}\inlaw\int_\R \boldsymbol{\mathcal{G}}_2(x)^2\,\mathrm{d}F_{\bh}(x).
\end{align*}
\end{Coro}

\begin{Rem}
An alternative to \ref{Theo:tighness:b} and Corollary \ref{Coro:KSCvM} is to consider a deterministic discretization of the statistics, for which the convergence in law is trivial from \ref{Theo:tighness:a}. For example, if $\|T_{n,\bh}\|_{\widetilde{\mathrm{KS}}}:=\max_{k=1,\ldots,G}|T_{n,\bh}(x_k)|$ for a grid $\{x_1,\ldots,x_G\}$, then
$\|T_{n,\bh}\|_{\widetilde{\mathrm{KS}}}\inlaw\|\mathbf{Z}_2\|_{\widetilde{\mathrm{KS}}}$, where $\mathbf{Z}_2\sim\mathcal{N}_G(\mathbf{0},\boldsymbol\Sigma)$, $\boldsymbol\Sigma_{ij}=K_2(x_i,x_j)$.
\end{Rem}

\section{Testing in practice}
\label{sec:testing}

The major advantage of testing $H_0^\bh$ over $H_0$ is that in $H_0^\bh$ the conditioning r.v. is real. The potential drawbacks of this universal method are a possible loss of power and that the outcome of the test may vary for different projections. Both inconveniences can be alleviated by sampling several directions $\bh_1,\ldots, \bh_K$, testing the projected hypotheses $H_0^{\bh_1},\ldots,H_0^{\bh_K}$, and selecting an appropriate way to mix the resulting $p$-values. For example, using the FDR method proposed in \cite{Benjamini2001} (see Section 2.2.2 of \cite{Cuesta-Albertos2010}), it is possible to control the final rejection rate to be \textit{at most} $\alpha$ under \nolinebreak[4]$H_0$. \\ 

The drawing of random directions is clearly influential in the power of the test. For example, in the extreme case where the directions are orthogonal to the data, that is, $\bX^\bh= 0$, then $T_{n,\bh}(x)=(n^{-1/2}\sum_{i=1}^n\hat{\varepsilon}_i)\mathds{1}_{\{0\leq x\}}$ and $\|T_{n,\bh}\|_\mathrm{N}=\|T^{*b}_{n,\bh}\|_\mathrm{N}=0$ under $H_0$. Therefore, Algorithm \ref{algo:boot} would fail to calibrate the level of the test and potentially yield spurious results due to numerical inaccuracies in $\|T^{*b}_{n,\bh}\|_\mathrm{N}\leq\|T_{n,\bh}\|_\mathrm{N}$. A data-driven compromise to avoid drawing directions in subspaces \textit{almost} orthogonal to the data is the following: (\textit{i}) compute the FPC of $\bX_1,\ldots,\bX_n$, that is, the eigenpairs $\{(\hat \lambda_ j, \hat\be_j)\}$; (\textit{ii}) choose $j_n:=\min\big\{k=1,\ldots,n-1:(\sum_{j=1}^{k}\hat\lambda_j^2)/(\sum_{j=1}^{n-1}\hat\lambda_j^2)\geq r\big\}$ for a variance threshold $r$, for example, $r=0.95$; (\textit{iii}) generate the data-driven Gaussian process $\bh_{j_n}:=\sum_{j=1}^{j_n}\eta_j\hat\be_j$, with $\eta_j\sim\mathcal{N}(0,s_j^2)$ and $s_j^2$ the sample variance of the scores in the $j$th FPC. Without loss of generality, we will use this data-driven projecting process for drawing $\bh$ in the rest of the paper (see the supplement for the consideration of other data generating processes). Formally, the Gaussian measure $\mu$ associated with $\bh_{j_n}$ does not respect the assumptions in Theorem \ref{Th:basic}, since it is degenerate (but recall that $\mu$ does not have to be independent from $\bX$). A nondegenerate Gaussian process can be obtained as $\bh_{j_n}+\boldsymbol{\mathcal{G}}$, with $\boldsymbol{\mathcal{G}}$ a Gaussian process tightly concentrated around zero, albeit employing $\bh_{j_n}$ or $\bh_{j_n}+\boldsymbol{\mathcal{G}}$ has negligible effects in practice. \\ 

The testing procedure is described in the following generic algorithm. 

\begin{Algo}[Testing procedure for $H_0$]
\label{algo:general}
Let $T_n$ denote a test for checking $H_0^{\bh}$ with $\bh$ chosen by a nondegenerate Gaussian measure $\mu$ on $\Hil$:
\begin{enumerate}[label=(\roman{*}), ref=(\roman{*})]
\item For $i=1,\ldots,K$, set by $p_i$ the $p$-value of $H_0^{\bh_i}$ obtained with the test $T_n$.
\item Set the final $p$-value of $H_0$ as $\min_{i=1,\ldots,K}\frac{K}{i}p_{(i)}$, where $p_{(1)}\leq\ldots\leq p_{(K)}$.
\end{enumerate}
\end{Algo}
The calibration of the test statistic for $H_0^{\bh}$ is done by a wild bootstrap resampling. The next algorithm states the steps for testing the FLM. The particular case of the simple null hypothesis corresponds to $\brho=\mathbf{0}$, so its calibration corresponds to setting $\hat\brho=\hat\brho^*=\mathbf{0}$ in the algorithm.

\begin{Algo}[Bootstrap calibration in FLM testing]
\label{algo:boot}
Let $\lrb{\lrp{\bX_i,Y_i}}_{i=1}^n$ be an i.i.d. sample from \eqref{funcmod} and a given $\bh\in\Hil$. To test \eqref{eq:H0m}, proceed as follows: 
	\begin{enumerate}[label=(\roman{*}), ref=(\roman{*})]
 		\item Estimate $\brho$ by FPC for a given $d_n$ and obtain $\hat\varepsilon_i=Y_i-\inprod{\bX_i}{\hat \brho}$.\label{algo:boot:i}
		\item Compute $\|T_{n,\bh}\|_\mathrm{N}=\left|\left|n^{-1/2} \sum_{i=1}^n \mathds{1}_{\lrb{\bX_i^\bh \leq x}}\hat\varepsilon_i\right|\right|_\mathrm{N}$ with $\mathrm{N}$ either $\mathrm{KS}$ or $\mathrm{CvM}$.\label{algo:boot:ii}
		\item \textit{Bootstrap resampling}. For $b=1,\ldots,B$:\label{algo:boot:iii}
		\begin{enumerate}[label=(\alph{*}), ref=(\alph{*})]
			\item Draw binary i.i.d. r.v.'s $V_1^*,\ldots,V_n^*$ such that $\prob{V^*=(1-\sqrt{5})/2}\allowbreak=(5+\sqrt{5})/10$ and $\prob{V^*=(1+\sqrt{5})/2}=(5-\sqrt{5})/10$.\label{algo:boot:iii:a}
			
			\item Set $Y_i^*:=\inprod{\bX_i}{\hat\brho}+\varepsilon_i^*$ using the bootstrap residuals $\varepsilon^*_i:=V_i^*\hat\varepsilon_i$.\label{algo:boot:iii:b}
			
			\item Estimate $\brho^{*}$ from $\lrb{\lrp{\bX_i,Y_i^*}}_{i=1}^n$ by FPC using the same $d_n$ used in \ref{algo:boot:i}. \label{algo:boot:iii:c} 

			\item Obtain the estimated bootstrap residuals $\hat\varepsilon_i^*:=Y_i^*-\inprod{\bX_i}{\hat\brho^{*}}$.\label{algo:boot:iii:d}
			
			\item Compute $\|T_{n,\bh}^{*b}\|_\mathrm{N}:=\left|\left|n^{-1/2} \sum_{i=1}^n \mathds{1}_{\lrb{\bX_i^\bh \leq x}}\hat\varepsilon_i^*\right|\right|_\mathrm{N}$.\label{algo:boot:iii:e}

		\end{enumerate}
		\item Approximate the $p$-value by $\frac{1}{B}\sum_{b=1}^B\mathds{1}_{\lrb{\|T_{n,\bh}^{*b}\|_\mathrm{N}\leq \|T_{n,\bh}\|_\mathrm{N}}}$.\label{algo:boot:iv}
	\end{enumerate}

\end{Algo}

Notice that the role of $c_n$ is the selection of $d_n$ in the estimation of $\brho$. The selection of $d_n$ can be done in a data-driven way by selecting from among a set of candidate $d_n$'s the optimal one in terms of a model-selection criterion. For example, we consider the corrected Schwartz Information Criterion \citep{McQuarrie1999}, defined as $\mathrm{SICc}(d_n):=\ell(\hat{\brho}_{d_n})+\frac{\log(n) d_n}{n-d_n-2}$, in order to overpenalize large $d_n$'s that generate noisy estimates of $\brho$, especially for low sample sizes. In the previous expression, $\ell(\hat{\brho}_{d_n})$ represents the log-likelihood of the FLM for $\brho$ estimated with $d_n$ FPC's. Of course, this selection could be done in terms of the $c_n$'s that determine the $d_n$'s but, since the latter are directly related to the model complexity, its analysis is more convenient in practice. Note that steps \ref{algo:boot:iii:c} and \ref{algo:boot:iii:d} can be easily computed using the properties of the linear model; see Section 3.3 of \cite{Garcia-Portugues:flm}. \\ 

The bootstrap process we are considering is given by (we consider $a_n=n^{-1/2}$)
\begin{align*}
T_{n,\bh}^{*}\lrp{x}:=n^{-1/2} \sum_{i=1}^n \mathds{1}_{\lrb{\bX_i^\bh \leq x}}\hat\varepsilon_i^*=n^{-1/2} \sum_{i=1}^n \mathds{1}_{\lrb{\bX_i^\bh \leq x}}\hat\varepsilon_iV_i^*+n^{-1/2} \sum_{i=1}^n \mathds{1}_{\lrb{\bX_i^\bh \leq x}}\bX_i^{\hat{\brho}-\hat\brho^*},
\end{align*}
which is estimating the distribution of
\[
T_{n,\bh}(x)=n^{-1/2} \sum_{i=1}^n \mathds{1}_{\lrb{\bX_i^\bh \leq x}}\hat\varepsilon_i+n^{-1/2} \sum_{i=1}^n \mathds{1}_{\lrb{\bX_i^\bh \leq x}}\bX_i^{\brho-\hat{\brho}}.
\]
The bootstrap consistency could be obtained as an adaptation of Lemma A.1 of \cite{Stute1998} for the first term of $T_{n,\bh}^*$, Lemma A.2 \textit{ibid} for the second term, and using the decomposition of $\hat{\brho}-\brho$ given in (11) in CMS.

\section{Simulation study and data application}
\label{sec:simu}

We illustrate the finite sample performance of the CvM and KS goodness-of-fit tests implemented using Algorithms \ref{algo:general} and \ref{algo:boot} for the composite hypothesis. In order to examine the possible effects of different functional coefficients $\brho$ and underlying processes for $\bX$, we considered nine possible scenarios combining both factors. The detailed description of these scenarios is given in the supplement, while a coarse-grained graphical idea can be obtained from Figure \ref{fig:sce}. \\

\begin{figure}[t!]
\centering
\includegraphics[width=\textwidth]{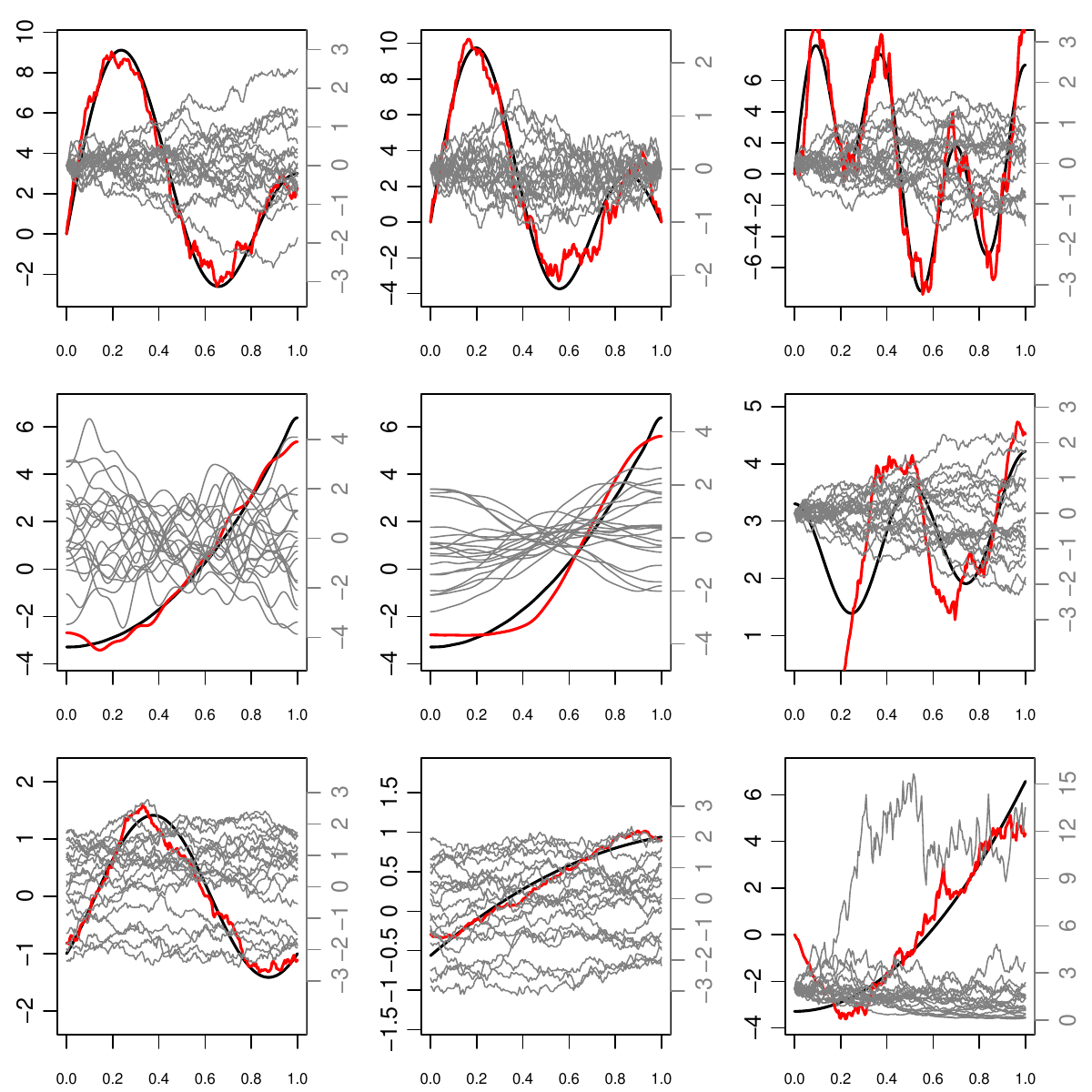}
\caption{\small From left to right and up to down, functional coefficients $\brho$ (black, right scale) and underlying processes (grey, left scale) for the nine different scenarios, labelled S1 to S9. Each graph shows a sample of $100$ realizations of the functional covariate $\bX$ and the estimate $\hat\brho$ (red) with $d_n$ selected by SICc.\label{fig:sce}} 
\end{figure}

The different data generating processes are encoded as follows. For the $k$th simulation scenario S$k$, with functional coefficient $\brho_k$, the deviation from $H_0$ is measured by a deviation coefficient $\delta_d$, with $\delta_0=0$ and $\delta_d>0$ for $d=1,2$. Then, with $H_{k,d}$ we denote the data generation from
\[
Y=\inprod{\bX}{\brho_k}+\delta_d \Delta_{\eta_k}(\bX)+\varepsilon,
\]
where $\boldsymbol{\eta}:=(1,2,1,2,2,1,2,3,3)'$ and the deviations from the linear model are constructed by including the nonlinear terms $\Delta_{1}(\bX):=\norm{\bX}$, $\Delta_{2}(\bX):=25\int_0^1\int_0^1\sin(2\pi ts)s(1-s)t(1-t)\bX(s)\bX(t)\,\mathrm{d}s\,\mathrm{d}t$, and $\Delta_{3}(\bX):=\inprod{e^{-\bX}}{\bX^2}$. The error $\varepsilon$ is distributed as a $\mathcal{N}(0,\sigma^2)$, where $\sigma^2$ was chosen such that, under $H_0$, $R^2=\frac{\V{\inprod{\bX}{\brho}}}{\V{\inprod{\bX}{\brho}}+\sigma^2}=0.95$. The selection of $d_n$ is done automatically by SICc throughout the section. The random directions are drawn from the data-driven Gaussian process described in Section \ref{sec:testing} (see the supplement for other data generating processes and their effects). The choice of the $\delta_d$'s is described in detail in the supplement. \\

\begin{figure}[b!]
\centering
\includegraphics[width=0.315\textwidth]{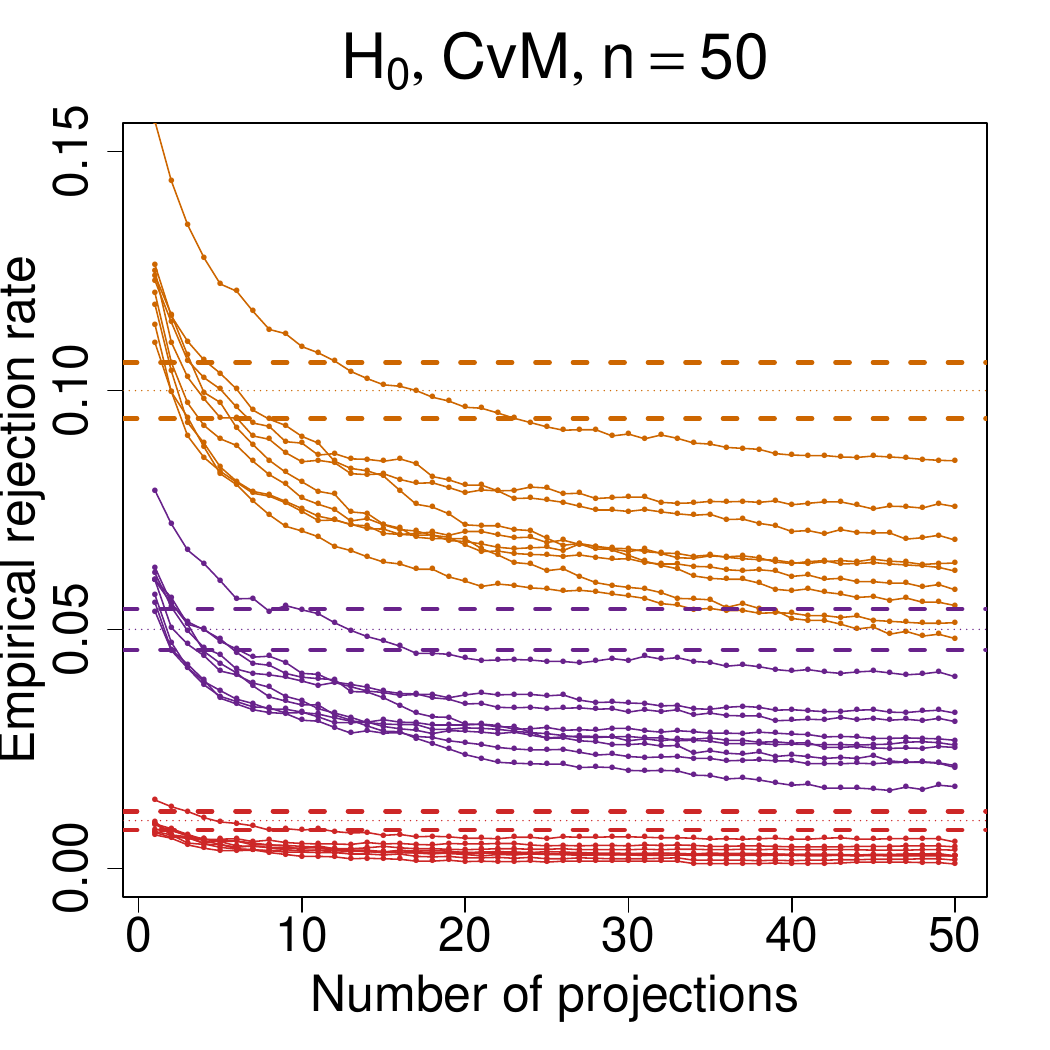}
\includegraphics[width=0.315\textwidth]{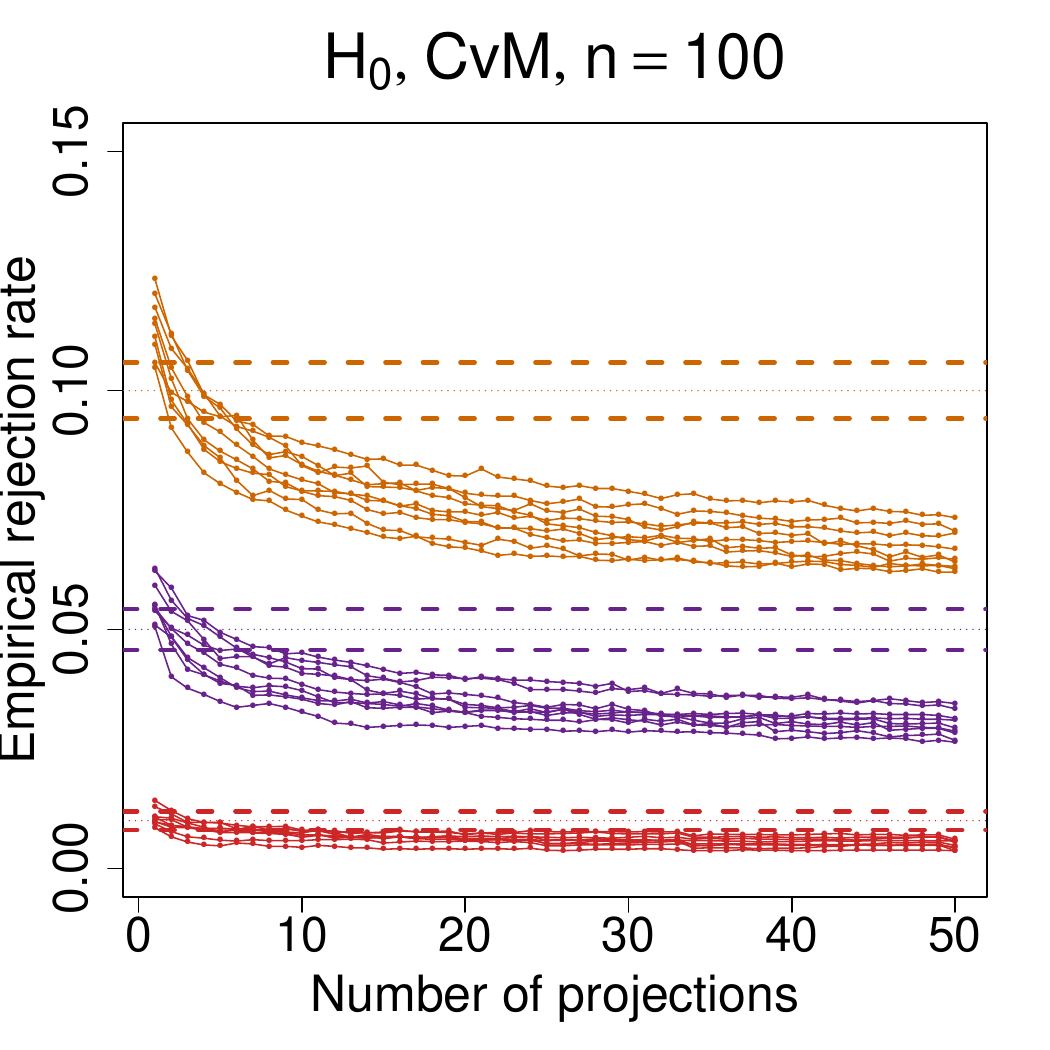}
\includegraphics[width=0.315\textwidth]{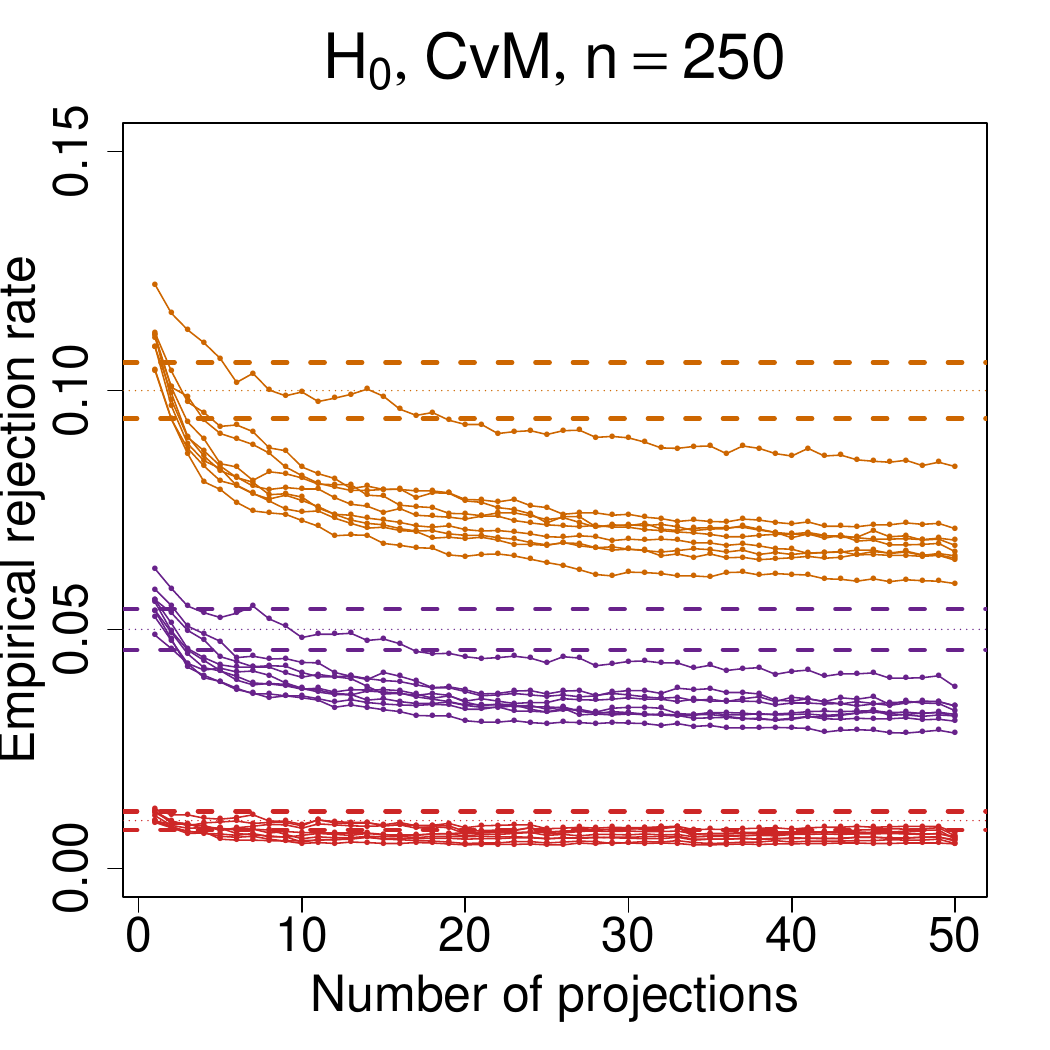}\\
\includegraphics[width=0.315\textwidth]{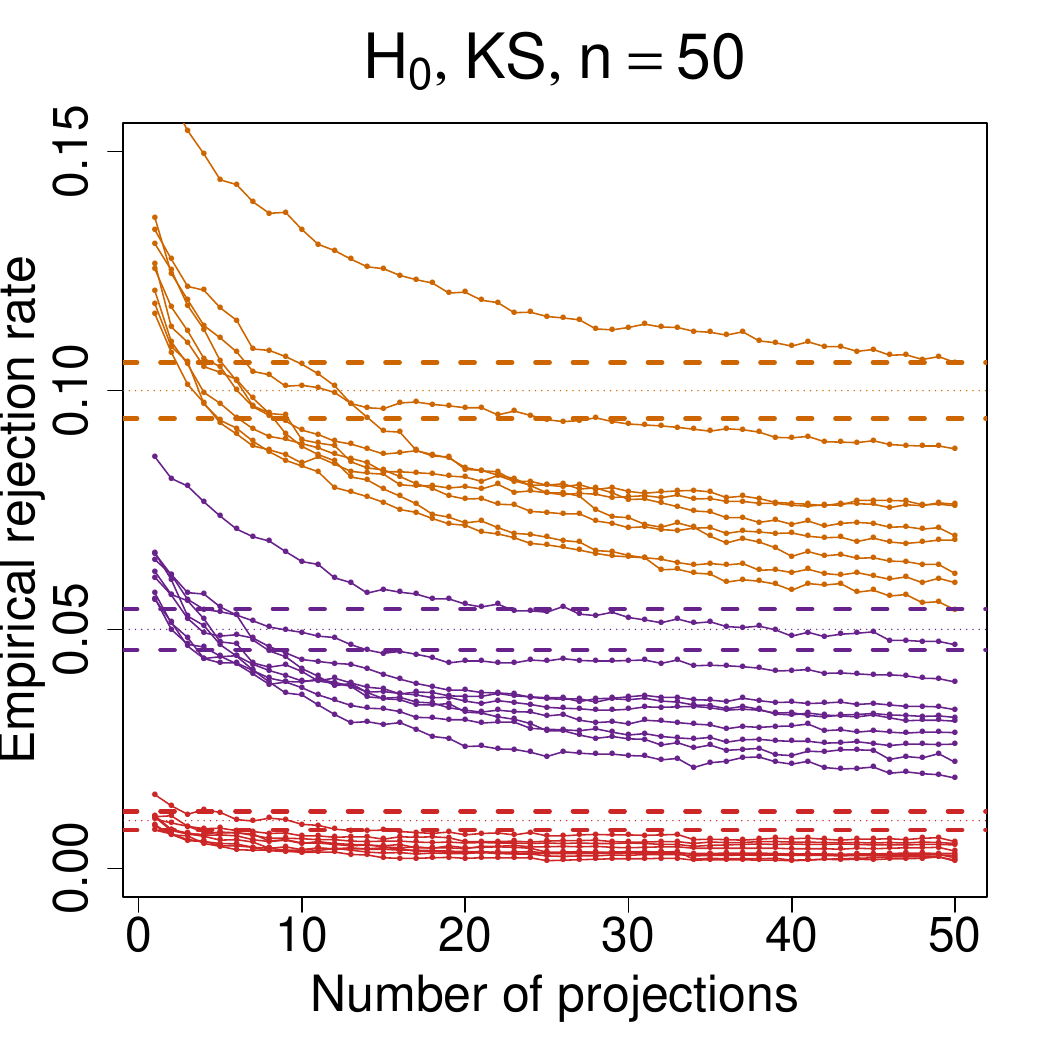}
\includegraphics[width=0.315\textwidth]{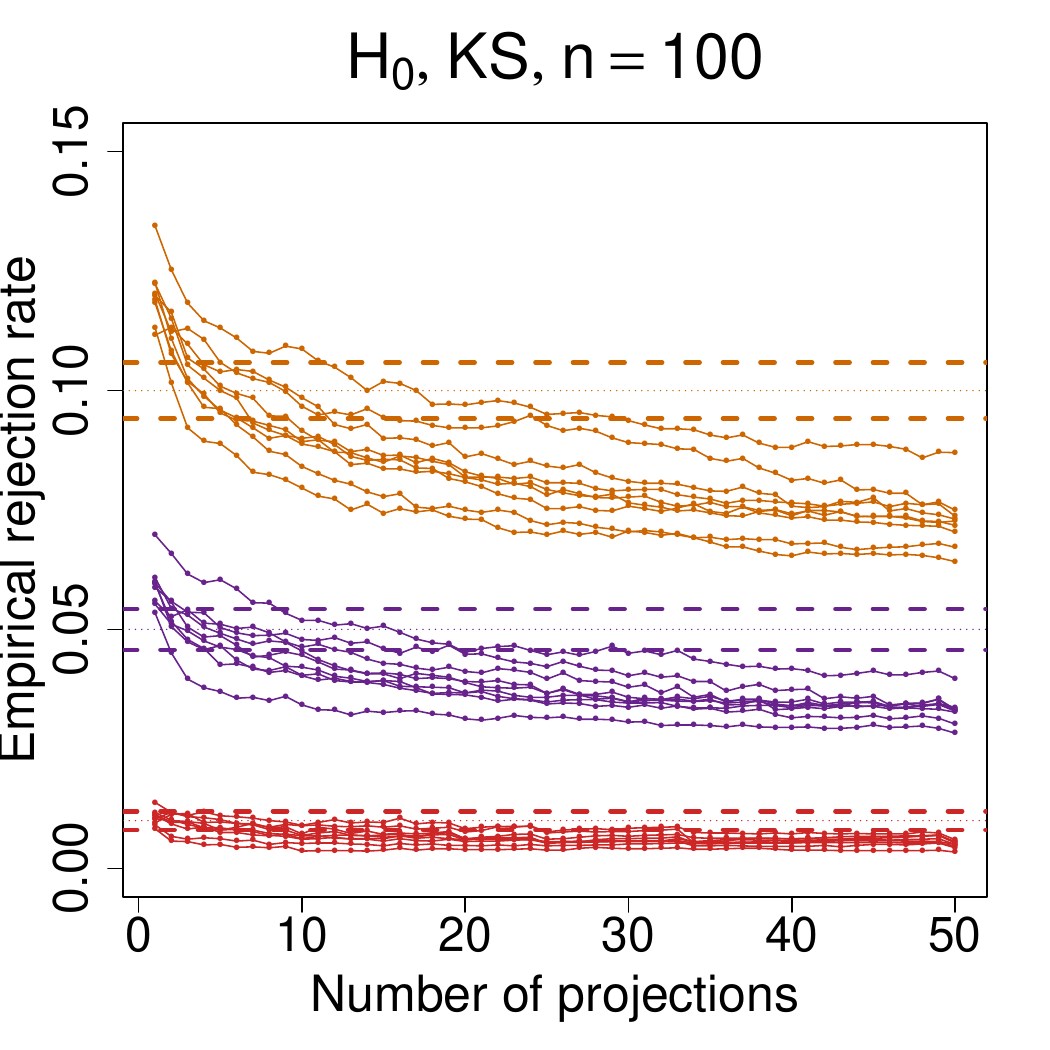}
\includegraphics[width=0.315\textwidth]{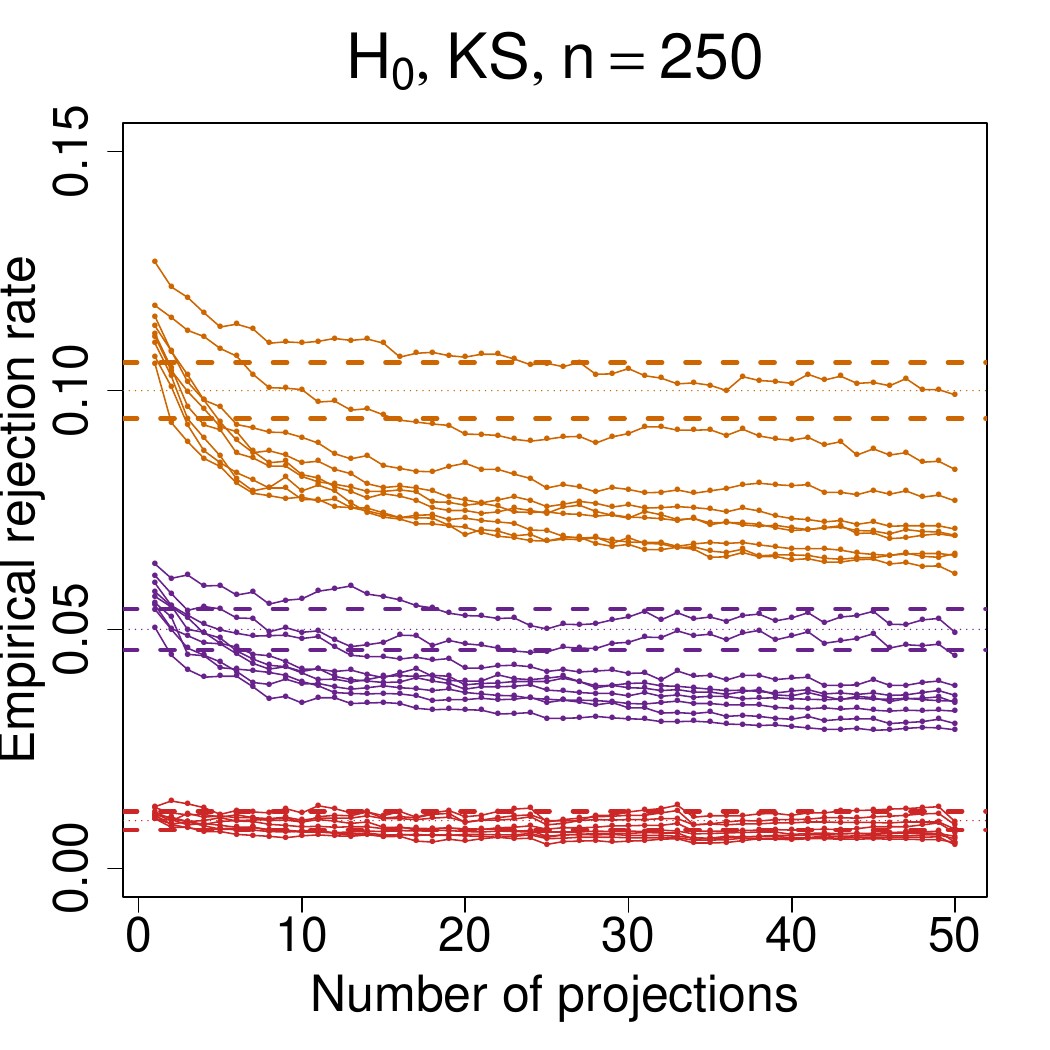}
\caption{\small Empirical sizes of the CvM (upper row) and KS (lower row) tests for scenario S$k$, $k=1,\ldots,9$, depending on the number of projections $K=1,\ldots,50$, and for sample sizes $n=50,100,250$ (from left to right). The empirical sizes associated with the significance levels $\alpha=0.01,0.05,0.10$ are coded in red, purple, and orange, respectively. Dashed thick lines represent the asymptotic $95 \%$ confidence interval for the proportion $\alpha$ obtained from $M$ replicates. \label{fig:sizeproj}}
\end{figure}

\begin{figure}[htb!]
\centering
\includegraphics[width=0.315\textwidth]{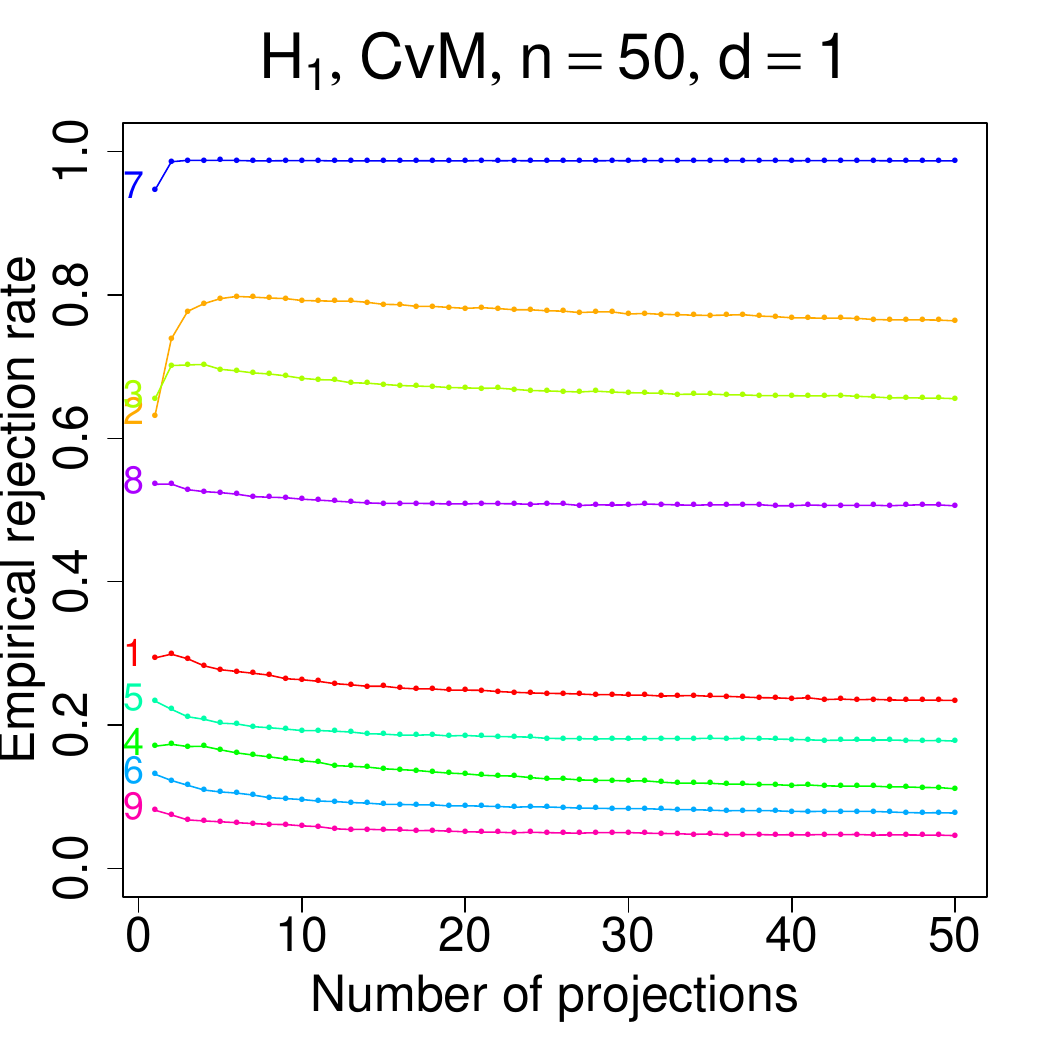}
\includegraphics[width=0.315\textwidth]{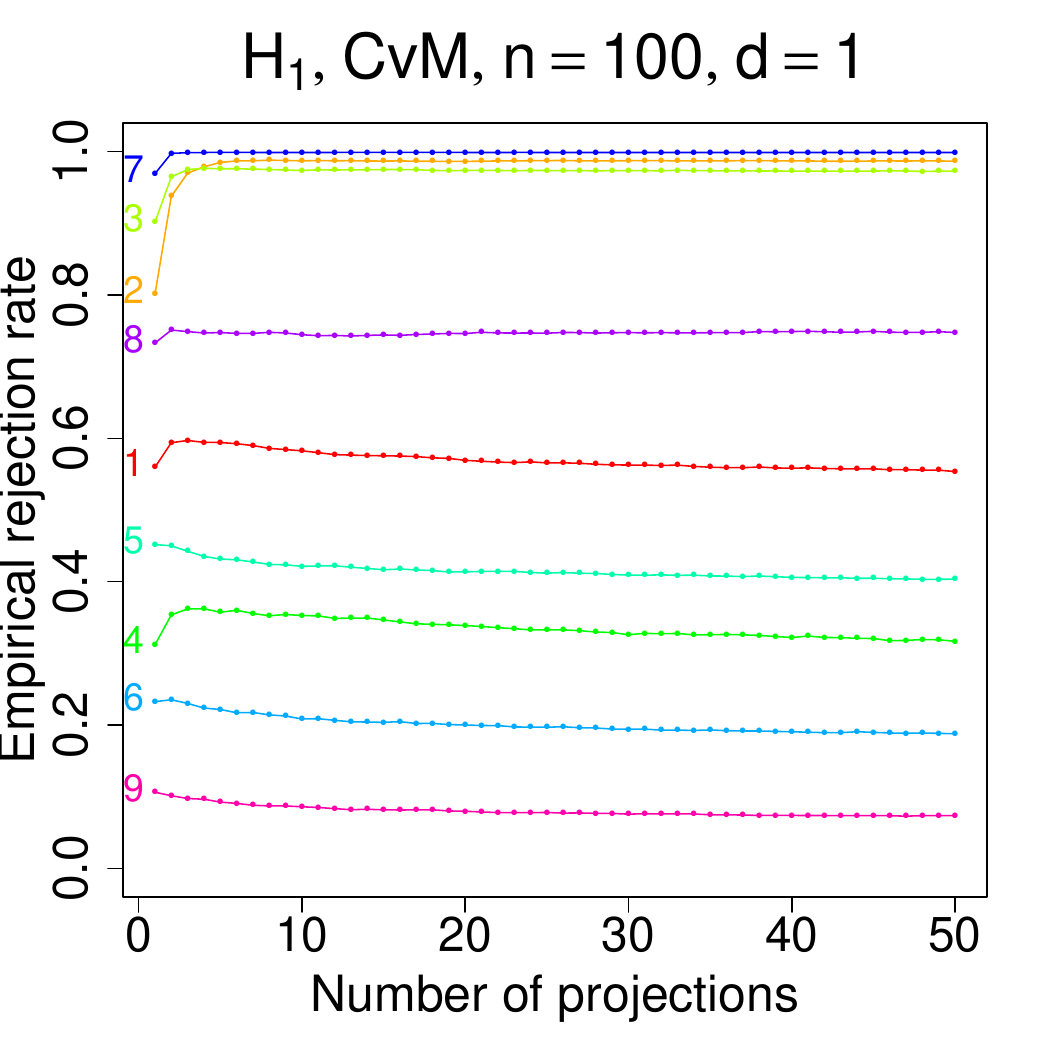}
\includegraphics[width=0.315\textwidth]{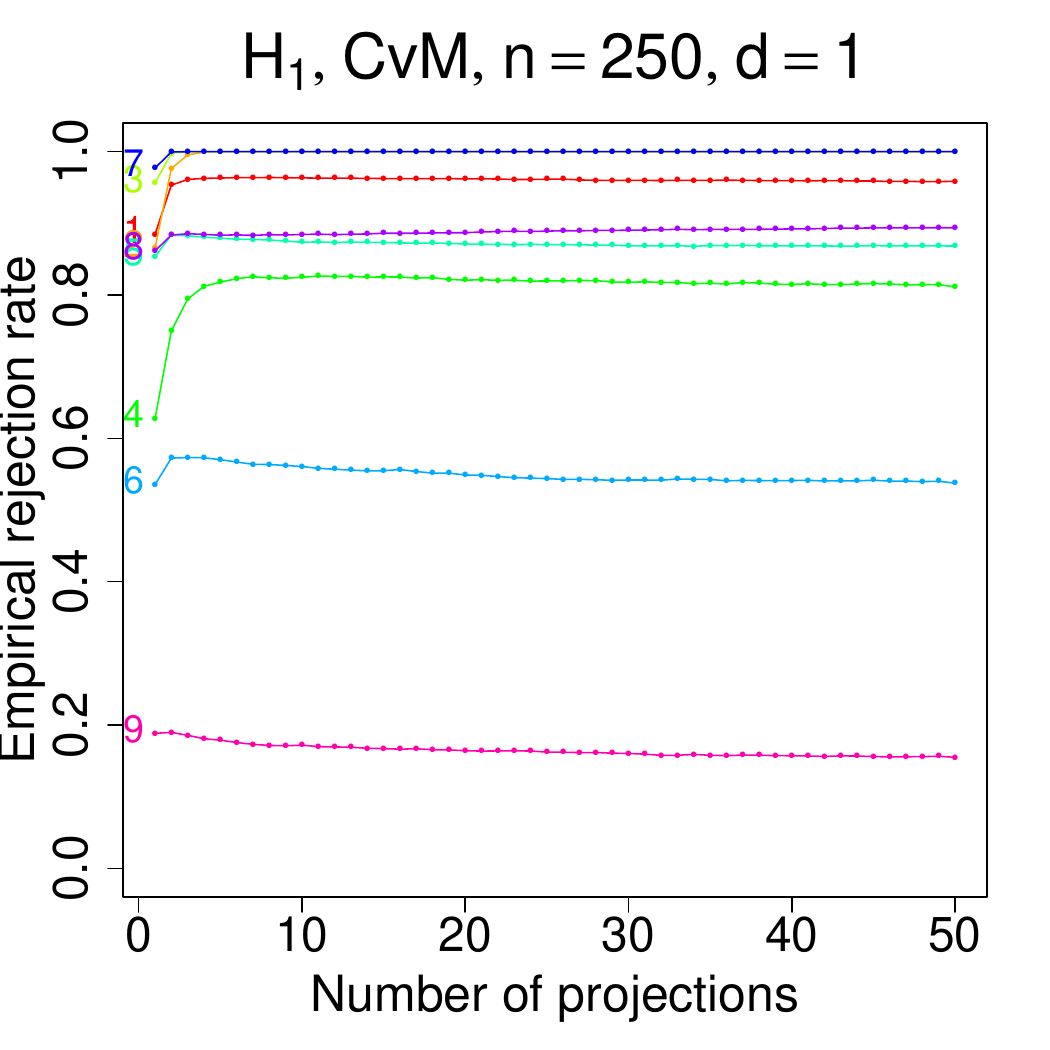}\\
\includegraphics[width=0.315\textwidth]{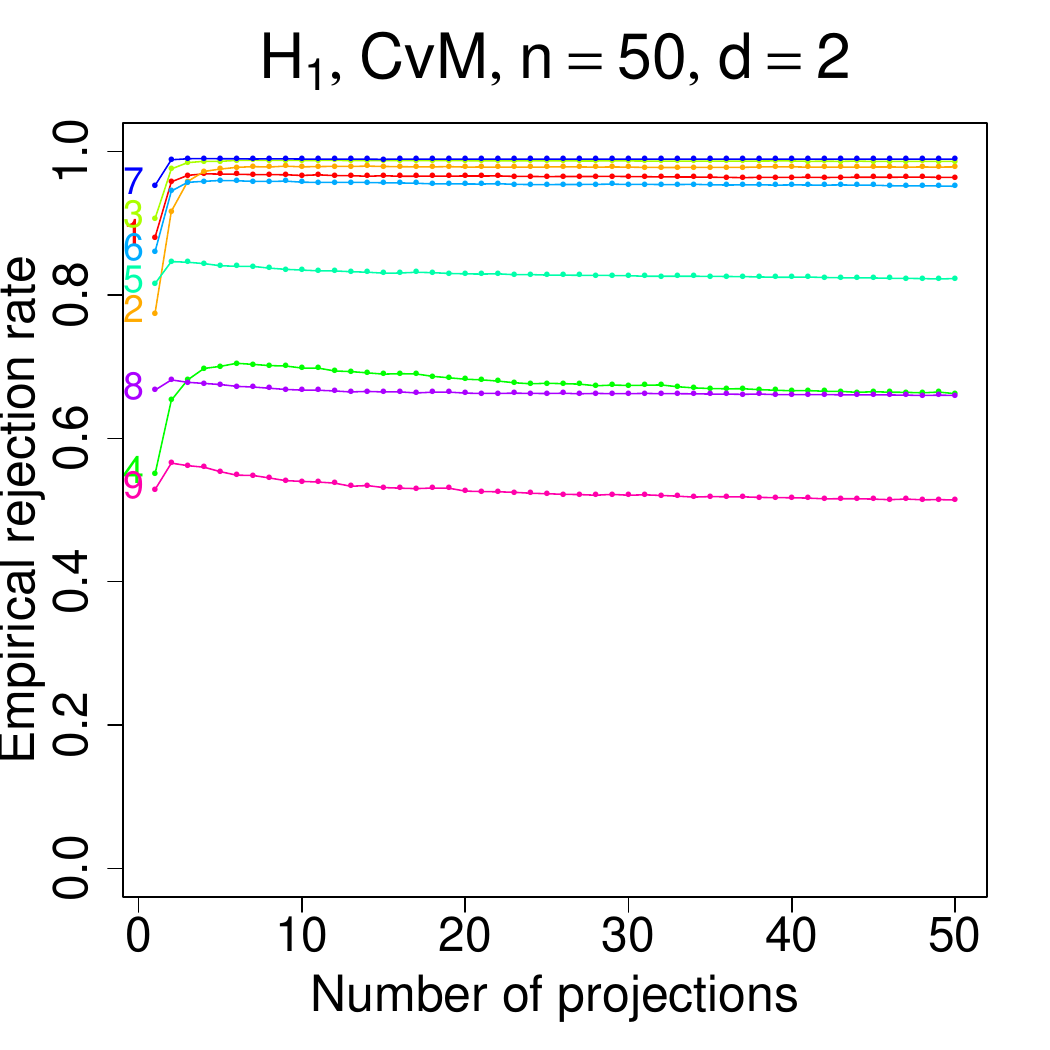}
\includegraphics[width=0.315\textwidth]{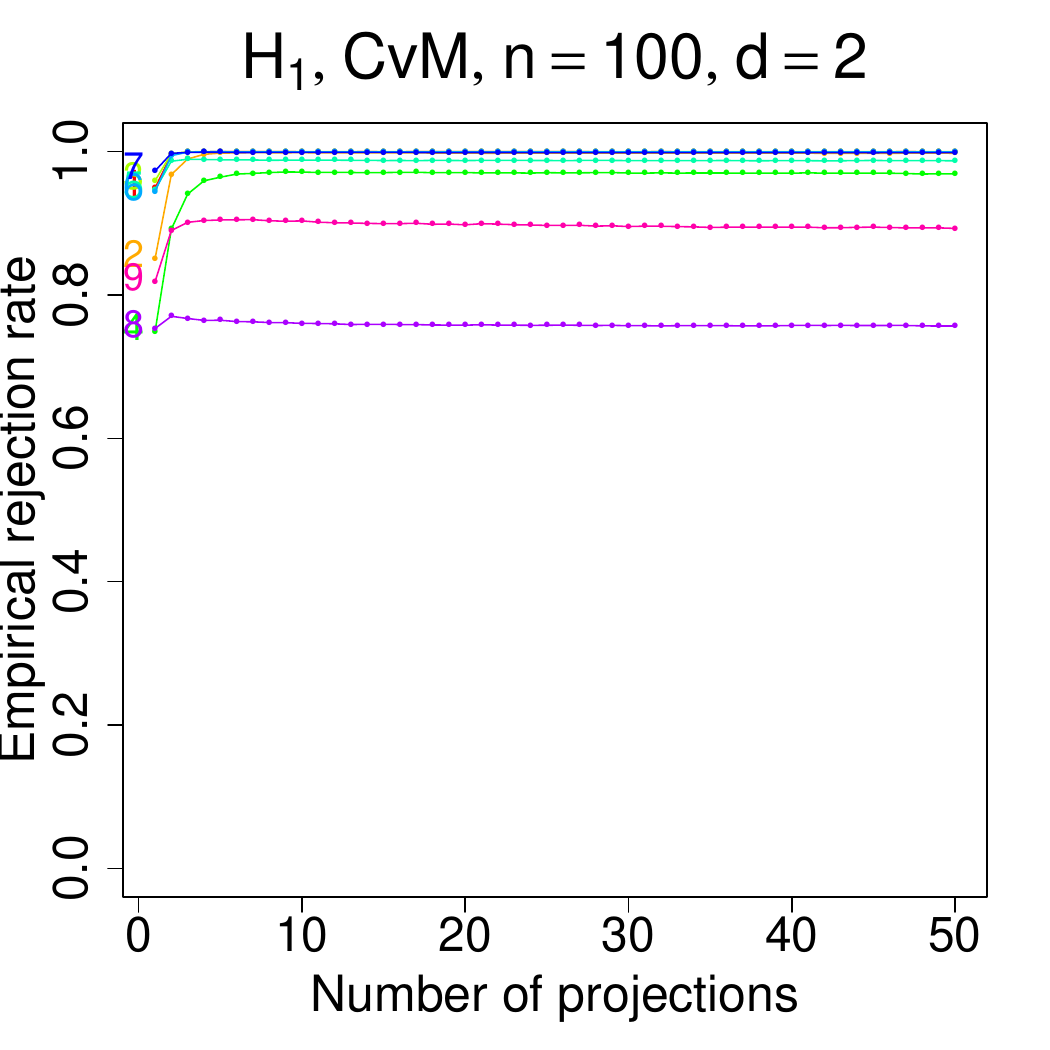}
\includegraphics[width=0.315\textwidth]{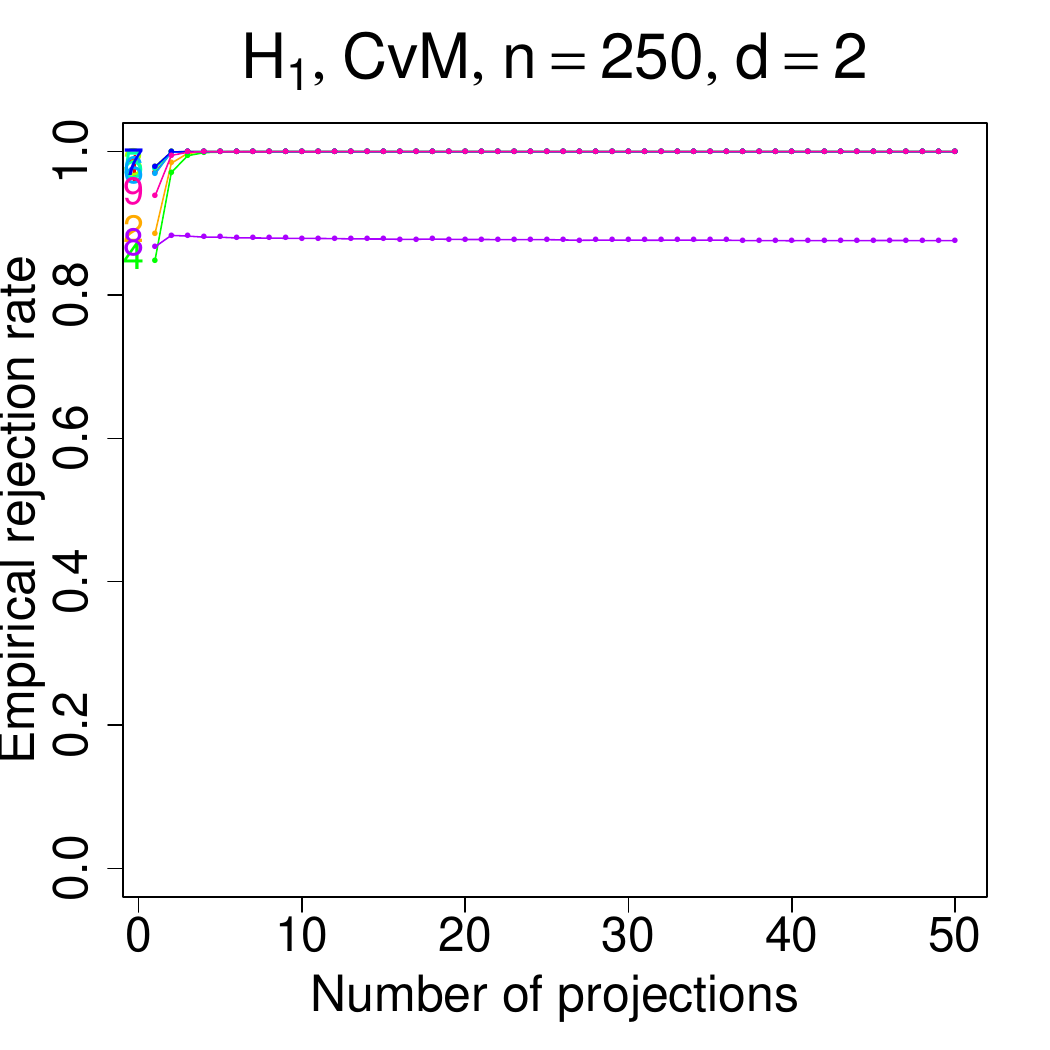}\\
\includegraphics[width=0.315\textwidth]{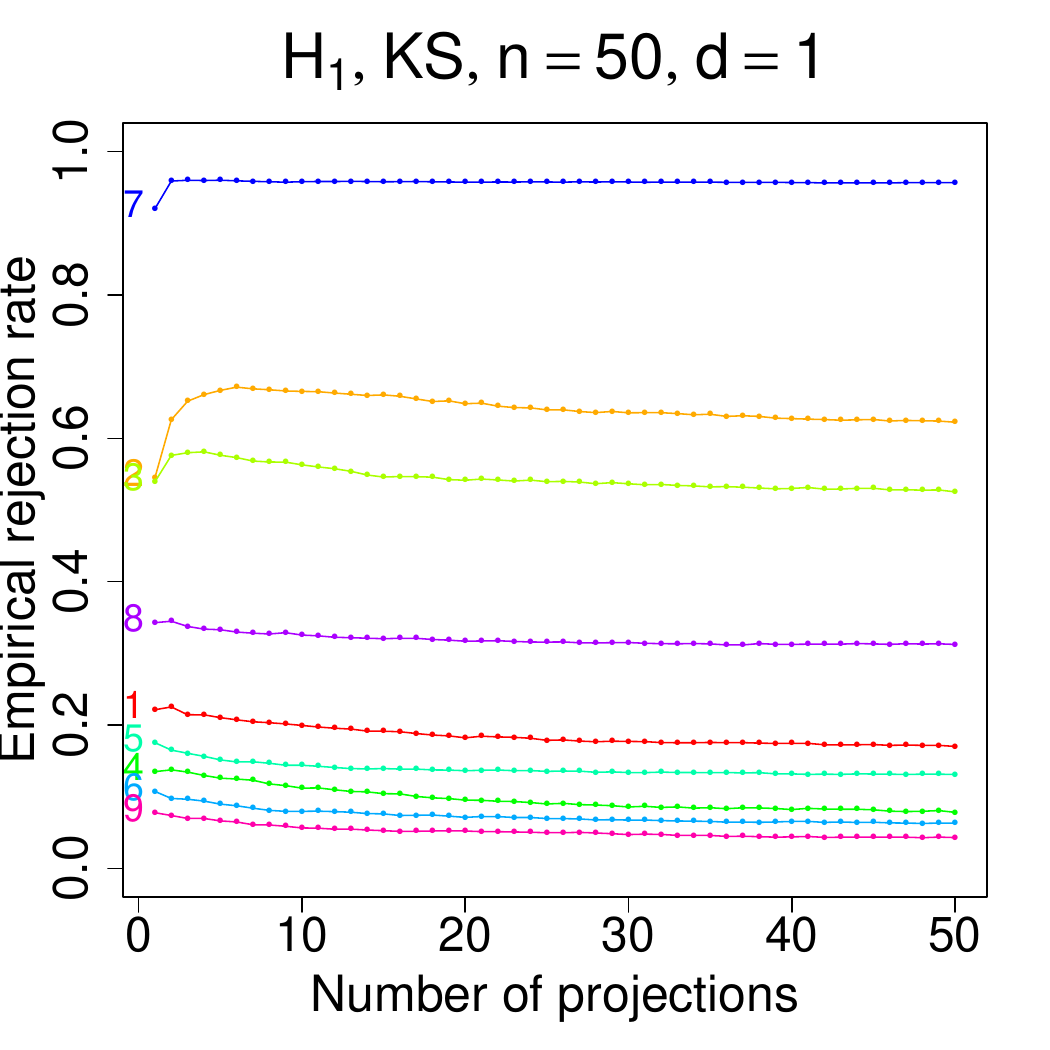}
\includegraphics[width=0.315\textwidth]{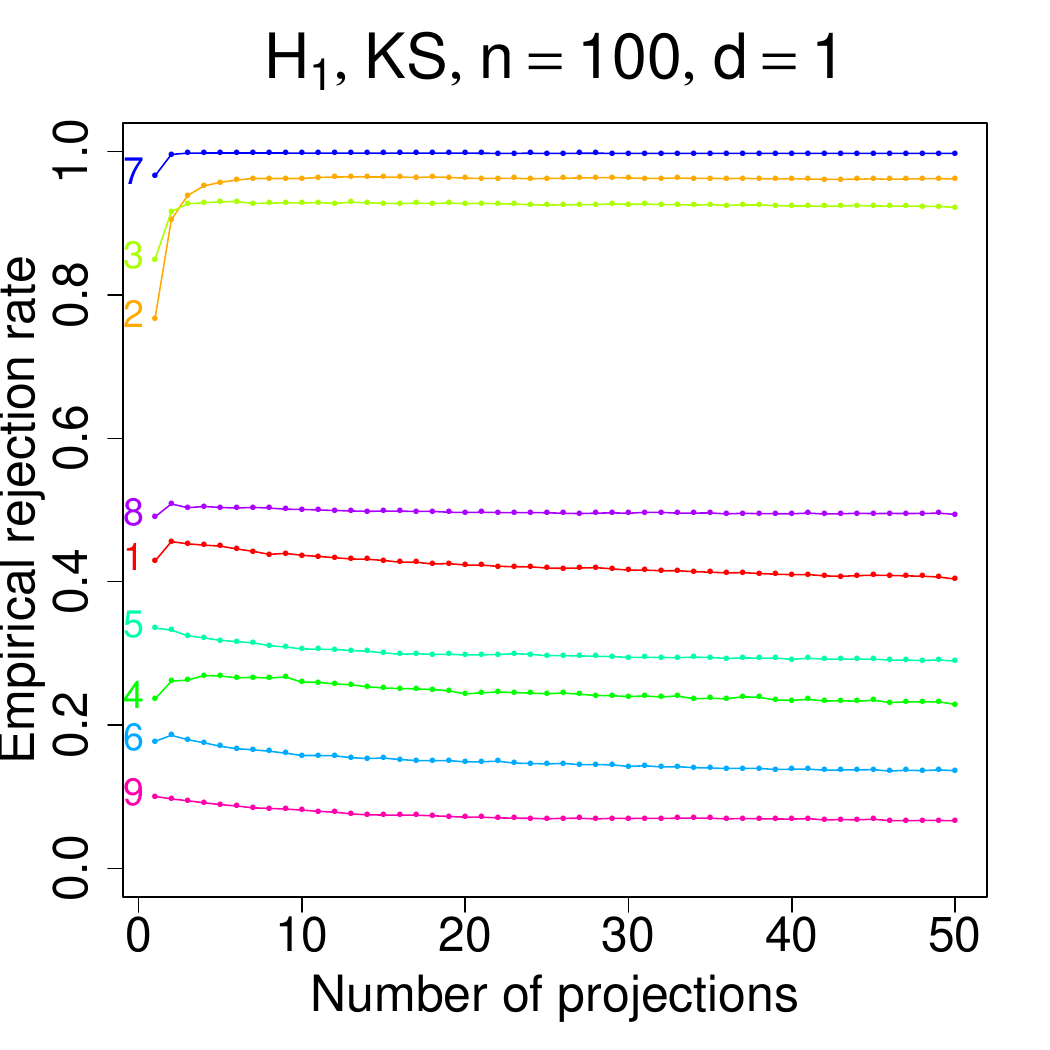}
\includegraphics[width=0.315\textwidth]{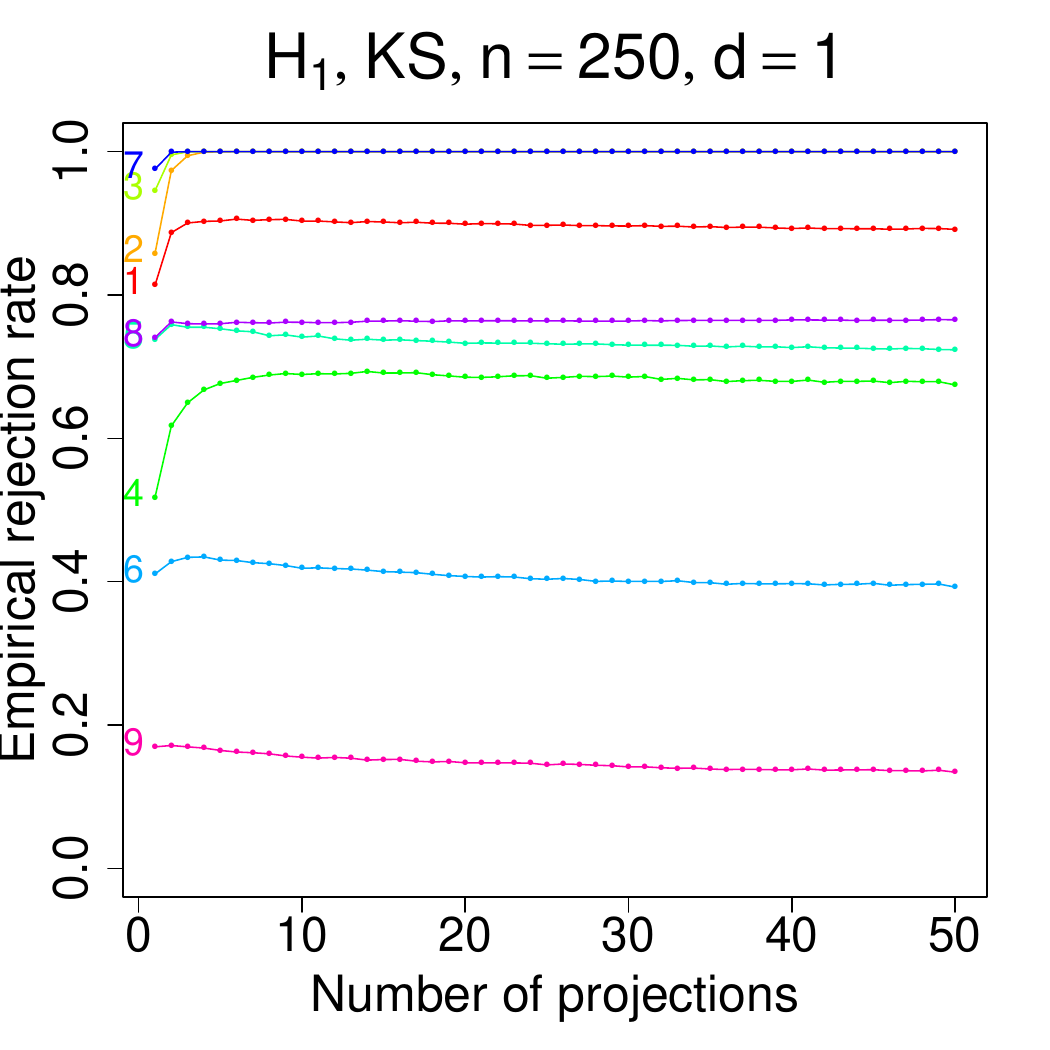}\\
\includegraphics[width=0.315\textwidth]{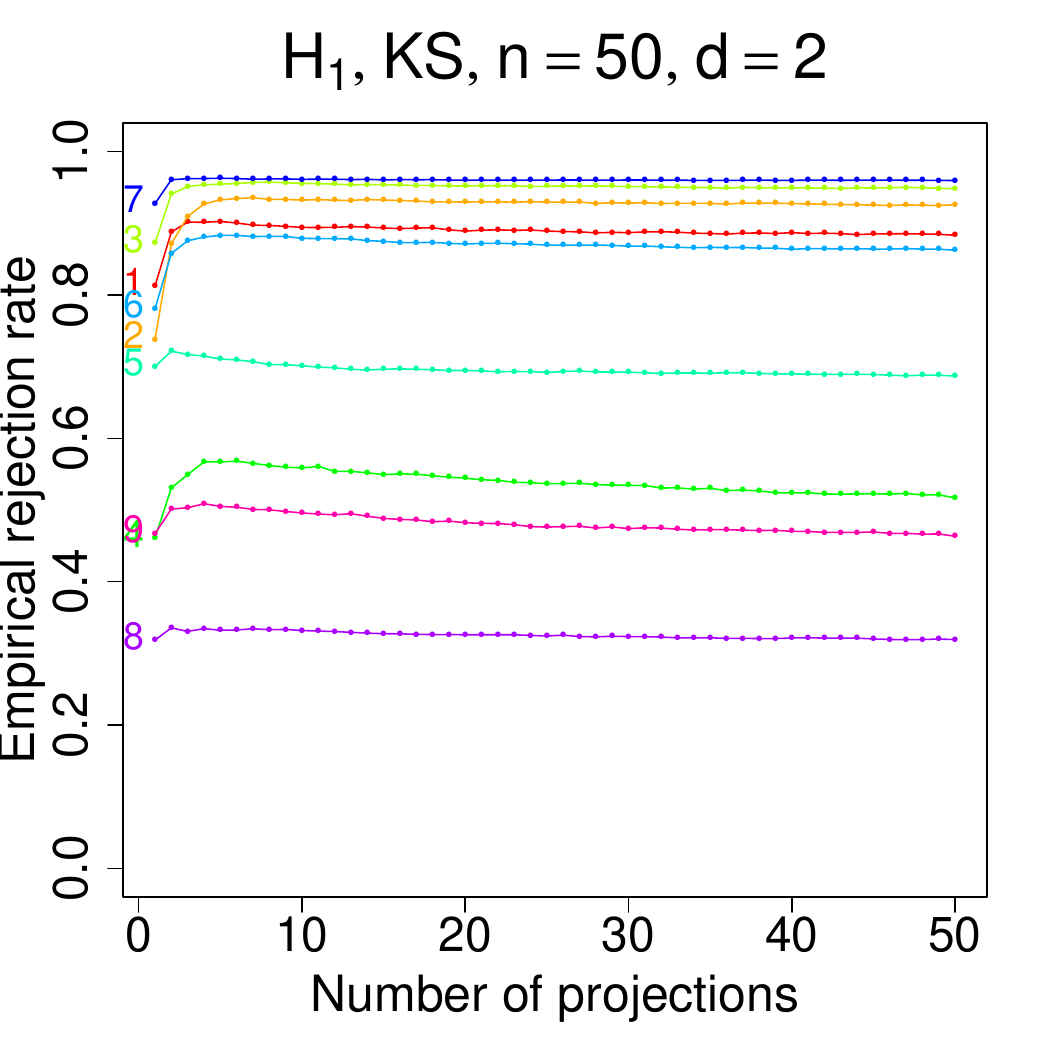}
\includegraphics[width=0.315\textwidth]{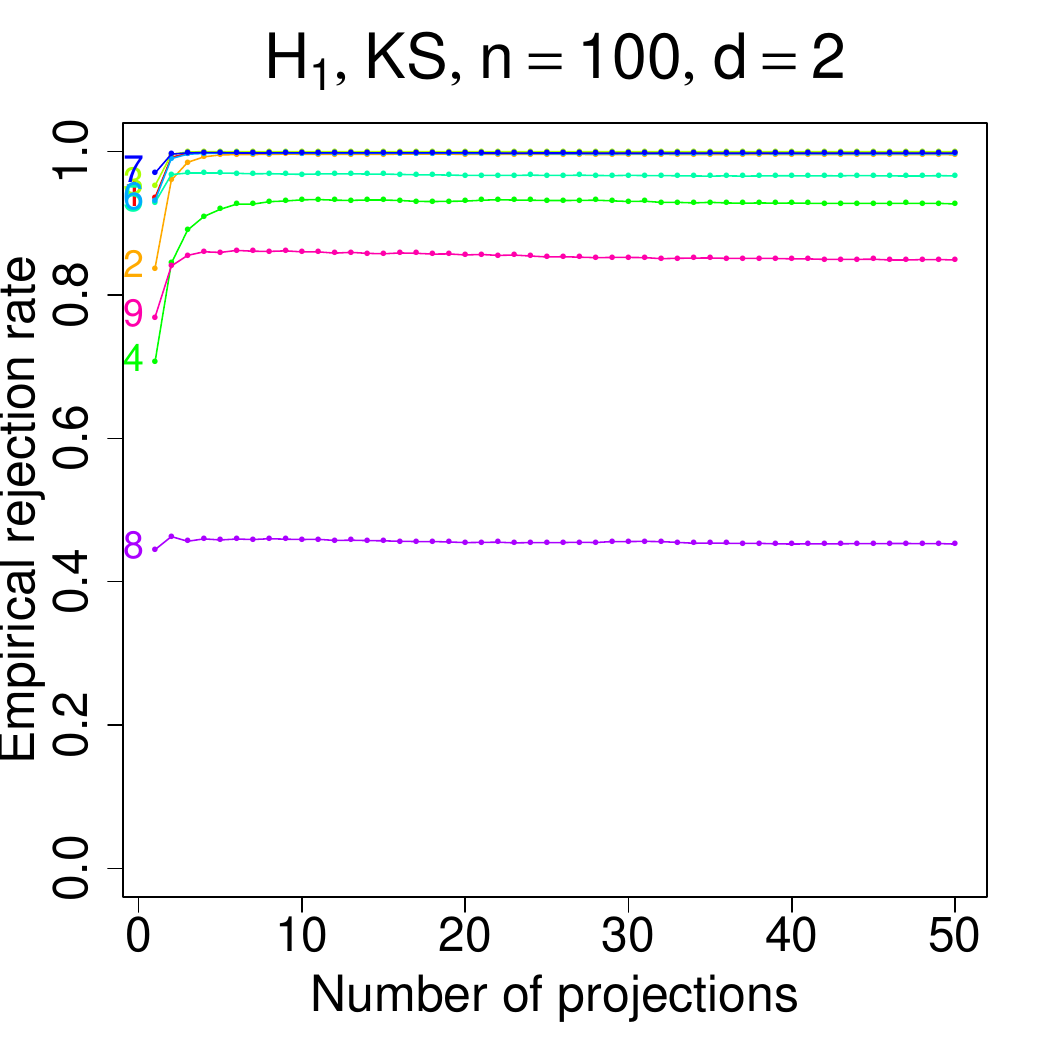}
\includegraphics[width=0.315\textwidth]{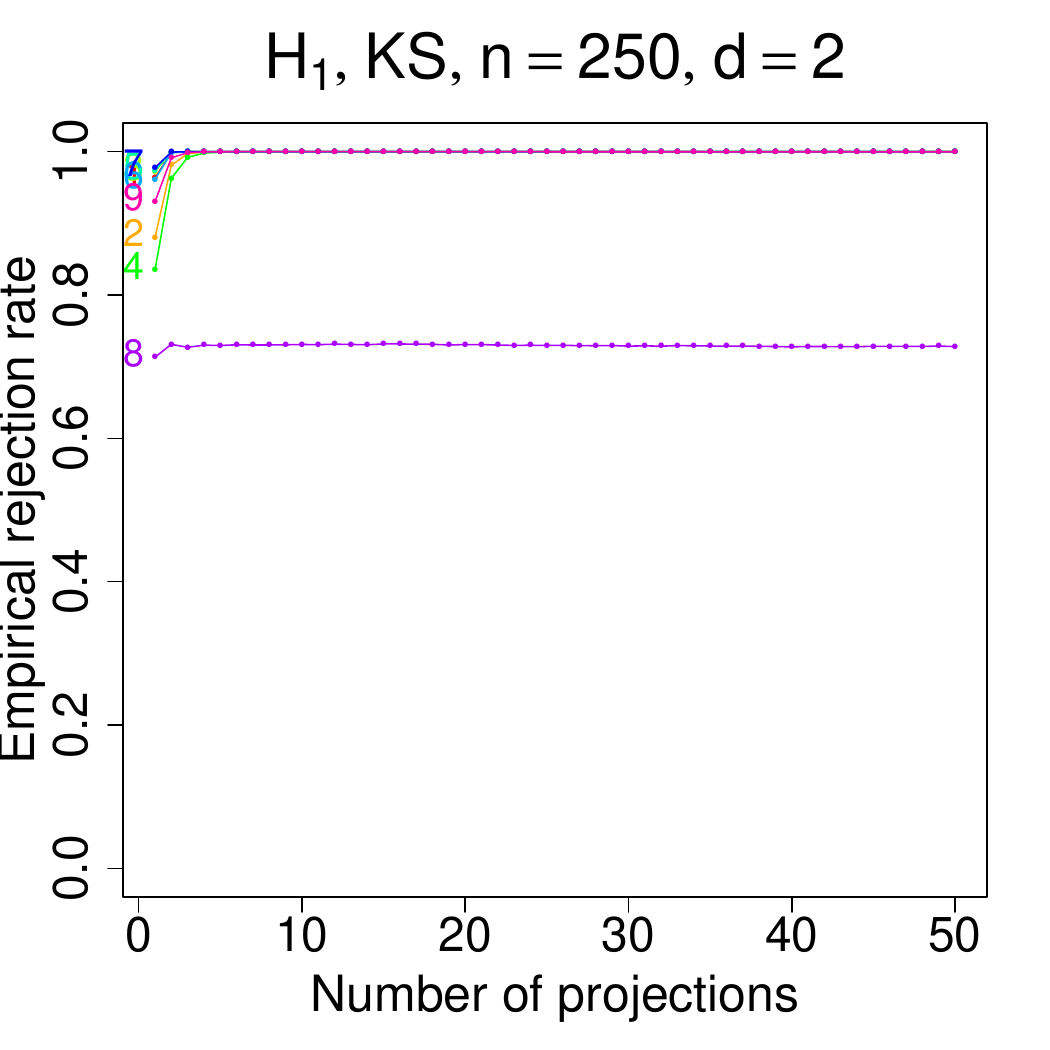}
\caption{\small Empirical powers of the CvM (first two rows) and KS (last two) tests for scenario S$k$, $k=1,\ldots,9$, depending on the number of projections $K=1,\ldots,50$. Odd rows correspond to the deviation index $d=1$, while even account for $d=2$. The significance level is $\alpha=0.05$ and the sample sizes are $n=50,100,250$ (columns, from left to right). \label{fig:powproj}}
\end{figure}

We explore first the dependence of the tests with respect to the number of projections $K$. Figure \ref{fig:sizeproj} shows the empirical level for each scenario, based on $M=10000$ Monte Carlo trials and $B=10000$ bootstrap replicates. There is a clear L-shaped pattern in the empirical rejection rate curves, which is produced by the conservativeness of the FDR correction --- under $H_0$, it ensures that the rejection rate is \textit{at most} $\alpha$ --- when dealing with the highly-dependent projected tests. For small $K$'s ($K\approx3$), both tests calibrate the three levels for different sample sizes reasonably well, with the main exception being $n=50$ and $\alpha=0.10$, for which the tests have a significant over-rejection of the null hypothesis. For moderate to large $K$'s, the empirical rejection rates decrease and stabilize below $\alpha$, resulting in a systematic violation of the confidence intervals. Figure \ref{fig:powproj} shows that the empirical powers with respect to $K$ are almost constant or exhibit mild decrements, except for certain bumps at lower values of $K$ that provide a significant power gain. Both facts point towards choosing the number of projections $K$ to be relatively small, $K\in\{1,2,3,4,5\}$ and particularly $K=3$, in order to make a reasonable compromise between correct calibration and power. In addition to the computational expediency that a small $K$ yields, it also avoids requiring a large $B$ to estimate properly the FDR $p$-values, provided that the FDR correction requires a finer precision in the discretization of the $p$-values for larger $K$ (see the supplement). \\

The tests based on the KS and CvM norms are compared with the test presented in \cite{Garcia-Portugues:flm} (denoted by PCvM), available in the \texttt{R} package \texttt{fda.usc} \citep{Febrero-Bande2011}, and whose test statistic can be regarded as the average of projected CvM statistics. The test was run with the same FPC estimation used in the new tests, the same number of components $d_n$, and $B=10000$ (considered also for CvM and KS). Table \ref{tab:results} presents the empirical rejection rates of the different simulation scenarios with $K=1, 3, 5$ for KS and CvM tests. The results show two consistent patterns. First, in our simulation scenarios, the CvM test consistently dominates over the KS test, with only one exception: $H_{9,1}$ with $n=50$ (see supplement for the latter). This is coherent with the fact that quadratic norms in goodness-of-fit tests are often more powerful than sup-norms (see, e.g., p. 110 of \cite{DAgostino1987} for the distribution case). Second, PCvM tends to have a larger power than CvM for most of the situations, especially for small sample sizes and mild deviations. As an illustration, for $n=50$, the average relative loss in the empirical power for CvM$_3$ with respect to PCvM is $12.7\%$ ($d=1$) and $4.6\%$ ($d=2$). For $n=100$, the losses drop to $9.3\%$ and $1.3\%$, respectively, and for $n=250$, to $5.2\%$ and $0.2\%$, respectively. The drop in performance for CvM with respect to PCvM is expected due to the construction of CvM, which opts for exploring a set of random directions instead of averaging uniformly distributed finite-dimensional directions, as PCvM does. This also yields one the strongest points of the CvM test, which is its relatively short running times, especially for large $n$. Not surprisingly, the number of evaluations performed for computing the CvM statistic is $\Oh(n)$, a notable reduction from PCvM's $\Oh((n^3-n^2)/2)$. Also, the memory requirement for CvM is $\Oh(n)$, instead of PCvM's $\Oh((n^2-n-2)/2)$. The running times in Figure \ref{fig.times} is evidence of this improvement. \\

The new tests were also applied to the two data applications described in \cite{Garcia-Portugues:flm}, yielding similar conclusions. Both datasets are provided in the library \texttt{fda.usc}. The first example uses the classical Tecator dataset, considered in Section 2.1.1 of \cite{Ferraty2006} as a motivating example for introducing nonlinear regression models. The dataset contains $215$ spectrometric curves measuring the absorbance at wavelengths $[850,1050]$ of finely chopped meat samples. Covariates giving the fat, water, and protein content of the meat are also available in the dataset. Typically, the goal is to predict the fat content of a meat sample using the spectrometric curve or any of its derivatives, and, for that, the FLM has been proposed as a candidate model. We test its adequacy for the dataset with the new goodness-of-fit tests proposed. The $p$-values obtained for $K=3$ projections and $B=10000$ are $0.020$ and $0.022$ for CvM and KS, respectively. Using $\hat{\brho}$ with $d_n$ selected by SICc in PCvM gave a $p$-value of $0.006$. Employing the first or second derivatives of the absorbance curves provided null $p$-values. In addition, the tests for $H_0:\brho=\textbf{0}$ also had null $p$-values for all of tests. As a consequence, we conclude that, at level $\alpha=0.05$, there is evidence against the FLM and there is a significant nonlinear relation between the fat content and the absorbance curves. \\

The second example mimics the classical dataset in \cite{Ramsay2005} on Canadian weather stations. The data are contains yearly profiles of temperature from $73$ weather stations of the AEMET (Spanish Meteorological Agency; Spanish acronym) network and other meteorological variables, and the goal is to explain the mean of the wind speed at each location. Prior to its analysis, the dataset was preprocessed to remove the $5\%$ most outlying curves using the \cite{Fraiman2001} depth. With the same settings as before, the CvM and KS tests for $K=3$ projections provided $p$-values equal to $0.612$ and $0.396$, respectively, and PCvM gave a $p\text{-value}=0.080$. In addition, the tests for $H_0:\brho=\textbf{0}$ yielded null $p$-values. We conclude that, at level $\alpha=0.05$, there is no evidence against the FLM and that the effect of the covariate on the response is significant and linear.
\begin{table}[htb!]
\centering
\footnotesize
\setlength{\tabcolsep}{1pt}
\begin{tabular}{l|>{\centering}m{1.00cm}>{\centering}m{1.00cm}>{\centering}m{1.00cm}|>{\centering}m{1.00cm}>{\centering}m{1.00cm}>{\centering}m{1.00cm}|c||>{\centering}m{1.00cm}>{\centering}m{1.00cm}>{\centering}m{1.00cm}|>{\centering}m{1.00cm}>{\centering}m{1.00cm}>{\centering}m{1.00cm}|c}
  \toprule  \toprule
\multirow{2}{*}{$H_{k,\delta}$} & \multicolumn{7}{c||}{$n=100$} & \multicolumn{7}{c}{$n=250$} \\ \cmidrule{2-15}
&  CvM$_{1}$ & CvM$_{3}$ & CvM$_{5}$ & KS$_{1}$ & KS$_{3}$ & KS$_{5}$ & PCvM  
&  CvM$_{1}$ & CvM$_{3}$ & CvM$_{5}$ & KS$_{1}$ & KS$_{3}$ & KS$_{5}$ & PCvM \\ \midrule
$H_{1,0}$ & $5.1$ & $3.9$ & $3.5$ & $5.5$ & $4.6$ & $4.2$ & $4.8$ & $5.4$ & $3.9$ & $3.6$ & $5.6$ & $4.2$ & $3.9$ & $4.9$ \\
$H_{2,0}$ & $5.4$ & $4.6$ & $4.2$ & $5.6$ & $5.1$ & $4.6$ & $3.6$ & $5.8$ & $4.8$ & $4.3$ & $6.1$ & $5.4$ & $4.9$ & $4.7$ \\
$H_{3,0}$ & $6.2$ & $4.9$ & $4.5$ & $7.0$ & $6.0$ & $5.2$ & $5.7$ & $5.6$ & $4.1$ & $3.8$ & $5.8$ & $4.7$ & $4.2$ & $5.3$ \\
$H_{4,0}$ & $5.9$ & $4.4$ & $4.1$ & $5.9$ & $5.0$ & $4.8$ & $4.6$ & $6.3$ & $5.2$ & $4.8$ & $6.4$ & $5.9$ & $5.6$ & $4.9$ \\
$H_{5,0}$ & $5.5$ & $4.0$ & $3.6$ & $6.0$ & $4.3$ & $4.0$ & $4.9$ & $4.9$ & $4.2$ & $3.8$ & $5.0$ & $4.0$ & $3.5$ & $5.0$ \\
$H_{6,0}$ & $5.4$ & $4.3$ & $3.9$ & $6.0$ & $4.9$ & $4.5$ & $5.2$ & $5.6$ & $4.3$ & $4.0$ & $6.0$ & $5.0$ & $4.8$ & $4.8$ \\
$H_{7,0}$ & $5.5$ & $3.9$ & $3.7$ & $6.0$ & $4.7$ & $4.0$ & $5.1$ & $5.4$ & $4.1$ & $3.8$ & $5.5$ & $4.7$ & $4.1$ & $5.2$ \\
$H_{8,0}$ & $5.1$ & $3.5$ & $3.3$ & $5.3$ & $3.7$ & $3.4$ & $4.9$ & $5.3$ & $3.9$ & $3.6$ & $5.4$ & $4.3$ & $3.9$ & $5.1$ \\
$H_{9,0}$ & $6.3$ & $4.8$ & $4.3$ & $6.1$ & $4.9$ & $4.5$ & $6.1$ & $5.6$ & $4.4$ & $4.1$ & $5.7$ & $4.8$ & $4.1$ & $5.9$ \\
\midrule
$H_{1,1}$ & $56.0$ & $59.4$ & $58.3$ & $42.9$ & $45.0$ & $43.7$ & $69.9$ & $88.4$ & $96.3$ & $96.3$ & $81.4$ & $90.3$ & $90.3$ & $98.4$ \\
$H_{2,1}$ & $80.1$ & $98.5$ & $98.7$ & $76.7$ & $95.7$ & $96.3$ & $99.2$ & $86.5$ & $100$ & $100$ & $85.8$ & $100$ & $100$ & $100$ \\
$H_{3,1}$ & $90.2$ & $97.6$ & $97.4$ & $85.0$ & $93.0$ & $92.8$ & $99.2$ & $95.6$ & $100$ & $100$ & $94.5$ & $100$ & $100$ & $100$ \\
$H_{4,1}$ & $31.2$ & $35.7$ & $35.3$ & $23.6$ & $26.8$ & $26.0$ & $43.6$ & $62.7$ & $81.8$ & $82.5$ & $51.8$ & $67.7$ & $68.9$ & $88.6$ \\
$H_{5,1}$ & $45.2$ & $43.1$ & $42.1$ & $33.5$ & $31.8$ & $30.6$ & $49.9$ & $85.4$ & $87.9$ & $87.4$ & $73.8$ & $75.3$ & $74.1$ & $91.5$ \\
$H_{6,1}$ & $23.3$ & $22.2$ & $20.8$ & $17.7$ & $17.0$ & $15.7$ & $27.9$ & $53.5$ & $57.0$ & $56.1$ & $41.1$ & $43.0$ & $41.9$ & $66.9$ \\
$H_{7,1}$ & $96.9$ & $99.9$ & $99.9$ & $96.6$ & $99.8$ & $99.8$ & $99.9$ & $97.7$ & $100$ & $100$ & $97.5$ & $100$ & $100$ & $100$ \\
$H_{8,1}$ & $73.3$ & $74.8$ & $74.5$ & $49.0$ & $50.3$ & $50.1$ & $74.7$ & $86.2$ & $88.3$ & $88.4$ & $74.1$ & $76.0$ & $76.2$ & $87.7$ \\
$H_{9,1}$ & $10.6$ & $9.2$ & $8.6$ & $10.0$ & $8.9$ & $8.1$ & $12.1$ & $18.8$ & $17.9$ & $17.2$ & $17.0$ & $16.4$ & $15.5$ & $22.3$ \\
\midrule
$H_{1,2}$ & $94.9$ & $100$ & $100$ & $93.6$ & $99.9$ & $99.9$ & $100$ & $97.0$ & $100$ & $100$ & $96.3$ & $100$ & $100$ & $100$ \\
$H_{2,2}$ & $85.0$ & $99.8$ & $99.9$ & $83.7$ & $99.5$ & $99.6$ & $99.9$ & $88.6$ & $100$ & $100$ & $88.0$ & $100$ & $100$ & $100$ \\
$H_{3,2}$ & $95.9$ & $100$ & $100$ & $95.2$ & $100$ & $100$ & $100$ & $97.9$ & $100$ & $100$ & $97.2$ & $100$ & $100$ & $100$ \\
$H_{4,2}$ & $74.8$ & $96.4$ & $97.2$ & $70.6$ & $92.0$ & $93.3$ & $98.2$ & $84.8$ & $99.9$ & $100$ & $83.5$ & $99.9$ & $100$ & $100$ \\
$H_{5,2}$ & $94.7$ & $98.9$ & $98.8$ & $92.9$ & $97.0$ & $96.8$ & $99.1$ & $97.6$ & $100$ & $100$ & $97.4$ & $100$ & $100$ & $100$ \\
$H_{6,2}$ & $94.4$ & $100$ & $100$ & $93.2$ & $99.8$ & $99.8$ & $100$ & $96.9$ & $100$ & $100$ & $96.1$ & $100$ & $100$ & $100$ \\
$H_{7,2}$ & $97.3$ & $99.9$ & $99.9$ & $97.0$ & $99.9$ & $99.8$ & $99.9$ & $97.9$ & $100$ & $100$ & $97.8$ & $100$ & $100$ & $100$ \\
$H_{8,2}$ & $75.3$ & $76.5$ & $76.0$ & $44.5$ & $45.8$ & $45.9$ & $78.2$ & $86.7$ & $88.0$ & $87.9$ & $71.3$ & $73.0$ & $73.1$ & $88.9$ \\
$H_{9,2}$ & $81.8$ & $90.5$ & $90.3$ & $76.9$ & $85.9$ & $86.0$ & $93.9$ & $93.8$ & $100$ & $100$ & $93.0$ & $100$ & $100$ & $100$ \\
\bottomrule\bottomrule
\end{tabular}
\caption{\small Empirical sizes and powers (in percentages) of the CvM, KS, and PCvM tests with $\alpha=0.05$, sample sizes $n=100, 250$, and estimation of $\brho$ by data-driven FPC ($d_n$ chosen by SICc). KS and CvM tests are shown with $1$, $3$, and $5$ projections. \label{tab:results}}
\end{table}

\begin{figure}[h!]
\centering
\includegraphics[width=0.55\textwidth]{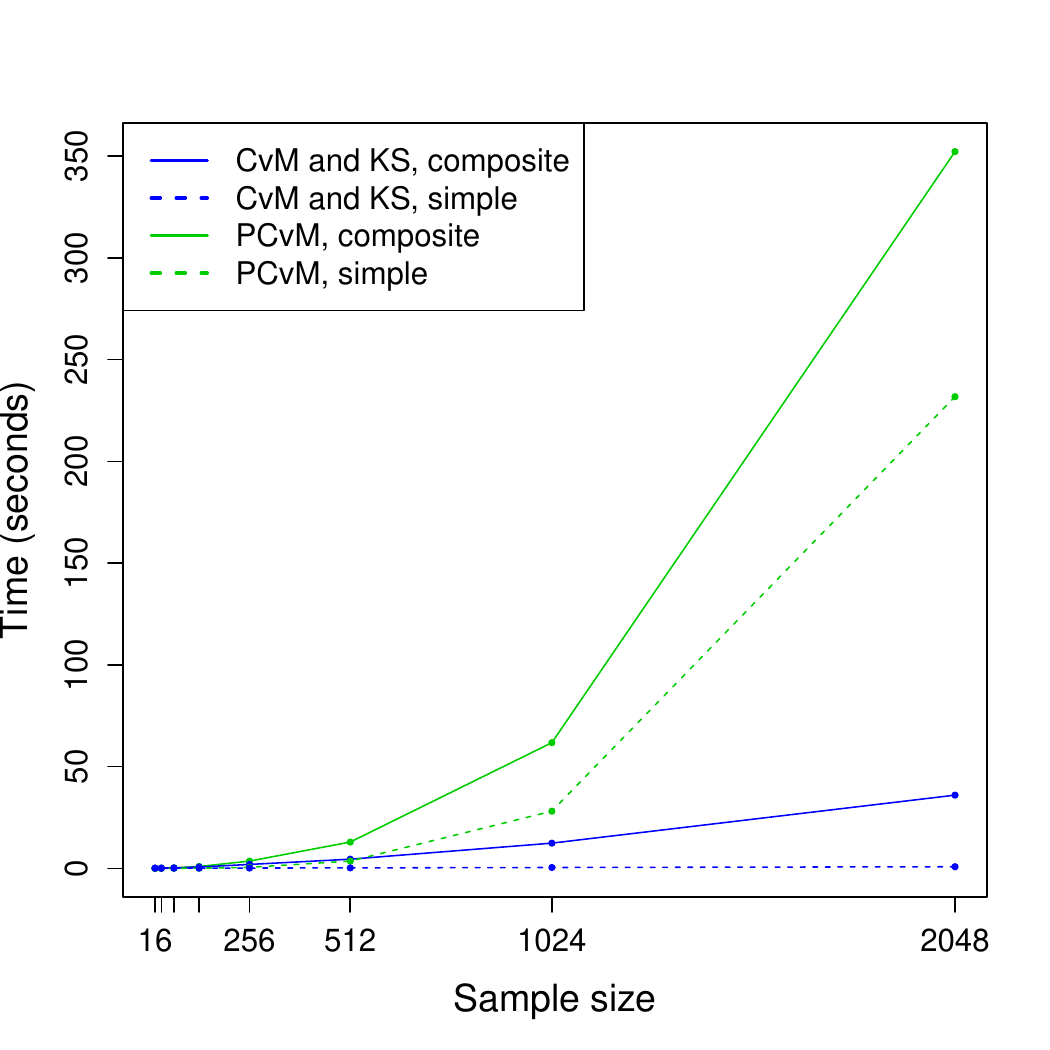}
\caption{\small Running times for the CvM and KS tests (the computation of both tests is done in the same routine) and PCvM, for the composite and simple hypotheses. The tests were averaged over $M=100$ trials and calibrated with $B=1000$. The sample sizes are $n=2^k$, $k=4,\ldots,11$, and the number of projections considered is $K=3$. Times were measured on a 2.53 GHz core. All the tests have a similar implementation in \texttt{R} that interfaces \texttt{FORTRAN} for the computation of the statistics. \label{fig.times}}
\end{figure}

\section{Discussion}
\label{sec:final}

We have presented a new way of building goodness-of-fit tests for regression models with functional covariates employing random projections. The methodology was illustrated using randomly projected empirical processes, which provided root-$n$ consistent tests for testing functional linearity. The calibration of the tests was done by a wild bootstrap resampling and the FDR was used to combine $K$ $p$-values coming from different projections to account for a higher power. The empirical analysis of the tests, conducted in a fully data-driven way, showed that, in our simulation scenarios, CvM yields higher powers than KS and that a selection of $K\in\{1,\ldots,5\}$, in particular $K=3$, is a reasonable compromise between respecting size and increasing power. There is still a price to pay in terms of a moderate loss of power with respect to the PCvM test, which averages across a set of uniformly distributed finite-dimensional directions. However, the reduction in computational complexity of the new tests is more than notable. \\ 

We conclude the paper by sketching some promising extensions of the methodology for the testing of more complex models involving functional covariates: 
\begin{itemize}
\item[(a)] Testing the significance of the functional covariate of $(\bX,{\bf W})\in\Hil\times\R^q$ in the functional partially linear model \citep{Aneiros-Perez2006} $Y=m(\bX)+{\bf W}'{\boldsymbol\beta}+\varepsilon$. The process to be considered for a sample $\{\lrp{\bX_i,{\bf W}_i,Y_i}\}_{i=1}^n$ and an estimator $\hat{\boldsymbol{\beta}}$ such that $\hat{\boldsymbol{\beta}}-\boldsymbol{\beta}=\OP(n^{-1/2})$ is
$n^{-1/2} \sum_{i=1}^n\mathds{1}_{\lrb{\bX_i^\bh\leq x}}\big(Y_i-{\bf W}_i'\hat{\boldsymbol{\beta}}\big)$.
\item[(b)] Testing a functional quadratic regression model \citep{Horvath2013}.
\item[(c)] Testing the significance of a functional linear model with functional response:
$H_0:\E{\bY|\bX}=\mathbf{0}$, where now $\lrp{\bX,\bY}\in\Hil_1\times\Hil_2$ and the associated empirical process is $n^{-1/2} \sum_{i=1}^n\mathds{1}_{\big\{\bX_i^{\bh_1}\leq x\big\}}\bY_i^{\bh_2}$.
\end{itemize} 

\section*{Software availability}

The \texttt{R} package \texttt{rp.flm.test}, openly available at \url{https://github.com/egarpor/rp.flm.test}, contains the implementation of the tests and allows reproduction of the simulation study and data applications. The main function, \texttt{rp.flm.test}, has also been included in the \texttt{R} package \texttt{fda.usc} since version 1.3.1. 

\section*{Supplement}

Two extra appendices are included as supplementary material, containing the proofs of the technical lemmas and further results for the simulation study.

\section*{Acknowledgments}

The first author was supported by projects MTM2014--56235--C2--2--P and MTM2017--86061--C2--2--P from the Spanish Ministry of Economy, Industry and Competitiveness. The remaining authors were supported by: projects MTM2013--41383--P and MTM2016--76969--P from the Spanish Ministry of Economy, Industry and Competitiveness, and the European Regional Development Fund; project 10MDS207015PR from Dirección Xeral de I+D, Xunta de Galicia; IAP network StUDyS from Belgian Science Policy. The second author acknowledges the support by FPU grant AP2010--0957 from the Spanish Ministry of Education and the Dynamical Systems Interdisciplinary Network, University of Copenhagen. We thank one anonymous referee and an Associate Editor for detailed and useful comments that led to significant improvements on an earlier draft of the paper. We gratefully acknowledge the computational resources of the Supercomputing Center of Galicia (CESGA).


\appendix

\section{Proofs of the main results}
\label{ap:proofs}

\subsection{Hypothesis projection}
\label{ap:hyproj}

\begin{proof}[Proof of Proposition \ref{PropNoRandom}]
We denote by $\bX_p$ both the vectors $(X_1,\ldots, X_p)'$ and $(X_1,\ldots, X_p, 0,\ldots)'$ containing the first $p$ coefficients of $\bX$ in an orthonormal basis of $\Hil$. We prove first the result for the finite subspace of $\Hil$ spanned by the first $p$ elements of the orthonormal basis. We need to show that
\begin{align}
\E{Y | \bX_p} = 0 \text{ a.s.} \iff \E{Y | \langle \bX_p, \bh \rangle} = 0 \text{ a.s. for every } \bh\in \Hil. \label{eq.1}
\end{align}
To prove this, we make use of Theorem 1 in \cite{Bierens1982}, which states that if $\bV$ and $\bZ$ are two $\R^p$-valued random vectors, then
\begin{align}
\E{\bV  | \bZ}=0 \text{ a.s.} \iff   \mathbb{E}\big[\bV e^{i \langle \mathbf{t},\mathbf{Z}\rangle}\big]=0\text{ for every } \bt \in \R^p. \label{eq.Bier}
\end{align}
Assume that $\E{Y | \bX_p} = 0$ and let $ \bh \in \Hil$. Since the $\sigma$-algebra generated by $\langle \bX_p, \bh\rangle$, $\sigma(\langle \bX_p, \bh\rangle)$, is contained in $\sigma(\bX_p)$, we have that
$\E{Y | \langle \bX_p, \bh\rangle} = \E{\E{Y |\bX_p}| \langle \bX_p, \bh\rangle} = 0$ a.s.,
which shows the \textit{if} part. To obtain the \textit{only if} part, let $\bh \in \Hil$, and compute
$\mathbb{E}\big[ Y e^{i t \left\langle \mathbf{X}_p, \bh\right\rangle}\big] =
\mathbb{E}\big[\mathbb{E}\big[ Y | \langle \bX_p, \bh\rangle \big] e^{i t \left\langle \mathbf{X}_p, \bh\right\rangle}\big]=0$, for every $t \in \R$. Then \eqref{eq.1} follows from \eqref{eq.Bier}.\\

Now we are in position to prove the result for $\Hil$. As before, the \textit{if} implication follows from $\sigma(\langle \bX, \bh\rangle)\subset\sigma(\bX)$. To prove the \textit{only if} implication, given $p\in\mathbb{N}$ and $\bh \in \Hil$, since $\bh_p \in \Hil$ and $\langle \bX, \bh_p\rangle =\langle \bX_p, \bh\rangle$, then $\sigma(\langle \bX_p, \bh\rangle)\subset\sigma(\langle \bX, \bh\rangle)$, and we have that the assumption implies that $\E{Y| \langle \bX_p, \bh\rangle} = 0$ a.s. Thus, from \eqref{eq.1}, we have that $\E{Y | \bX_p} = 0$ a.s. for every $p$, and the result follows from the fact that $\sigma(\bX_p) \uparrow \sigma(\bX)$ because of the integrability assumption on $Y$.
\end{proof}

\begin{proof}[Proof of Lemma \ref{Prop:Auxiliar}]
From the properties of the conditional expectation, the Cauchy--Schwartz, and Jensen inequalities, we have that
\[
l_k=\E{\|\bX\|^k \E{|Y| | \bX}}\leq\left(m_{2k}\right)^{1/2}
\big(\E{Y^2}\big)^{1/2}.
\]
Thus $l_k$ is finite. By the convexity of the function $t \mapsto t^{(2k+1)/2k}$ and Jensen's inequality, $m_{2k}^{1/2k} \leq m_{2k+1}^{1/(2k+1)}$. Hence, $\sum_{k=1}^\infty m_{2k}^{-1/2k}=\infty$.
\end{proof}

\begin{proof}[Proof of Theorem \ref{Th:basic}]
The \textit{only if} part is trivial because $\sigma(\bX^{\bh}) \subset \sigma(\bX)$, and then $\E{Y|\bX}=0$ a.s. implies that $\mu(\Hil_0)=1$. Concerning the \textit{if} part, let us assume that $\mu(\Hil_0)>0$. From the assumptions, we have that $\E{|Y| | \bX}<\infty$, and, if we take $\bh \in \Hil_0$, then
\begin{align}
\E{Y} = \mathbb{E}\big[\mathbb{E}\big[Y \big| \bX^{\bh}\big]\big] =0. \label{Eq:0}
\end{align}
Let us assume that $\E{Y | \bX}$ is not zero a.s. Then the random variables
\begin{align*}
\Phi^+(\bX):=&\left(\E{Y | \bX}\right)^+ = \max\left\{\E{Y | \bX},0\right\},\\
\Phi^-(\bX):=&\left(\E{Y | \bX}\right)^- = \max\left\{-\E{Y | \bX},0\right\},
\end{align*}
are integrable and positive with positive probability. Thus, \eqref{Eq:0} implies that
\[
V:=\int \Phi^+(\bx) \,\mathrm{d}P_{\mathbf{X}}(\bx) = \int \Phi^-(\bx) \,\mathrm{d}P_{\mathbf{X}}(\bx)>0.
\]
Consider now the probability measures $\nu_{\Phi}^+$ and $\nu_{\Phi}^-$, which are defined on $\mathcal{H}$ and whose Radon--Nikodym derivatives with respect to $P_{\mathbf{X}}$ are, respectively,
\[
\frac {\mathrm{d}\nu_\Phi^+}{\mathrm{d}P_{\mathbf{X}}}(\bx)  := V^{-1} \Phi^+(\bx)
\quad \text{and} \quad
\frac {\mathrm{d}\nu_\Phi^-}{\mathrm{d}P_{\mathbf{X}}} (\bx)  := V^{-1} \Phi^-(\bx).
\]
For $k\in\mathbb{N}$, the moments of $\nu_{\Phi^+}$ verify that (analogously for $\Phi^-$)
\[
\int \| \bx\|^k\,\mathrm{d}\nu_{\Phi^+}(\bx)
\leq
V^{-1} \int  \| \bx\|^k\E{|Y| | \bX=\bx} \,\mathrm{d}P_{\mathbf{X}}(\bx)=l_k,
\]
and then, due to Lemma \ref{Prop:Auxiliar}, they satisfy \ref{T:cwgauss:a} in Lemma \ref{T:cwgauss}. Given $\bh \in \Hil_0$, the r.v. $\bX^{\bh}$ is $\bX$-measurable. Thus, a.s.
\[
0 =
\mathbb{E}\big[Y \big| \bX^{\bh} \big]
=
\mathbb{E}\big[\E{Y|\bX} \big| \bX^{\bh} \big]
=
\mathbb{E}\big[\E{Y|\bX}^+ \big| \bX^{\bh} \big]
-
\mathbb{E}\big[\E{Y|\bX}^- \big| \bX^{\bh} \big].
\]
From here, it is easy to prove that the marginal distributions of $\nu_{\Phi}^+$ and $\nu_{\Phi}^-$ on the one-dimensional subspace generated by $\bX^{\bh}$ coincide if $\bh \in \Hil_0$. Since $\Hil_0$ has a positive $\mu$-measure, from Lemma \ref{T:cwgauss}, we obtain that these probability measures indeed coincide and, as a consequence, $V^{-1} \left(\E{Y | \bX}\right)^+=V^{-1} \left(\E{Y | \bX}\right)^-$ a.s., which trivially implies that $\E{Y|\bX}=0$ a.s.
\end{proof}

\subsection{Testing the linear model}
\label{ap:linmod}

\begin{proof}[Proof of Theorem \ref{Theo:DistPuntual}]
We analyse the asymptotic distribution of the three terms separately by invoking some auxiliary lemmas. Their proofs are collected in the supplementary material.\\

The asymptotic distribution of $T_{n,\bh}^1(x)$ follows from Corollary \ref{coro:single}:
$n^{-1/2}T_{n,\bh}^1(x)\allowbreak\inlaw\mathcal{N}(0,K_1(x,x))$. So, if $a_n=o(n^{-1/2})$, then $a_n T_{n,\bh}^1(x)=\op(1)$. The following two lemmas give insights into the asymptotic behaviour of $k_n$ and are required for the analysis of $T_{n,\bh}^2$ and $T_{n,\bh}^3$.

\begin{Lemm}
\label{Prop.1}
Under \ref{Assump:C4} and \ref{Assump:C7}, $k_n^3 (\log k_n)^2 = o(n^{1/2})$.
\end{Lemm}

\begin{Lemm} \label{Prop.Tn}
Under \ref{Assump:C4} and \ref{Assump:C7}, we have that $\nu[d_n = k_n] \conv 1$.
\end{Lemm}

We employ the decomposition (11) from page 338 in CMS to arrive at
\begin{align}
\hat{\brho} - \brho = \bL_n + \bY_n + \bS_n + \bR_n,\label{Eq.Decomp11}
\end{align}
where
$\bL_n:=-\sum_{j=k_n+1}^\infty\inprod{\brho}{\be_j}\be_j$, $\bY_n:=\sum_{j=1}^{k_n}(\inprod{\brho}{\hat\be_j}\hat\be_j-\inprod{\brho}{\be_j}\be_j)$, $\bS_n:=(\Gamma_n^\dag-\Gamma^\dag)\bU_n$, $\bR_n:=\Gamma^\dag\bU_n$, and $\bU_n:=\frac{1}{n}\sum_{i=1}^n\bX_i\otimes\varepsilon_i$. The decomposition (11) in CMS contains an extra term, $\mathbf{T}_n$, which is null here because of our construction of $\Gamma_n^\dagger$.\\

We will profusely employ the notation
\[
\Medh := \frac{1}{n} \sum_{i=1}^n \mathds{1}_{\lrb{\bX_i^\bh \leq x}}\bX_i.
\]
From \eqref{Eq.Decomp11}, the term $T_{n,\bh}^2(x)$ can be expressed as
\[
T_{n,\bh}^2(x)
=  n
\langle \Medh - \mathbf{E}_{x,\bh},  \bL_n + \bS_n +  \bY_n +  \bR_n\rangle.
\]
As a consequence of the following lemmas, we have that $T_{n,\bh}^2(x)=\op(n^{1/2})$.
\begin{Lemm}\label{LemmTn2.Ln}
Under \ref{Assump:C1} and \ref{Assump:C2},
$n^{1/2} \langle \Medh - \mathbf{E}_{x,\bh},  \bL_n \rangle = \op(1)$.
\end{Lemm}

\begin{Lemm}\label{LemmTn2.Rn}
Under \ref{Assump:C4}, \ref{Assump:C6}, and \ref{Assump:C7},
$n^{1/2} \langle \Medh - \mathbf{E}_{x,\bh},  \bR_n \rangle = \op(1)$.
\end{Lemm}

\begin{Lemm}\label{LemmTn2.Sn}
Under \ref{Assump:C3}, \ref{Assump:C4}, \ref{Assump:C6}, and \ref{Assump:C7},
$n^{1/2} \langle \Medh - \mathbf{E}_{x,\bh},  \bS_n \rangle = \op(1)$.
\end{Lemm}

\begin{Lemm}\label{LemmTn2.Yn}
Under \ref{Assump:C2}, \ref{Assump:C4}, \ref{Assump:C6}, and \ref{Assump:C7},
$n^{1/2} \langle \Medh - \mathbf{E}_{x,\bh},  \bY_n \rangle = \op(1)$.
\end{Lemm}

The behaviour of the third term, yielding statement \ref{Theo:DistPuntual:a}, is given by the next lemma.
\begin{Lemm}\label{Lemm.Tn3}
Under \ref{Assump:C1}, \ref{Assump:C2}, \ref{Assump:C4}, \ref{Assump:C5}, and \ref{Assump:C7}, $n^{-1/2}t_{n,\mathbf{E}_{x,\bh}}^{-1} T_{n,\bh}^3(x) \inlaw \mathcal{N}(0, \sigma^2_{\varepsilon})$.
\end{Lemm}
From the above results, $a_nT_{n,\bh}^2(x)=\op(1)$ for cases \ref{Theo:DistPuntual:b} and \ref{Theo:DistPuntual:c}, $T_{n,\bh}^3$ is the dominant term in \ref{Theo:DistPuntual:b}, and both $T_{n,\bh}^1$ and $T_{n,\bh}^3$ are dominant in \ref{Theo:DistPuntual:c}.
\end{proof}

\begin{proof}[Proof of Proposition \ref{Prop.tnx}]
By the definition of $t_{n,\mathbf{E}_{x,\bh}}$ and \eqref{Eq.Kar_loeve},
\[
t_{n,\mathbf{E}_{x,\bh}}^2 = \sum_{j=1}^{k_n}
\frac{\E{
 \mathds{1}_{\lrb{\bX^\bh \leq x}}
 \langle {\bX} , \be_j \rangle}^2}
 {\lambda_j}=\sum_{j=1}^{k_n}
 \E{\mathds{1}_{\lrb{\bX^\bh \leq x}}\xi_j}^2
  \leq
 \sum_{j=1}^{k_n}
  \E{
  \xi_j^2
}
 =k_n.
\]

We assume now that $\bX$ is Gaussian. Obviously, the two-dimensional random vector $(\xi_j,\bX^\bh)$, $j \in \N$, is centred normal. Moreover, the variance of $\xi_j$ is one and, if $h_j =\langle \bh ,\be_j\rangle$, then
$ \sigma^2_{\bh} =  \sum_{j=1}^\infty h_j^2\lambda_j < \infty$ (due to $\sum_{j=1}^\infty h_j^2 < \infty$ and \ref{Assump:C1}) and $\mathbb{C}\mathrm{ov}[\xi_j,\bX^\bh] =  h_j\lambda_j^{1/2}$. Notice that, if $\bh \neq {\bf 0}$, $\sigma_{\bh} >0$ since $\lambda_j>0$ for all $j\in\N$. Denoting by $\phi_{\bh}(u,v)$ the joint density function of $(\xi_j,\bX^\bh)$ and by $\phi_{\bh,2}(v)$ its second marginal, we have that
\begin{align*}
\E{
\mathds{1}_{\lrb{\bX^\bh \leq x}}
\xi_j}
=&\;
\int_{-\infty}^\infty
\left(\int_{-\infty}^\infty u \mathds{1}_{\{v \leq x\}} \frac{\phi_{\bh}(u,v)}{\phi_{\bh,2}(v)}\,\mathrm{d}u \right)\phi_{\bh,2}(v)\,\mathrm{d}v
\\
=&\;
\int_{-\infty}^x
\mathbb{E}\big[ \xi_j | \bX^\bh=v\big] \phi_{\bh,2}(v)\,\mathrm{d}v
\\
=&\;
\int_{-\infty}^x
\frac{h_j \sqrt{\lambda_j} v}{\sigma_{\bh}^2}
\phi_{\bh,2}(v)\,\mathrm{d}v
=
-\frac{h_j \sqrt{\lambda_j}}{\sigma_{\bh}}{\phi(x/\sigma_{\bh})}.
\end{align*}
This, the initial development, and \ref{Assump:C1} give us that
\[
t_{n,\mathbf{E}_{x,\bh}}^2 =
\frac{\phi^2(x/\sigma_{\bh})}{\sigma^2_{\bh}}
\sum_{j=1}^{k_n}  h_j ^2 \lambda_j
\conv
 {\phi^2(x/\sigma_{\bh})}.
\]
\end{proof}

\begin{proof}[Proof of Theorem \ref{Theo:tighness}]
We first prove \ref{Theo:tighness:a}. The joint asymptotic normality of $\big(T_{n,\bh}(x_1),\allowbreak\ldots,T_{n,\bh}(x_k)\big)$ for $\lrp{x_1,\ldots,x_k}\in\R^k$ follows by the Cram\'er--Wold device and the same arguments used in Lemma \ref{Lemm.Tn3}. Also, in the proof of that lemma it is shown that $n^{1/2}\inprod{\mathbf{E}_{x,\bh}}{\bL_n+\bY_n+\bS_n}=\op(1)$. Then, due to \eqref{Eq.Decomp11} and \eqref{Eq.Kar_loeve},
\begin{align*}
T_{n,\bh}(x)=&\;n^{-1/2}\big(T_{n,\bh}^1(x)+T_{n,\bh}^3(x)\big)\\
=&\;n^{-1/2}\sum_{i=1}^n{\mathds{1}_{\left\{\bX_i^{\bh}\le x\right\}}}\varepsilon_i+n^{-1/2}\sum_{i=1}^n{\inprod{\mathbf{E}_{x,\bh}}{\Gamma^\dag \bX_i}\varepsilon_i}+\op(1)\\
=&\;n^{-1/2}\sum_{i=1}^n\{A^i_x+B^i_x\}+\op(1),
\end{align*}
with $A_x^i:=\mathds{1}_{\left\{\bX_i^{\bh}\le x\right\}}\varepsilon_i$, and $B_x^i:=\sum_{j=1}^{k_n}\inprod{\mathbf{E}_{x,\bh}}{\be_j}\lambda_j^{-1/2}\xi_j^i\varepsilon_i$. Since the $\bX_i$'s and $\varepsilon_i$'s are i.i.d., and $\E{\varepsilon|\bX}=0$,
\begin{align*}
\mathbb{C}\mathrm{ov}\bigg[n^{-1/2}\sum_{i=1}^n&\{A^i_s+B^i_s\},n^{-1/2}\sum_{i'=1}^n\{A^{i'}_t+B^{i'}_t\}\bigg]\\
=&\;\E{A^1_sA^{1}_t}+\E{A^1_sB^{1}_t}+\E{B^1_sA^{1}_t}+\E{B^1_sB^{1}_t}.
\end{align*}
Applying the tower property with the conditioning variables $\bX^\bh$ (first expectation) and $\bX$ (second and third), it follows that
\begin{align*}
\E{A^1_sA^1_t}=&\;K_1(s,t),\\
\E{A^1_sB^1_t}=&\;\int_{\{\bx^\bh\leq s\}} \V{Y|\bX=\bx}\inprod{\mathbf{E}_{t,\bh}}{\Gamma^{\dagger}\bx}\,\mathrm{d}P_\bX(\bx),\\
\E{B^1_sB^1_t}=&\;\int \V{Y|\bX=\bx}\inprod{\mathbf{E}_{s,\bh}}{\Gamma^{\dagger}\bx}\inprod{\mathbf{E}_{t,\bh}}{\Gamma^{\dagger}\bx}\,\mathrm{d}P_\bX(\bx).
\end{align*}
Since $\Gamma^\dagger\conv\Gamma^{-1}$ in the operator norm $\norm{\cdot}_\infty$, Cauchy--Schwartz and $\big|\big|(\Gamma^{\dagger}-\allowbreak\Gamma^{-1})\bx\big|\big|\leq \norm{\Gamma^{\dagger}-\Gamma^{-1}}_\infty\norm{\bx}$ give that $\E{A^1_sB^1_t}-C_1(s,t)$ and $\E{B^1_sB^1_t}-C_2(s,t)$ converge to zero. The result then follows from Slutsky's theorem.\\

We now prove \ref{Theo:tighness:b}. The tightness of $n^{-1/2}T_{n,\bh}^1$ is obtained using the same arguments as in Theorem 1.1 of \cite{Stute1997}. For the tightness of $n^{-1/2}T_{n,\bh}^3$, define
\[
\bar T_{n,\bh}^3(u):=n\inprod{\E{\mathds{1}_{\lrb{U_\bh \leq u}} \bX}}{\brho-\hat \brho}, \quad U_\bh:=F_{\bh}(\bX^\bh),
\]
as the time-changed version of $\bar T_{n,\bh}^3$ by $F_{\bh}$, that is,
\[
T_{n,\bh}^3(x)=\bar T_{n,\bh}^3(F_{\bh}(x)).
\]
Consider $0\leq u_1<u<u_2\leq 1$ and the differences
\begin{align*}
n^{-1/2}\big(\bar T_{n,\bh}^3(u)-\bar T_{n,\bh}^3(u_1)\big)&=n^{1/2} \inprod{\E{\mathds{1}_{\lrb{u_1<U_\bh \leq u}} \bX}}{\brho-\hat \brho},\\
n^{-1/2}\big(\bar T_{n,\bh}^3(u_2)-\bar T_{n,\bh}^3(u)\big)&=n^{1/2} \inprod{\E{\mathds{1}_{\lrb{u<U_\bh \leq u_2}} \bX}}{\brho-\hat \brho}.
\end{align*}
Then, by the Cauchy--Schwartz and Jensen inequalities,
\begin{align*}
\mathbb{E}\Big[n^{-2}\big|&\bar T_{n,\bh}^3(u)-\bar T_{n,\bh}^3(u_1)\big|^2\abs{\bar T_{n,\bh}^3(u_2)-\bar T_{n,\bh}^3(u)}^2\Big]\\
&\leq n^{2} \mathbb{E}\Big[\norm{\E{\mathds{1}_{\lrb{u_1<U_\bh \leq u}} \bX}}^2  \norm{\E{\mathds{1}_{\lrb{u_1<U_\bh \leq u}} \bX}}^4\norm{\brho-\hat \brho}^4\Big]\\
&=n^{2} \mathbb{E}\big[\norm{\brho-\hat \brho}^4\big] \int \E{\bX(t)\mathds{1}_{\lrb{u_1<U_\bh\leq u}}}^2\,\mathrm{d}t \int \E{\bX(t)\mathds{1}_{\lrb{u<U_\bh\leq u_2}}}^2\,\mathrm{d}t\\
&\leq n^{2} \mathbb{E}\big[\norm{\brho-\hat \brho}^4\big]\int \E{\bX^2(t)\mathds{1}_{\lrb{u_1<U_\bh\leq u}}}\,\mathrm{d}t \int \E{\bX^2(t)\mathds{1}_{\lrb{u<U_\bh\leq u_2}}}\,\mathrm{d}t\\
&=n^{2} \mathbb{E}\big[\norm{\brho-\hat \brho}^4\big] \lrc{F(u)-F(u_1)}\lrc{F(u_2)-F(u)} \\
&\leq n^{2} \mathbb{E}\big[\norm{\brho-\hat \brho}^4\big]\lrc{F(u_2)-F(u_1)}^2\\
&\leq\lrc{G(u_2)-G(u_1)}^2,
\end{align*}
where $F(u):=\int \E{\bX^2(t)\mathds{1}_{\lrb{U_\bh\leq u}}}\,\mathrm{d}t$ and $G(u):=\sup_n\{n^{2}\mathbb{E}\big[\norm{\brho-\hat \brho}^4\big]\} F(u)$ are nondecreasing and continuous functions on $[0,1]$. This corresponds to employing $\gamma=2$ and $\alpha=1$ in Theorem 15.6 of \cite{Billingsley1968}, which gives the weak convergence of $n^{-1/2}\bar T_{n,\bh}^3$ in $D([0,1])$ and, as a consequence of the Continuous Mapping Theorem (CMT), $n^{-1/2}T_{n,\bh}^3\inlaw\boldsymbol{\mathcal{G}}_2$ in $D(\R)$.\\

Finally, we prove that $n^{-1/2}T_{n,\bh}^2\inprob \mathbf{0}$. Note first that, by Cauchy--Schwartz,
\begin{align*}
\sup_{x\in\R}\abs{n^{-1/2}T_{n,\bh}^2(x)} \leq&\;\sup_{x\in\R}\norm{\Medh - \mathbf{E}_{x,\bh}}n^{1/2}\norm{\hat{\brho}-\brho}.
\end{align*}
Assumption $\mathbb{E}\big[\norm{\brho-\hat{\brho}}^4\big]=\Oh\big(n^{-2}\big)$ implies $\norm{\hat{\brho}-\brho}=\OP(n^{-1/2})$. In addition, the weak law of large numbers in \Hil\ (e.g., \cite{hoffmann-jorgensen1976}) and \ref{Assump:C1} give $\Medh - \mathbf{E}_{x,\bh}\inprob\mathbf{0}$ in \Hil. Therefore, the CMT yields $\sup_{x\in\R}\allowbreak\big|n^{-1/2}T_{n,\bh}^2(x)\big|\inprob0$ and, as a consequence, $n^{-1/2}T_{n,\bh}^2\inprob\mathbf{0}$ in $D(\R)$.
\end{proof}

\begin{proof}[Proof of Corollary \ref{Coro:KSCvM}]
$\|T_{n,\bh}\|_\mathrm{KS}\inlaw\|\boldsymbol{\mathcal{G}}_2\|_\mathrm{KS}$ follows from the CMT. For the Cram\'er--von Mises norm, we use
\begin{align}
\|T_{n,\bh}\|_\mathrm{CvM}=\int_\R T_{n,\bh}(x)^2\,\mathrm{d}F_{\bh}(x)+\int_\R T_{n,\bh}(x)^2\,\mathrm{d}(F_{n,\bh}-F_{\bh})(x)\label{eq:cvm}
\end{align}
and $F_{n,\bh}- F_{\bh}\inprob\mathbf{0}$. By Slutsky's theorem, $(T_{n,\bh}, F_{n,\bh}-F_{\bh})\inlaw (\boldsymbol{\mathcal{G}}_2, \mathbf{0})$. Then, by the CMT, $\int_\R T_{n,\bh}(x)^2\,\mathrm{d}(F_{n,\bh}-F_{\bh})(x)\inlaw\mathbf{0}$ and $\int_\R T_{n,\bh}(x)^2\,\mathrm{d}F_{\bh}(x)\allowbreak\inlaw\int_\R \boldsymbol{\mathcal{G}}_2(x)^2\,\mathrm{d}F_{\bh}(x)$, which completes the proof.
\end{proof}

\fi

\ifsupplement

\newpage
\title{Supplement to ``Goodness-of-fit tests for the functional linear model based on randomly projected empirical processes''}
\setlength{\droptitle}{-1cm}
\predate{}%
\postdate{}%
\date{}

\author{
	Juan A. Cuesta-Albertos$^{1}$, Eduardo Garc\'ia-Portugu\'es$^{2,3,5}$,\\
	Manuel Febrero-Bande$^{4}$, and Wenceslao Gonz\'alez-Manteiga$^{4}$}

\footnotetext[1]{
Department of Mathematics, Statistics and Computer Science, University of Cantabria (Spain).}
\footnotetext[2]{
Department of Statistics, Carlos III University of Madrid (Spain).}
\footnotetext[3]{
UC3M-BS Institute of Financial Big Data, Carlos III University of Madrid (Spain).}
\footnotetext[4]{
Department of Statistics, Mathematical Analysis and Optimization, University of Santiago de Compostela (Spain).}
\footnotetext[5]{Corresponding author. e-mail: \href{mailto:edgarcia@est-econ.uc3m.es}{edgarcia@est-econ.uc3m.es}.}
\maketitle
\begin{abstract}
This supplement is organized as follows. Section \ref{ap:proofslemmas} proves the auxiliary lemmas used in the main results of the paper. Section \ref{ap:fursim} gives further details about the simulation study and contains extra results omitted in the paper. 
\end{abstract}
\begin{flushleft}
\small
	\textbf{Keywords:} Empirical process; Functional data; Functional linear model; Functional principal components; Goodness-of-fit; Random projections.
\end{flushleft}

\section{Proofs of the auxiliary lemmas}
\label{ap:proofslemmas}

Some general setting required for the proofs of the auxiliary lemmas is introduced as follows. We will use the notation
$\bX_{x,\bh}^i := \bX_{i} \mathds{1}_{\lrb{\bX_{i}^\bh \leq x}} - \mathbf{E}_{x,\bh}$, $i=1,\ldots, n$, and  $\Gamma_{z} := z I - \Gamma$ ($I$ is the identity operator in $\Hil$) with $z\in\mathbb{C}$. We also consider the linear operator $\Gamma_z^{-1/2}$, which is defined in CMS as the operator with the same eigenfunctions as $\Gamma$ and with $j$th eigenvalue equal to $(z-\lambda_j)^{-1/2}$ (the square root is taken in the complex space). We refer by $\xi_j^i$ to the $j$th random coefficient in the decomposition of $\bX_i$ in \eqref{Eq.Kar_loeve}. We also write $\bX^i_{x,\bh} = \sum_{j=1}^\infty  D_{x,\bh}^{i,j} \be_j$ for the expansion of $\bX^i_{x,\bh} $ in the basis of the eigenfunctions of $\Gamma$, so $D_{x,\bh}^{i,j} = \xi_j^{i}\mathds{1}_{\lrb{\bX_{i}^\bh \leq x}} -  \mathbb{E}\big[\xi_j^1\mathds{1}_{\lrb{\bX^\bh \leq x}}\big]$. We make use of the sets $\mathcal{B}_j$ (defined in page 339 in CMS), which are the oriented circles of the complex plane with centre $\lambda_j$ and radius $\delta_j/2$, and the functions $G_n(z)=\Gamma_z^{-1/2}(\Gamma_n-\Gamma)\Gamma_z^{-1/2}$, defined on page 351 in CMS. Finally, $\tilde f_n(z) := z^{-1} \mathds{1}_{\cup_j \mathcal{B}_j} (z)$ are analytic extensions of $f_n(x)= x^{-1}\mathds{1}_{\{x \geq c_n\}}$. A general constant $C$ (independent from $x$) will appear in the proofs, which may change from place to place. \\

We note that the assumptions stated in Subsection \ref{subsec:rho} could be slightly weakened in the case in which $\lim_n t_{n,\mathbf{E}_{x,\bh}}= \infty$. The analysis of the proofs shows how this can be done, depending on the speed of convergence of $t_{n,\mathbf{E}_{x,\bh}}$.

\begin{proof}[Proof of Lemma \ref{Prop.1}]
\ref{Assump:C4} implies that there exists a finite positive number $C$ such that $\lambda_n \leq C n^{-4} \log n $. If $i_n=\lfloor n^{1/7}\rfloor$, we have that
\[
\lambda_{i_n} \leq C \frac{\log n}{i_n^{4}}= o(n^{-1/2}).
\]
Then, by the definition of $k_n$ and \ref{Assump:C7}, we have that $k_n = \Oh(n^{1/7})$ and, consequently,
\[
k_n^6 (\log k_n)^4 = \Oh\lrp{n^{6/7} (\log n)^4} = o(n).
\]
\end{proof}

\begin{proof}[Proof of Lemma \ref{Prop.Tn}]
Let us define the set
\[
\mathcal{A}_n :=\left\{\omega\in\Omega: \sup_{j \leq k_n+1} \frac{|\hat \lambda_j (\omega)- \lambda_j |}{\delta_j} < \frac 1 2 \right\}.
\]
If $\omega_0 \in \mathcal{A}_n $, we have that
\[
\hat \lambda_{k_n+1}(\omega_0) \leq \lambda_{k_n+1} +\frac 1 2 \delta_{k_n+1}  < c_n
\]
and, consequently, $d_n(\omega_0) \leq k_n$. Moreover, from here, if we denote $\mathcal{A}_n^j:=\mathcal{A}_n\cap \{d_n=j\}$, then $\mathcal{A}_n^j=\emptyset$ for every $j>k_n$. Therefore, $\mathcal{A}_n = \cup_{j=1}^{k_n} \mathcal{A}_n^j$ and, for every $j=1,\ldots, k_n$, if $\omega \in \mathcal{A}_n^j$, we have that
\[
 \lambda_{j} + \frac{\delta_{j}}{2} >   \hat \lambda_{j}(\omega) = \hat \lambda_{d_n}(\omega) \geq c_n.
\]
Thus, if $\omega \in \mathcal{A}_n $, $k_n \geq d_n(\omega)$ and then $d_n(\omega)= k_n$. However, under \ref{Assump:C4} and \ref{Assump:C7}, Lemma \ref{Prop.1} holds and, in particular, $k_n^2 \log k_n = o(n^{1/2})$. The proof of Lemma 5 in CMS shows that $\nu [\mathcal{A}_n ] \conv 1$, which completes the proof.
\end{proof}

\begin{proof}[Proof of Lemma \ref{LemmTn2.Ln}]
By the definition of $\bL_n$, we have that
\[
  \langle \Medh - \mathbf{E}_{x,\bh},  \bL_n \rangle
=
-\sum_{j=k_n+1} ^\infty
\langle \brho , \be_j \rangle
\langle \Medh-\mathbf{E}_{x,\bh},\be_j \rangle,
  \]
  and, then, taking into account that the $\bX_i$'s are i.i.d.,
 \begin{align*}
\E{\langle \Medh - \mathbf{E}_{x,\bh},  \bL_n \rangle^2}
=&\;
\frac 1 {n}  \mathbb{E}\bigg[  \sum_{j=k_n+1} ^\infty
\langle \brho , \be_j \rangle
\langle \mathds{1}_{\lrb{\bX_1^\bh \leq x}} \bX_1-\mathbf{E}_{x,\bh},\be_j \rangle
\bigg]^2
\\
\leq&\;
\frac 1 {n}  \mathbb{E}\bigg[ \sum_{j=k_n+1} ^\infty
| \langle \brho , \be_j \rangle |
\| \mathds{1}_{\lrb{\bX_1^\bh \leq x}} \bX_1-\mathbf{E}_{x,\bh}\|
\bigg]^2
\\
\leq&\;
\frac 1 {n} \E{\| \bX_1\|^2}  \bigg(\sum_{j=k_n+1} ^\infty
|\langle \brho , \be_j \rangle| \bigg)^2,
 \end{align*}
where we have used that
\begin{align*}
\E{\| \mathds{1}_{\lrb{\bX_1^\bh \leq x}} \bX_1-\mathbf{E}_{x,\bh}\|^2}
\leq&\;
\E{\| \mathds{1}_{\lrb{\bX_1^\bh \leq x}} \bX_1\|^2}
\leq \E{\| \bX_1\|^2}.
\end{align*}
Then the result follows from \ref{Assump:C1}, \ref{Assump:C2}, and Chebyshev's inequality.
\end{proof}

\begin{proof}[Proof of Lemma \ref{LemmTn2.Rn}]
By the definition of $\bR_n$, we have that $\langle \bR_n ,\bx \rangle = n^{-1}\sum_{i=1}^n\allowbreak \langle \Gamma^\dagger \bX_i, \bx\rangle \varepsilon_i$. Then
\[
\langle \Medh - \mathbf{E}_{x,\bh},  \bR_n \rangle
 =
 \frac 1 {n} \sum_{i=1}^n  \langle \Gamma^\dagger \bX_i,  \Medh - \mathbf{E}_{x,\bh}\rangle \varepsilon_i.
\]

The $\varepsilon_i$'s are i.i.d. and independent from the rest of the involved quantities. Therefore:
\begin{align}
\nonumber
\mathbb{E}\big[\langle \Medh &- \mathbf{E}_{x,\bh},  \bR_n \rangle ^2\big]
\\
\nonumber
=&\;
  \frac {\sigma_\varepsilon^2} {n} \E{\langle \Gamma^\dagger \bX_1,  \Medh - \mathbf{E}_{x,\bh}\rangle^2}
\\
\nonumber
=&\;
  \frac {\sigma_\varepsilon^2} {n^3}      \E{\bigg(\sum_{i'=1}^n \langle \Gamma^\dagger \bX_1, \bX_{i'} \mathds{1}_{\lrb{\bX_{i'}^\bh \leq x}} - \mathbf{E}_{x,\bh}\rangle\bigg)^2}
\\
\label{Eq.tnx2_Rn_2}
=&\;
  \frac {\sigma_\varepsilon^2} {n^3}      \mathbb{E}\Bigg[\bigg(\sum_{i'=1}^n
   \sum_{j=1}^{k_n}\xi_j^1\left(\xi_j^{i'}\mathds{1}_{\lrb{\bX_{i'}^\bh \leq x}}- \E{\xi_j^{1}\mathds{1}_{\lrb{\bX_{1}^\bh \leq x}}}\right)\bigg)^2\Bigg].
\end{align}
We need to compute the sum of the expectations in \eqref{Eq.tnx2_Rn_2}. This sum contains $n^2k_n^2$ terms, and we analyse their possible behaviours next.\\

Let $i', i''=1,\ldots, n$. If $i'\neq i''$, then either $i'\neq 1$ or $i''\neq 1$. Assuming the first holds, we have that $\big(\xi_j^{i'}\mathds{1}_{\lrb{\bX_{i'}^\bh \leq x}}- \mathbb{E}\big[\xi_j^{1}\mathds{1}_{\lrb{\bX_{1}^\bh \leq x}} \big]\big)$ is independent of $\big(\xi_j^{i''}\mathds{1}_{\lrb{\bX_{i'}^\bh \leq x}}- \mathbb{E}\big[\xi_j^{1}\mathds{1}_{\lrb{\bX_{1}^\bh \leq x}} \big]\big)$, $\xi_j^{1}$, and $\xi_{j'}^{1}$. Therefore,
\begin{align*}
\mathbb{E}\Big[\xi_j^1&\left(\xi_j^{i'}\mathds{1}_{\lrb{\bX_{i'}^\bh \leq x}}- \E{\xi_j^{1}\mathds{1}_{\lrb{\bX_{1}^\bh \leq x}}}\right)\\
&\times\xi_{j'}^1\left(\xi_{j'}^{i''}\mathds{1}_{\lrb{\bX_{i''}^\bh \leq x}}- \E{\xi_j^{1}\mathds{1}_{\lrb{\bX_{1}^\bh \leq x}}}\right)\Big]
 =
0,
\end{align*}
and we have $k_n^2n(n-1)$  terms in \eqref{Eq.tnx2_Rn_2} whose expectations are also zero.
With respect to the remaining terms, we can elaborate a bit more on  \eqref{Eq.tnx2_Rn_2}  to obtain
\begin{align*}
\E{\langle \Medh - \mathbf{E}_{x,\bh},  \bR_n \rangle ^2}
&\leq
  \frac {\sigma_\varepsilon^2} {n^3}      \E{\bigg(\sum_{i'=1}^n
   \sum_{j=1}^{k_n}|\xi_j^1|\left(|\xi_j^{i'}| +  \E{|\xi_j^{1}|}\right)\bigg)^2}
     \\
  &=:
     \frac {\sigma_\varepsilon^2} {n^3}      \E{\bigg(\sum_{i'=1}^n    \sum_{j=1}^{k_n} T(i',j)\bigg)^2}.
\end{align*}
The expansion of the square in the last expression yields the terms $\mathbb{E}\big[T(i',j)\allowbreak T(i'',j')\big]$. Each of these can be bounded by $M$ in \ref{Assump:C6}, since they are the sum of four terms of the form $\Ebig{|\xi_j^1| |\xi_j^{i'}| |\xi_{j'}^1| |\xi_{j'}^{i''}|}$, $\Ebig{|\xi_j^1| |\xi_j^{i'}| |\xi_{j'}^1|} \Ebig{|\xi_j^{1}|}$, or $\Ebig{|\xi_j^1| |\xi_{j'}^{1}|} \Ebig{|\xi_j^{1}|} \Ebig{|\xi_{j'}^{1}|}$, with $i',i'' = 1,\ldots, n$ and $j,j' = 1,\ldots , k_n$. If we apply Cauchy--Schwartz's and Jensen's inequalities, we have that, for instance,
\begin{align*}
\E{|\xi_j^1| |\xi_{j'}^{1}|} \E{|\xi_j^{1}|} \E{|\xi_{j'}^{1}|} &\leq
\left(\E{|\xi_j^1|^4}\E{|\xi_{j'}^{1}|^4}\right)^{1/4} \E{|\xi_j^{1}| ^4}^{1/4}\E{|\xi_{j'}^{1}| ^4}^{1/4}\\
& \leq M.
\end{align*}

The remaining terms are handled similarly, and we can conclude that the non-null terms in  \eqref{Eq.tnx2_Rn_2} are bounded by $M$. Thus, we obtain that
\[
\E{\langle \Medh - \mathbf{E}_{x,\bh},  \bR_n \rangle}^2
\leq\frac {\sigma_\varepsilon^2}{n^3}M\left(n^2 k_n^2  -k_n^2 n(n-1)\right)=\frac {\sigma_\varepsilon^2}{n^2}M k_n^2.
\]
Consequently, Chebyshev's inequality and Lemma \ref{Prop.1} give the result.
\end{proof}

\begin{proof}[Proof of Lemma \ref{LemmTn2.Sn}]
The proof follows the steps of the proof of Proposition 3 in CMS, but replaces the $\bX_{n+1}$ term by the difference $ \Medh - \mathbf{E}_{x,\bh}$. There are some technical differences as well because, here, some involved terms are not independent.\\

The proof in the beginning of Proposition 3 in CMS allows us to conclude that
\begin{align}
\left| \langle \Medh - \mathbf{E}_{x,\bh},  \bS_n \rangle  \right| \leq C \sum_{j=1}^{k_n} H_{j,n}, \label{Eq.tnx2_0}
\end{align}
with
\[
H_{j,n}  \leq
C \int_{\mathcal{B}_j} \big|\tilde f_n(z)\big| \| G_n(z) \|_\infty \big\|  \Gamma_z^{-1/2}\bU_n\big\| \big\| \Gamma_z^{-1/2}(\Medh - \mathbf{E}_{x,\bh}) \big\|\, \mathrm{d}z.
\]
The application of Cauchy--Schwartz's inequality twice, plus Lemma 3 in CMS (bounds $\sup_{z \in \mathcal{B}_j}\allowbreak\E{\|G_n(z)\|_\infty^2}$), and some arguments developed in Proposition 3 in CMS, give us that
\begin{align}
\nonumber
\E{H_{j,n}} \leq&\;
C \int_{\mathcal{B}_j} \big|\tilde f_n(z)\big|
\left(\E{\| G_n(z) \|_\infty^2} \right)^{1/2} \nonumber \\
&\times \left( \E{\|  \Gamma_z^{-1/2}\bU_n\|^4} \E{\|\Gamma_z^{-1/2}(\Medh - \mathbf{E}_{x,\bh})\| ^4} \right)^{1/4}\, \mathrm{d}z\nonumber
\\
\nonumber
\leq&\;
C \text{diam}(\mathcal{B}_j) \frac{j \log j}{\sqrt{n}} \nonumber\\
&\times  \sup_{z \in \mathcal{B}_j}\left\{ \big|\tilde f_n(z)\big| \left( \E{\|  \Gamma_z^{-1/2}\bU_n\|^4} \E{\|\Gamma_z^{-1/2}(\Medh - \mathbf{E}_{x,\bh})\| ^4} \right)^{1/4} \right\}\nonumber
\\
\label{Eq.tnex2.Sn1}
\leq&\;
C  \frac{ j \log j}{\sqrt{n}} \nonumber\\
&\times\sup_{z \in \mathcal{B}_j} \left\{  \left(
\E{\|  \Gamma_z^{-1/2}\bU_n\|^4}
\E{\|\Gamma_z^{-1/2}(\Medh - \mathbf{E}_{x,\bh})\| ^4}
\right)^{1/4} \right\} .
\end{align}

Let us analyse the two expectations included in \eqref{Eq.tnex2.Sn1}. First, by the definition of $\bU_n$, we have that
\begin{align}
\nonumber
\E{\|  \Gamma_z^{-1/2}\bU_n \|^4}=&\; \frac 1 {n^4} \sum_{r,s,r^\prime,s^\prime=1}^n
\E{\langle\Gamma_z^{-1/2}\bX_r,\Gamma_z^{-1/2}\bX_s\rangle\langle \Gamma_z^{-1/2}\bX_{r'},\Gamma_z^{-1/2}\bX_{s'}\rangle}\nonumber\\
&\times \E{\varepsilon_{r}\varepsilon_{s}\varepsilon_{r'}\varepsilon_{s'}}.\label{Eq.tnex2.Sn2_1}
\end{align}
This sum contains $n^4$ terms. However, the $\varepsilon$'s are centred, independent variables, so $\E{\varepsilon_{r}\varepsilon_{s}\varepsilon_{r'}\varepsilon_{s'}}$ equals zero unless the vector $(r,r',s,s')$ contains two pairs of identical components. This only happens at most in $n + \frac 1 2 \binom{4}{2} n(n-1)=3n^2 -2n$ terms, in which $\E{\varepsilon_{r}\varepsilon_{s}\varepsilon_{r'}\varepsilon_{s'}}=1$. Let us compute the value of the other involved expectation on those terms. We have that
\[
\langle  \Gamma_z^{-1/2}\bX_r,  \Gamma_z^{-1/2}\bX_s\rangle
=
\sum_{i=1}^\infty \frac{\lambda_i}{|z-\lambda_i|}\xi_i^r\xi_i^s
\]
and, then,
\begin{align*}
\nonumber
\Big| \mathbb{E}\Big[\langle  \Gamma_z^{-1/2}\bX_r,  \Gamma_z^{-1/2}\bX_s\rangle&\langle  \Gamma_z^{-1/2}\bX_{r'},  \Gamma_z^{-1/2}\bX_{s'}\rangle\Big] \Big|
\\
\leq&\;
\sum_{i=1}^\infty\sum_{i`=1}^\infty \frac{\lambda_i}{|z-\lambda_i|}\frac{\lambda_{i'}}{|z-\lambda_{i'}|}\E{\left| \xi_i^r\xi_i^s\xi_{i'}^{r'}\xi_{i'}^{s'}\right|}
\\
\leq&\;
M \sum_{i=1}^\infty\sum_{i`=1}^\infty \frac{\lambda_i}{|z-\lambda_i|}\frac{\lambda_{i'}}{|z-\lambda_{i'}|},
\end{align*}
where we have applied Cauchy--Schwartz's inequality twice and \ref{Assump:C6}.
Those findings, replaced in \eqref{Eq.tnex2.Sn2_1}, give us that
\begin{align} \label{Eq.tnex2.Sn2_2}
\E{\|  \Gamma_z^{-1/2}\bU_n \|^4}
\leq \frac{3n^2-2n}{n^4} M
\left(
\sum_{i=1}^\infty \frac{\lambda_i}{|z-\lambda_i|}
\right)^2.
\end{align}

We analyse the last factor in \eqref{Eq.tnex2.Sn1}.
We have that
\[
\Gamma_z^{-1/2}(\Medh - \mathbf{E}_{x,\bh}) =
\frac 1 n \sum_{r=1}^n  \Gamma_z^{-1/2}  \bX_{x,\bh}^r
=
\frac 1 n  \sum_{i=1}^\infty \sqrt{\frac{\lambda_i}{z-\lambda_i}} \left(\sum_{r=1}^n D_{x,\bh}^{r,i} \right) \be_i.
\]
Therefore,
\begin{align}
\label{Eq.tnex2.Sn2_2.2}
\big \|\Gamma_z^{-1/2}(\Medh - \mathbf{E}_{x,\bh}) \big\|^2
=
\frac 1 {n^2} \sum_{i=1}^\infty \frac{\lambda_i}{|z-\lambda_i|}\left(  \sum_{r=1}^n  D_{x,\bh}^{r,i} \right)^2
\end{align}
and
\begin{align}
\nonumber
\mathbb{E}\Big[\big \|\Gamma_z^{-1/2}&(\Medh - \mathbf{E}_{x,\bh}) \big\|^4\Big]
\\
\label{Eq.tnex2.Sn2_3}
=&\;
\frac 1{n^4} \sum_{i,i'=1}^\infty \frac{\lambda_i}{|z-\lambda_i|}\frac{\lambda_{i'}}{|z-\lambda_{i'}|}\sum_{r,r's,s'=1}^n \E{D_{x,\bh}^{r,j}D_{x,\bh}^{r',j}D_{x,\bh}^{s,j'}D_{x,\bh}^{s',j'}}.
\end{align}
We are in a similar situation to that in \eqref{Eq.tnex2.Sn2_1}, and it happens that all expectations here are zero except, at most, $3n^2-2n$ of them. Moreover, those non-null terms can be bounded. Applying the Cauchy--Schwartz inequality twice,
\begin{align}
\nonumber
\Big| \mathbb{E}\Big[D_{x,\bh}^{r,j}
D_{x,\bh}^{r',j}
&D_{x,\bh}^{s,j'}
D_{x,\bh}^{s',j}
\Big]\Big|
\\
\nonumber
\leq&\;
\E{\big| D_{x,\bh}^{r,j}
D_{x,\bh}^{r',j}
D_{x,\bh}^{s,j'}
D_{x,\bh}^{s',j}\big|
}
\\
\leq&\;
 \label{Eq.tnex2.Sn5_2}
\left(\E{\big(D_{x,\bh}^{r,j} \big)^4}
\E{\big(D_{x,\bh}^{r',j} \big)^4}
\E{\big(D_{x,\bh}^{s,j'} \big)^4}
\E{\big(D_{x,\bh}^{s',j} \big)^4}
\right)^{1/4}.
\end{align}

And, given $j\in\N $ and $r=1,\ldots, n$,
\[
\E{\big( D_{x,\bh}^{r,j} \big)^4}
=
\E{\big(\xi_{j}\mathds{1}_{\lrb{\bX_{j}^\bh \leq x}} - \mathbf{E}_{x,\bh} ^{j}\big)^4}
\leq
\E{\left(|\xi_{j}| + \E{|\xi_j|}\right)^4},
\]
which gives a fourth-order polynomial whose terms are $\binom{4}{s} \E{|\xi_{j}|^s}(\E{|\xi_j|})^{4-s}$ for $s=0,\ldots, 4$. If we apply Jensen's inequality and \ref{Assump:C6} to each of these, we have that
$$\E{|\xi_{j}|^s}(\E{|\xi_j|})^{4-s}
\leq
\E{|\xi_{j}|^4}^{s/4}(\E{|\xi_j|^4})^{(4-s)/4} =  \E{|\xi_{j}|^4} \leq M.
$$
This and \eqref{Eq.tnex2.Sn5_2} give
$\big| \Ebig{   D_{x,\bh}^{r,j}
D_{x,\bh}^{r',j}
D_{x,\bh}^{s,j'}
D_{x,\bh}^{s',j}
}\big|
\leq
2^4 M$, and, with \eqref{Eq.tnex2.Sn2_3}, yield
\[
\E{\big \|\Gamma_z^{-1/2}(\Medh - \mathbf{E}_{x,\bh}) \big\|^4}
\leq
\frac{3n^2-2n}{n^4} 2^4M
\left(
\sum_{i=1}^\infty \frac{\lambda_i}{|z-\lambda_i|}
\right)^2.
\]

If we replace this bound and \eqref{Eq.tnex2.Sn2_2} in \eqref{Eq.tnex2.Sn1}, we obtain that
\[
\E{H_{j,n}}
\leq
C  \frac{ j \log j}{\sqrt{n}}
\frac{(3n^2-2n)^2}{n^8}
\sup_{z \in \mathcal{B}_j} \left\{
\sum_{i=1}^\infty \frac{\lambda_i}{|z-\lambda_i|}
\right\}^4 .
\]
\ref{Assump:C3} allows us to apply Lemmas 1 and 2 in CMS, which gives us that $ \sup_{z\in \mathcal{B}_j}\allowbreak \lrb{\sum_{i=1}^\infty \frac {\lambda_i}{|z - \lambda_i|}} \leq C j \log j$ and, then, that
\[
\E{H_{j,n}}
\leq
C   (j \log j)^5
\frac{(3n^2-2n)^2}{n^{17/2}}.
\]

This and \eqref{Eq.tnx2_0} give us
\begin{align*}Ê
\sqrt{n} \E{\left| \langle \Medh - \mathbf{E}_{x,\bh},  \bS_n \rangle  \right|}
\leq&\;C \frac{(3n^2-2n)^2}{n^{8}} \sum_{j=1}^{k_n}  (j \log j)^5
\\
\leq&\;
C \frac{(3n^2-2n)^2}{n^{8}} k_n^6 (\log k_n)^5 = o(n^{-2})
\end{align*}
due to \ref{Assump:C4}, which proves the lemma.
\end{proof}

\begin{proof}[Proof of Lemma \ref{LemmTn2.Yn}]
  This term is handled following the scheme in Proposition 2 in CMS. Using the notation in that proposition, we have that $\bY_n = \mathcal{S}_n\brho + \mathcal{R}_n\brho$, where
 \begin{align*}
 \mathcal{S}_n:=&\;\frac{1}{2\i\pi}\sum_{j=1}^{k_n}\int_{\mathcal{B}_j}\lrc{(zI-\Gamma)^{-1}(\Gamma_n-\Gamma)(zI-\Gamma)^{-1}}\, \mathrm{d}z,\\
  \mathcal{R}_n:=&\;   \frac{1}{2\i\pi}\sum_{j=1}^{k_n}\int_{\mathcal{B}_j}\lrc{(zI-\Gamma)^{-1}(\Gamma_n-\Gamma)(zI-\Gamma)^{-1}(\Gamma_n-\Gamma)(zI-\Gamma_n)^{-1}}\, \mathrm{d}z.\\
 \end{align*}
  The Cauchy--Schwartz inequality gives
  \begin{align}
    \nonumber
 \mathbb{E}\big[|\langle  \mathcal{S}_n  \brho, \Medh &- \mathbf{E}_{x,\bh}\rangle|\big]
 \\
 \label{Eq.tnx2.Yn.1}
 \leq&\;
  \E{{  \|  \mathcal{S}_n  \brho\| \| \Medh - \mathbf{E}_{x,\bh}\|}} \nonumber\\
 \leq&\;
  \left(
   \E{\|  \mathcal{S}_n  \brho\| ^2}  \E{\| \Medh - \mathbf{E}_{x,\bh}\|^2}
  \right)^{1/2}.
  \end{align}

  On the other hand, it happens that
  \begin{align*}
  \E{\| \Medh - \mathbf{E}_{x,\bh}\|^2}
\leq&\;
\frac{1}{n^2}  \E{
  \bigg\| \sum_{r=1}^n \bX^r_{x,\bh} \bigg\|^2
}
  =
   \frac{1}{n^2} \E{\bigg\|
\sum_{l=1}^\infty \bigg(\sum_{r=1}^n D_{x,\bh}^{r,l} \bigg)\be_i\bigg\|^2
}
\\
 =&\;
   \frac{1}{n^2}
   \sum_{l=1}^\infty \E{
 \bigg(\sum_{r=1}^n D_{x,\bh}^{r,l} \bigg)^2
}
=
\frac{1}{n^2}
\sum_{l=1}^\infty \sum_{r=1}^n \E{\big(D_{x,\bh}^{r,l} \big)^2}
\\
=&\;
\frac{1}{n}
\sum_{l=1}^\infty \E{\big(D_{x,\bh}^{1,l} \big)^2}
\leq
\frac{1}{n}   \sum_{l=1}^\infty
\E{\xi_l^2}
= \frac{1}{n}  \sum_{l=1}^\infty \lambda_l.
\end{align*}

From here and \eqref{Eq.tnx2.Yn.1}, we have
\begin{align}
\E{|\langle  \mathcal{S}_n  \brho, \Medh - \mathbf{E}_{x,\bh}\rangle|}
\leq&\;
\frac C {\sqrt{n}}
\left(
\E{\|\mathcal{S}_n  \brho\| ^2}
\right)^{1/2}\nonumber\\
=&\;
\frac C {\sqrt{n}}
\left(
\sum_{l=1}^\infty   \E{\langle\mathcal{S}_n \brho, \be_l \rangle^2}
\right)^{1/2}.  \label{Eq.tnx2.Yn.2}
\end{align}
On pages 347 and 348 in CMS a reasoning is developed which gives the next two bounds:
\[
 \E{\langle\mathcal{S}_n \brho, \be_l \rangle^2}
\leq \left\{
\begin{array}{ll}
\frac{M}{n} \left(
    \sum_{l'=k_n+1}^\infty \langle \brho, \be_{l'}\rangle
    \frac{\sqrt{\lambda_l\lambda_{l'}}}{\lambda_{l'}- \lambda_{l}}
    \right)^2, & \text{if } l \leq k_n,
\\
\frac M n \left(
    \sum_{l'=1}^{k_n} \langle \brho, \be_{l'}\rangle
    \frac{\sqrt{\lambda_l\lambda_{l'}}}{\lambda_{l'}- \lambda_{l}}
    \right)^2, & \text{if } l > k_n.
\end{array}
\right.
\]
This and \eqref{Eq.tnx2.Yn.2} give
\begin{align}
\label{Eq.tnx2.Yn.3}
n\Bigg(\E{|\langle  \mathcal{S}_n  \brho, \Medh - \mathbf{E}_{x,\bh}\rangle|}\Bigg)^2
\leq&\;
\frac{C}{n}\sum_{l=1}^{k_n}
\Bigg(
\sum_{l'=k_n+1}^\infty \langle \brho, \be_{l'}\rangle
\frac{\sqrt{\lambda_l\lambda_{l'}}}{\lambda_{l'}- \lambda_{l}}
\Bigg)^2
\\
\label{Eq.tnx2.Yn.4}
&
+   \frac{C}{{n}} \sum_{l=k_n+1}^\infty
\Bigg(
\sum_{l'=1}^{k_n} \langle \brho, \be_{l'}\rangle
\frac{\sqrt{\lambda_l\lambda_{l'}}}{\lambda_{l'}- \lambda_{l}}
\Bigg)^2.
\end{align}

Lemma 1 in CMS  applied to the term in  \eqref{Eq.tnx2.Yn.3} leads  to
\begin{align*}
\frac{C}{n}\sum_{l=1}^{k_n}
\Bigg(
\sum_{l'=k_n+1}^\infty \langle \brho, \be_{l'}\rangle
\frac{\sqrt{\lambda_l\lambda_{l'}}}{\lambda_{l'}- \lambda_{l}}
\Bigg)^2
\leq&\;
\frac{C}{n}\sum_{l=1}^{k_n}
\Bigg(
\sum_{l'=k_n+1}^\infty
\sqrt{\frac{\lambda_{l'}}{\lambda_{l}}}\frac{|\langle \brho, \be_{l'}\rangle |}{1- \frac{l}{l'}}
\Bigg)^2
\\
\leq&\;
\frac{C}{n}\sum_{l=1}^{k_n}
\Bigg(
\sum_{l'=k_n+1}^\infty
\frac{|\langle \brho, \be_{l'}\rangle |}{1- \frac{l}{l'}}
\Bigg)^2
\\
\leq&\;
\frac{2C}{n}\sum_{l=1}^{k_n}
\Bigg(
\sum_{l'=k_n+1}^{k_n+h_n}
\frac{|\langle \brho, \be_{l'}\rangle |}{1- \frac{1}{l'}}
\Bigg)^2
\\
&
+
\frac{2C}{n}\sum_{l=1}^{k_n}
\Bigg(
\sum_{l'=k_n+h_n+1}^\infty |\langle \brho, \be_{l'}\rangle |
\frac{1}{1- \frac{l}{l'}}
\Bigg)^2,
\end{align*}
where $h_n=\big\lfloor\sqrt{\frac{k_n}{\log k_n}}\big\rfloor$. From this definition we obtain that the second term satisfies
\begin{align*}
\frac{2C}{n}\sum_{l=1}^{k_n}
\Bigg(
\sum_{l'=k_n+h_n+1}^\infty &
\frac{|\langle \brho, \be_{l'}\rangle |}{1- \frac{l}{l'}}
\Bigg)^2\\
\leq&\;
\frac{2C}{n}\sum_{l=1}^{k_n}
\Bigg(
\sum_{l'=k_n+h_n+1}^\infty | \langle \brho, \be_{l'}\rangle| (1+ \sqrt{k_n \log k_n})
\Bigg)^2
\\
\leq&\;
\frac{8C}{n}\sum_{l=1}^{k_n} k_n \log k_n
\Bigg(
\sum_{l'=k_n+1}^\infty | \langle \brho, \be_{l'}\rangle|
\Bigg)^2
\\
\leq&\;
\frac{8C}{n}k_n^2 \log k_n
\Bigg(
\sum_{l'=k_n+1}^\infty | \langle \brho, \be_{l'}\rangle| \Bigg)^2\to0,
\end{align*}
where the convergence follows from \ref{Assump:C2} and Lemma \ref{Prop.1}. On the other hand, the first term verifies that
\begin{align*}
\frac{2C}{n}\sum_{l=1}^{k_n}
\Bigg(
\sum_{l'=k_n+1}^{k_n+h_n}
\frac{|\langle \brho, \be_{l'}\rangle |}{1- \frac{l}{l'}}
\Bigg)^2
\leq&\;
\frac{2C}{n} k_n  h_n^2 \max_{\substack{k_n < l' \leq k_n+h_n\\ 1 \leq l \leq k_n}}
\Bigg\{
\frac{|\langle \brho, \be_{l'}\rangle |}{\left| 1- \frac{l}{l'} \right|}
\Bigg\}^2
\\
\leq&\;
\frac{2C}{n} k_n  \frac{k_n}{\log k_n} (k_n+h_n)^2 \max_{k_n < l'}
\left\{
|\langle \brho, \be_{l'}\rangle |^2
\right\}
\\
=&\;
 8C  \frac{k_n^4}{n \log k_n}  \max_{k_n < l'}
\left\{
|\langle \brho, \be_{l'}\rangle |^2
\right\}
\conv 0
\end{align*}
due to \ref{Assump:C2} and Lemma \ref{Prop.1}. \\

Then the term in \eqref{Eq.tnx2.Yn.3} converges to zero. Let us analyse the term in \eqref{Eq.tnx2.Yn.4}. As before, Lemma 1 in CMS gives us
\begin{align}
\nonumber
\frac{C}{{n}} \sum_{l=k_n+1}^\infty
 \left(
\sum_{l'=1}^{k_n} \langle \brho, \be_{l'}\rangle
\frac{\sqrt{\lambda_l\lambda_{l'}}}{\lambda_{l'}- \lambda_{l}}
\right)^2
\leq&\;
 \frac{C}{{n}} \sum_{l=k_n+1}^\infty
 \left(
\sum_{l'=1}^{k_n}  \sqrt{\frac{\lambda_l}{\lambda_{l'}}}
\frac{|\langle \brho, \be_{l'}\rangle|}{1- \frac{l`}{l}}
\right)^2
\\
\nonumber
\leq&\;
\frac{C}{{n}} \frac 1 {\lambda_{k_n}}\sum_{l=k_n+1}^\infty \lambda_l
 \left(
\sum_{l'=1}^{k_n}
\frac{|\langle \brho, \be_{l'}\rangle|}{1- \frac{l`}{l}}
\right)^2
\\
\nonumber
\leq&\;
\frac{2C}{{n}} \frac 1 {\lambda_{k_n}} \sum_{l=k_n+1}^{k_n+h_n} \lambda_l
 \left(
\sum_{l'=1}^{k_n}
\frac{|\langle \brho, \be_{l'}\rangle|}{1- \frac{l`}{l}}
\right)^2
\\
\label{Eq.tnx2.Yn.5}
& +    \frac{2C}{{n}} \frac 1 {\lambda_{k_n}} \sum_{l=k_n+h_n+1}^\infty \lambda_l
 \left(
\sum_{l'=1}^{k_n}
\frac{|\langle \brho, \be_{l'}\rangle|}{1- \frac{l`}{l}}
\right)^2.
\end{align}

We have that
\begin{align*}
\sum_{l=k_n+h_n+1}^\infty \lambda_l&
 \Bigg(
\sum_{l'=1}^{k_n} |\langle \brho, \be_{l'}\rangle|
\frac{1}{1- \frac{l`}{l}}
\Bigg)^2\\
\leq&\;
\sum_{l=k_n+h_n+1}^\infty \lambda_l
 \Bigg(
\sum_{l'=1}^{k_n} |\langle \brho, \be_{l'}\rangle|
\left(1+ \frac{k_n}{h_n}  \right)
\Bigg)^2
\\
\leq&\;
4 k_n \log k_n    \Bigg( \sum_{l'=1}^{\infty} |\langle \brho, \be_{l'}\rangle| \Bigg)^2 \sum_{l=k_n+h_n+1}^\infty \lambda_l
\end{align*}
and that
\begin{align*}
\sum_{l=k_n+1}^{k_n+h_n} \lambda_l&
 \Bigg(
\sum_{l'=1}^{k_n} |\langle \brho, \be_{l'}\rangle|
\frac{1}{1- \frac{l`}{l}}
\Bigg)^2 \\
\leq&\;
\sum_{l=k_n+1}^{k_n+h_n} \lambda_l
 \Bigg(
\sum_{l'=1}^{\infty} |\langle \brho, \be_{l'}\rangle|    \Bigg)^2
\max_{\substack{k_n<l<k_n+h_n\\ l' \leq k_n}}\Bigg\{ \frac{1}{1- \tfrac{l`}{l}}
\Bigg\}^2
\\
\leq&\;
   (k_n+h_n)^2\left(
\sum_{l'=1}^{\infty} |\langle \brho, \be_{l'}\rangle|    \right)^2
\sum_{l=k_n+1}^{k_n+h_n} \lambda_l
\\
\leq&\;
4k_n ^2 \left(
\sum_{l'=1}^{\infty} |\langle \brho, \be_{l'}\rangle|    \right)^2
\sum_{l=k_n+1}^{k_n+h_n} \lambda_l.
\end{align*}

Replacing the last two inequalities in \eqref{Eq.tnx2.Yn.5}, we obtain that
\begin{align*}
\frac{C}{{n}} \sum_{l=k_n+1}^\infty
 \left(
\sum_{l'=1}^{k_n} \langle \brho, \be_{l'}\rangle
\frac{\sqrt{\lambda_l\lambda_{l'}}}{\lambda_{l'}- \lambda_{l}}
\right)^2
\leq&\;
\frac{2C}{{n}} \frac 1 {\lambda_{k_n}}   8k_n^2 \left(
\sum_{l'=1}^{\infty} |\langle \brho, \be_{l'}\rangle|    \right)^2
\sum_{l=k_n+1}^{\infty} \lambda_l
\\
\leq&\;
\frac{C}{{n}}    k_n^2 \left(
\sum_{l'=1}^{\infty} |\langle \brho, \be_{l'}\rangle|    \right)^2
(k_n+2)\frac {\lambda_{k_n+1}} {\lambda_{k_n}}
\\
\leq&\;
\frac{C}{{n}}    k_n^3 \left(
\sum_{l'=1}^{\infty} |\langle \brho, \be_{l'}\rangle|    \right)^2  ,
\end{align*}
where we have applied Lemma 1 in CMS. Obviously, this quantity converges to zero due to \ref{Assump:C2} and \ref{Assump:C7}. This proves that $\sqrt{n}  \E{|\langle  \mathcal{S}_n  \brho, \Medh - \mathbf{E}_{x,\bh}\rangle |} \conv 0$, and, hence, that $ \sqrt{n}  \langle  \mathcal{S}_n  \brho, \Medh - \mathbf{E}_{x,\bh}\rangle  \stackrel{p}{\to} 0$. Therefore, we only need to show that $ \sqrt{n} \langle  \mathcal{R}_n  \brho, \Medh - \mathbf{E}_{x,\bh}\rangle  \stackrel{p}{\to} 0$.\\

This term is dealt with in the proof on pages 350 and 351 in CMS. According to the arguments on those pages, it happens that we only need to show that
\begin{align}\label{Eq.tnx2.Yn.0_0}
\sqrt{n} \sum_{j=1}^{k_n} \int_{\mathcal{B}_j}
\E{\left\| G_n(z)\right\|^2_\infty \big\| (zI-\Gamma)^{-1/2}(\Medh - \mathbf{E}_{x,\bh})\big\|} \big\|\Gamma_z^{-1/2}\brho\big\|\, \mathrm{d}z \conv 0.
\end{align}

According to Lemma 3 in CMS, if $z \in \mathcal{B}_j$, then
\begin{align*}
\left\| G_n(z)\right\|^2_\infty
\leq&\;
4 \sum_{l=1}^\infty\sum_{\substack{k=1\\ k\neq j}}^\infty \frac{\langle (\Gamma_n - \Gamma)\be_l,\be_k\rangle^2}{|\lambda_j - \lambda_l| |\lambda_j - \lambda_k|}
\\
&+
2 \sum_{\substack{k=1\\ k\neq j}}^\infty \frac{\langle (\Gamma_n - \Gamma)\be_j,\be_k\rangle^2}{\delta_j |\lambda_j - \lambda_k|}
+\frac{\langle (\Gamma_n - \Gamma)\be_j,\be_j\rangle^2}{\delta_j ^2}.
\end{align*}
From \eqref{Eq.tnex2.Sn2_2.2} we have that if we take the positive value of the square root, then
\[
\big\|\Gamma_z^{-1/2}(\Medh - \mathbf{E}_{x,\bh}) \big\| \leq
\frac 1 {n} \sum_{i=1}^\infty \sqrt{\frac{\lambda_i}{|z-\lambda_i|}}  \sum_{r=1}^n  D_{x,\bh}^{r,i}.
\]
Thus, if $z \in \mathcal{B}_j$, then
\begin{align}
\nonumber
\mathbb{E}\Big[
\left\| G_n(z)\right\|^2_\infty &
\big\|\Gamma_z^{-1/2}(\Medh - \mathbf{E}_{x,\bh}) \big\|
\Big]
\\
\nonumber
\leq&\;
4\sum_{i=1}^\infty\sum_{l=1}^\infty\sum_{\substack{k=1\\k\neq j}}^\infty  \sqrt{\frac{\lambda_i}{|z-\lambda_i|}}   \frac{ \E{D_{x,\bh}^{1,i} \langle (\Gamma_n - \Gamma)\be_l,\be_k\rangle^2}}{|\lambda_j - \lambda_l| |\lambda_j - \lambda_k|}
\\
\nonumber
&
+
2 \sum_{i=1}^\infty\sum_{\substack{k=1\\k\neq j}}^\infty \sqrt{\frac{\lambda_i}{|z-\lambda_i|}}   \frac{ \E{D_{x,\bh}^{1,i} \langle (\Gamma_n - \Gamma)\be_j,\be_k\rangle^2}}{\delta_j |\lambda_j - \lambda_k|}
\\
\label{Eq.tnx2.Yn.7}
&    + \sum_{i=1}^\infty \sqrt{\frac{\lambda_i}{|z-\lambda_i|}}   \frac{ \E{D_{x,\bh}^{1,i} \langle (\Gamma_n - \Gamma)\be_j,\be_j\rangle^2}}{\delta_j ^2}.
\end{align}

It is not difficult to check that
$\langle (\Gamma_n - \Gamma)\be_l,\be_k\rangle
=
\frac{\sqrt{\lambda_l \lambda_k}}{n} \sum_{s=1}^n (\xi_l^s\xi_k^s - \delta_{l,k})$, so
\[
D_{x,\bh}^{1,i} \langle (\Gamma_n - \Gamma)\be_l,\be_k\rangle^2
=
\frac{\lambda_l \lambda_k}{n^2} \sum_{s,s'=1}^n D_{x,\bh}^{1,i}(\xi_l^s\xi_k^s - \delta_{l,k})(\xi_l^{s'}\xi_k^{s'} - \delta_{l,k}).
\]

Let $s,s' \in \{1,\ldots,n\}$. If $s\neq s'$, at least one of these is different from $1$. Assume that $s\neq 1$. By the independence of the sample, we have that
\begin{align*}
\mathbb{E}\Big[
D_{x,\bh}^{1,i}&(\xi_l^s\xi_k^s - \delta_{l,k})(\xi_l^{s'}\xi_k^{s'} - \delta_{l,k})
\Big]\\
=&\;
\E{
(\xi_l^s\xi_k^s - \delta_{l,k})
}
\E{D_{x,\bh}^{1,i}
(\xi_l^{s'}\xi_k^{s'} - \delta_{l,k})
}
=0.
\end{align*}

Similarly, if $s=s'$ and $s\neq 1$, we have $\Ebig{
D_{x,\bh}^{1,i}(\xi_l^s\xi_k^s - \delta_{l,k})(\xi_l^{s'}\xi_k^{s'} - \delta_{l,k})
}
=0$. Thus,
\[
\E{D_{x,\bh}^{1,i} \langle (\Gamma_n - \Gamma)\be_l,\be_k\rangle^2}
=
\frac{\lambda_l \lambda_k}{n^2} \E{D_{x,\bh}^{1,i}(\xi_l^1\xi_k^1 - \delta_{l,k})^2}.
\]

On the other hand, it can be checked that, for every $i,l,k$, it happens that
$\big|\Ebig{ D_{x,\bh}^{1,i}(\xi_l^1\xi_k^1 - \delta_{l,k})^2}\big| \leq C M,
$ where $M$ is given in \ref{Assump:C6}. For instance, let us assume that $l=k$ with $k\neq i$. We have that
\[
\left|\E{D_{x,\bh}^{1,i}(\xi_l^1\xi_k^1 - \delta_{l,k})^2}\right|
\leq
\E{|\xi_i^1| (|\xi_k^1|^2+1)^2}+
\E{\E{|\xi_i^1|} (|\xi_k^1|^2+1)^2}.
\]

The term with higher order expectations is $\E{|\xi_i^1| |\xi_k^1|^4}$, and it can be bounded:
\[
\E{|\xi_i^1| |\xi_k^1|^4} \leq \E{\max (|\xi_i^1|,|\xi_k^1|)^5}
\leq
\E{|\xi_i^1|^5+|\xi_k^1|^5} \leq 2M.
\]

In summary, we have that
\[
\left| \E{D_{x,\bh}^{1,i} \langle (\Gamma_n - \Gamma)\be_l,\be_k\rangle^2}\right|
\leq
\frac{C \lambda_l \lambda_k}{n^2}.
\]

From here and \eqref{Eq.tnx2.Yn.7} we obtain that if $z \in \mathcal{B}_j$, then
\begin{align}
\nonumber
\mathbb{E}\Big[
 \| G_n(z)\|^2_\infty
&\big\|\Gamma_z^{-1/2}(\Medh - \mathbf{E}_{x,\bh})  \big\|
\Big]
\\
\nonumber
\leq&\;
\frac{C}{n^2}
\Bigg(\sum_{l=1}^\infty\sum_{\substack{k=1\\k\neq j}}^\infty   \frac{ \lambda_l \lambda_k}{|\lambda_j - \lambda_l| |\lambda_j - \lambda_k|}\\
&+
\sum_{\substack{k=1\\ k\neq j}}^\infty   \frac{\lambda_j \lambda_k}{\delta_j |\lambda_j - \lambda_k|}
+  \frac{\lambda_j^2}{\delta_j ^2}
\Bigg) \sum_{i=1}^\infty  \sqrt{\frac{\lambda_i}{|z-\lambda_i|}}
\nonumber\\
\label{Eq.tnx2.Yn.8_0}
\leq&\;
\frac C{n^2} (j \log j)^2
\Bigg(
  \sum_{\substack{i=1\\ i\neq j}}^\infty  \sqrt{\frac{\lambda_i}{|\lambda_j-\lambda_i|}}
+ \sqrt{\frac{\lambda_j}{\delta_j}}
\Bigg),
\end{align}
where we have applied the same argument as in the final part of the proof of Lemma 3 in CMS and Lemma \ref{LemmTn2.Ln} in this paper. \\

Since $\delta_j \geq \lambda_j - \lambda_{j+1}$, the first part of Lemma 1 in CMS gives us that
\begin{align}\label{Eq.tnx2.Yn.8}
\sqrt{\frac{\lambda_j}{\delta_j}}
\leq
\sqrt{\frac{\lambda_j}{\lambda_j-\lambda_{j+1}}}\leq  \sqrt{j+1} \leq 2\sqrt{k_n}.
\end{align}

On the other hand, we have that
\begin{align}
\sum_{\substack{i=1\\ i\neq j}}^\infty  \sqrt{\frac{\lambda_i}{|\lambda_j-\lambda_i|}}
\leq
\sqrt{\frac{\lambda_j}{\delta_j}}\frac{1}{\sqrt{\lambda_j}} \sum_{\substack{i=1\\ i\neq j}}^\infty  \sqrt{\lambda_i}
\leq C \sqrt{k_n}  \frac{1}{\sqrt{c_n}},  \label{Eq.tnx2.Yn.9}
\end{align}
where we have employed that $\delta_j \leq | \lambda_j - \lambda_i|$ if $i \neq j$; \eqref{Eq.tnx2.Yn.8}; that \ref{Assump:C4} implies $ \sum_{\substack{i=1\\ i\neq j}}^\infty  \sqrt{\lambda_i}< \infty$; and, by the definitions of $\delta_j$ and $k_n$, that for every $j\leq k_n$, $c_n \leq \lambda_j + \frac {\delta_j} 2 \leq 3\frac {\lambda_j} 2$. Inequalities \eqref{Eq.tnx2.Yn.8_0}, \eqref{Eq.tnx2.Yn.8}, and \eqref{Eq.tnx2.Yn.9} give us that
\begin{align*}
\int_{\mathcal{B}_j}
\mathbb{E}\Big[\left\| G_n(z)\right\|^2_\infty &\big\| (zI-\Gamma)^{-1/2}(\Medh - \mathbf{E}_{x,\bh})\big\|\Big]
\big\|\Gamma_z^{-1/2}\brho\big\|\, \mathrm{d}z
\\
\leq&\;
C \text{diam}(\mathcal{B}_j) \frac{(j\log j)^2 \sqrt{k_n}}{n^2 \sqrt{c_n}}\sup_{z \in \mathcal{B}_j}\big\|\Gamma_z^{-1/2}\brho\big\|
\\
\leq&\;
C{\delta_j} \frac{(j\log j)^2 \sqrt{k_n}}{n^2 \sqrt{c_n}} \frac{1}{\sqrt{\delta_j}} \| \brho \|
\leq
C \frac{(j\log j)^{3/2} \sqrt{k_n}}{n^{7/4}},
\end{align*}
where we have applied the facts that if $z \in \mathcal{B}_j$, then $|z - \lambda_i| \geq \delta_j/2$; that if $j$ is large enough, then $\delta_j < C(j \log j)^{-1}$; and \ref{Assump:C7}. From here, we have that
\begin{align*}
\sqrt{n} \sum_{j=1}^{k_n} &\int_{\mathcal{B}_j}
\E{\left\| G_n(z)\right\|^2_\infty \big\| (zI-\Gamma)^{-1/2}(\Medh - \mathbf{E}_{x,\bh})\big\|} \big\|\Gamma_z^{-1/2}\brho\big\|\, \mathrm{d}z
\\
& \leq
C  \frac{k_n^3 (\log k_n)^{3/2}}{n^{5/4}}\conv0
\end{align*}
by Lemma \ref{Prop.1}. This proves \eqref{Eq.tnx2.Yn.0_0} and, consequently, that $n^{1/2} \langle \Medh - \mathbf{E}_{x,\bh},  \bY_n \rangle = \op(1)$.
\end{proof}

\begin{proof}[Proof of Lemma \ref{Lemm.Tn3}]
According to Lemma 8 in CMS, we have that
\[
\frac{\sqrt{n}}{t_{n,\textbf{E}_{x,\bh}}}\langle \mathbf{E}_{x,\bh}, \bR_n\rangle \inlaw  \mathcal{N}(0,\sigma^2_{\varepsilon}).
\]

From the proof of Proposition \ref{Prop.tnx}, we have that $\sup_j \frac {\langle \mathbf{E}_{x,\bh}, \be_j \rangle^2}{ \lambda_j} \leq 1$. Since the sequence $ \left\{ t_{n,\mathbf{E}_{x,\bh}}\right\}_n$ is strictly increasing, with its terms strictly positive, Lemma  \ref{Prop.1} implies that  $\frac{k_n^3 (\log k_n)^2}{\sqrt nt_{n,\textbf{E}_{x,\bh}}} \conv 0.$ Therefore,  the second part of Proposition 3 in CMS (page 352) gives us that $ \langle \mathbf{E}_{x,\bh}, \bY_n\rangle =\op (n^{-1/2})$. \\

Then, the result will be proven if we show that
\begin{align} \label{Eq.tnx.2}
n^{1/2} \left(
\langle \mathbf{E}_{x,\bh},  \bL_n \rangle
+
\langle \mathbf{E}_{x,\bh}, \bS_n \rangle
\right) =\op (1).
\end{align}

To prove \eqref{Eq.tnx.2}, we analyse separately both terms on the left hand side of this expression. To do this, we follow the steps in some proofs in CMS. Before we do so, notice that, in the computation of $\mathbf{E}_{x,\bh}$, we can take $\bX$ independent from all the remaining variables in the problem. In particular, for every $n \in \mathbb{N}$, we can consider $\bX$ to be independent of $\textbf{W}_n$, where $\textbf{W}_n$ denotes either $\bL_n$ or $\bS_n$. Thus, we have that
\begin{align*}
\E{
\left| \left\langle \E{\mathds{1}_{\lrb{\bX^\bh \leq x}} \bX}, \textbf{W}_n \right\rangle \right|}
=&\;
\E{
\left|
\E{
\langle \mathds{1}_{\lrb{\bX^\bh \leq x}} \bX , \textbf{W}_n \rangle
\Big|\textbf{W}_n
}
\right|
}
\\
\leq&\;
\E{\mathds{1}_{\lrb{\bX^\bh \leq x}} \left|  \langle  \bX , \textbf{W}_n \rangle \right|}
\leq
\E{|\langle  \bX , \textbf{W}_n \rangle  |}.
\end{align*}

Hence, if we show that $\E{|\langle  \bX , \textbf{W}_n \rangle  |} = \op(n^{-1/2})$, then Markov's inequality gives \eqref{Eq.tnx.2}. \\

Lemma 7 in CMS gives the following two bounds:
\[
\E{|\langle \bX,  \bL_n \rangle|}
\leq
\left\{
\begin{array}{l}
|\langle \brho, \be_{k_n}\rangle | \sqrt{ \sum_{j=k_n+1}^\infty  \lambda_j},
\\
\frac{\lambda_{k_n}}{\sqrt{k_n \log k_n}}  \sqrt{ \sum_{j=k_n+1}^\infty  \langle \brho,\be_j\rangle}.
\end{array}
\right.
\]
This inequality, either \ref{Assump:C1} or \ref{Assump:C2}, and \ref{Assump:C5} yield $
 \E{|  \langle  \bX , \bL_n \rangle  |}  =\op(n^{-1/2})$.\\

To handle the term $n^{1/2} \E{|\langle\bX,  \bS_n \rangle|}$, we follow the argument in Proposition 2 in CMS. We have that
\[
\langle \bX, \bS_n \rangle =  \langle \bX, \mathcal{R}_n \brho \rangle +  \langle \bX, \mathcal{S}_n \brho \rangle.
\]
Lemma \ref{Prop.1} implies that $k_n^2 \log k_n= o(n^{1/2})$. Then the reasoning in the last part of the proof of Proposition 2, page 352, in CMS leads to
\[
\E{|\langle \mathcal{R}_n\brho , \bX \rangle|}\leq C \frac{k_n^3(\log k_n)^2}{n} = o(n^{-1/2}).
\]
Concerning the term $\langle \bX, \mathcal{S}_n \brho \rangle$, let us consider the decomposition of $\E{\langle \bX, \mathcal{S}_n \brho \rangle^2}$, which appears in (24) and (25) in CMS. The term in (24) is bounded by the expression in (26) and then the authors of CMS find a bound for each term in (26). The difference in our case is that our term $\E{\langle \bX, \mathcal{S}_n \brho \rangle}^2$ is not divided by $k_n$. Therefore, in order to be able to apply the inequality in the second display on page 350 of CMS (which allows them to bound the second term in (26) in their paper), we need to reinforce the assumption $\lambda_{k_n} k_n \log k_n \conv 0$, used in CMS, to $\lambda_{k_n} k_n^3 \log k_n \conv 0$, which follows from~\ref{Assump:C4}. \\

The bound for the first term in (26) in CMS in our case is
\[
k_n M \frac {k_n} {\log k_n} \lambda_{k_n} k_n^2 \max_{k_n < j \leq k_n+h_n}\left\{\langle \brho, \be_{j}\rangle^2 \right\},
\]
which converges to zero by \ref{Assump:C2} and \ref{Assump:C4}. The convergence to zero of the term in (25) in CMS is proved similarly.
\end{proof}

\section{Supplement to the simulation study}
\label{ap:fursim}

\subsection{Detailed description of simulation scenarios}
\label{ap:fursim:sce}

The different data generating processes used in the simulation study are encoded as follows. For the $k$th simulation scenario S$k$, with functional coefficient $\brho_k$, the deviation from $H_0$ is measured by a deviation coefficient $\delta_d$, with $\delta_0=0$ and $\delta_d>0$ for $d=1,2$. Then, under $H_{k,d}$, we denote data generation by
\[
Y=\inprod{\bX}{\brho_k}+\delta_d \Delta_{\eta_k}(\bX)+\varepsilon,
\]
where $\boldsymbol{\eta}:=(1,2,1,2,2,1,2,3,3)'$ and the deviations from the linear model are constructed by including the nonlinear terms $\Delta_{1}(\bX):=\norm{\bX}$, $\Delta_{2}(\bX):=25\int_0^1\int_0^1\sin(2\pi ts)s(1-s)t(1-t)\bX(s)\bX(t)\,\mathrm{d}s\,\mathrm{d}t$, and $\Delta_{3}(\bX):=\inprod{e^{-\bX}}{\bX^2}$. The error $\varepsilon$ is distributed as a $\mathcal{N}(0,\sigma^2)$, where $\sigma^2$ was chosen such that, under $H_0$, $R^2=\frac{\V{\inprod{\bX}{\brho}}}{\V{\inprod{\bX}{\brho}}+\sigma^2}=0.95$.\\

The description of the simulation scenarios is given in Table \ref{tab:scenarios}. The functional processes $\mathbf{X}(t)$, all of them indexed in $[0,1]$ and discretized in $201$ equidistant points, are the following:
\begin{description}
\item[BM.] Brownian motion, denoted by $\mathbf{B}$, whose eigenfunctions are $\boldsymbol\psi_j(t):=\sqrt{2}\sin\big((j-\frac{1}{2})\pi t\big)$, $j\geq1$.
\item[HHN.] The functional process considered in \cite{Hall2006}, given by $\bX(t)=\sum_{j=1}^{20}\xi_j \boldsymbol{\phi}_j(t)$, where $\boldsymbol{\phi}_j(t):=\sqrt{2}\cos\lrp{j\pi t}$ and $\xi_j$ are independent r.v.'s distributed as $\mathcal{N}\lrp{0,j^{-2l}}$, with $l=1,2$.
\item[BB.] Brownian bridge, defined as $\bX(t)=\mathbf{B}(t)-t\mathbf{B}(1)$. Its eigenfunctions are $\tilde\bpsi_j(t):=\bpsi_{j+\frac{1}{2}}(t)$, $j\geq1$.
\item[OU.] Ornstein--Uhlenbeck process, defined as the zero-mean Gaussian process with covariance given by $\mathbb{C}\mathrm{ov}[\bX(s),\bX(t)]=\frac{\sigma^2}{2\alpha}e^{-\alpha(s+t)}\big(e^{2\alpha\min(s,t)}-1\big)$. We consider $\alpha=\frac{1}{3}$, $\sigma=1$, and $\bX(0)\sim\mathcal{N}(0,\frac{\sigma^2}{2\alpha})$.
\item[GBM.] Geometric Brownian motion, defined as $\bX(t)=s_0\exp\big\{\big(\mu-\frac{\sigma^2}{2}\big)t+\sigma\mathbf{B}(t)\big\}$. We consider $\sigma=1$, $\mu=\frac{1}{2}$, and $s_0=2$.
\end{description}

\begin{table}[h!]
\centering
\footnotesize
\setlength{\tabcolsep}{2pt}
\begin{tabular}{c|c|c|l}
\toprule\toprule
Scenario & \multicolumn{1}{c|}{Coefficient $\brho(t)$} & \multicolumn{1}{c|}{Process $\bX$} & \multicolumn{1}{c}{Deviation}  \\
\midrule
S1  & $(2\bpsi_1(t)+4\bpsi_2(t)+5\bpsi_3(t))/\sqrt{2}$ & BM & $\Delta_1$, $\delta=\lrp{0,\frac{1}{4},\frac{3}{4}}'$ \\\midrule 

S2  & $(2\tilde\bpsi_1(t)+4\tilde\bpsi_2(t)+5\tilde\bpsi_3(t))/\sqrt{2}$ & BB & $\Delta_2$, $\delta=\lrp{0,-2,-\frac{15}{2}}'$ \\\midrule 

S3  & $(2\bpsi_2(t)+4\bpsi_3(t)+5\bpsi_7(t))/\sqrt{2}$ & BM  & $\Delta_1$, $\delta=\lrp{0,-\frac{1}{5},-\frac{1}{2}}'$ \\\midrule 

S4  & $\sum_{j=1}^{20}2^{3/2}(-1)^jj^{-2}\bphi_j(t)$ & HHN ($l=1$) & $\Delta_2$, $\delta=(0,-1,-3)'$ \\\midrule 

S5  & $\sum_{j=1}^{20}2^{3/2}(-1)^jj^{-2}\bphi_j(t)$  & HHN ($l=2$) & $\Delta_2$, $\delta=(0,-1,-3)'$ \\\midrule 

S6  & $\log\lrp{15t^2 + 10} + \cos(4\pi t)$ & BM & $\Delta_1$, $\delta=\lrp{0,\frac{1}{5},1}'$ \\\midrule 

S7  & $\sin(2\pi t)-\cos(2\pi t)$ & OU & $\Delta_2$, $\delta=\lrp{0,-\frac{1}{4},-1}'$ \\\midrule 

S8  & $t-\lrp{t-\frac{3}{4}}^2$ & OU & $\Delta_3$, $\delta=\lrp{0,-\frac{1}{100},-\frac{1}{10}}'$ \\\midrule 

S9  & $\pi^2\lrp{t^2 - \frac{1}{3}}$ & GBM &  $\Delta_3$, $\delta=\lrp{0,\frac{1}{2},\frac{5}{2}}'$ \\ 
\bottomrule\bottomrule
\end{tabular}
\caption{\small Simulation scenarios and deviations from the null hypothesis.\label{tab:scenarios}}
\end{table}

The first scenario, S1, contains a $\brho$ based on example (a) in Section 5 of \cite{Cardot2003a}, which is a linear combination of the first three eigenfunctions of the Brownian motion. Variations on the same idea -- a $\brho$ that is a finite linear combination of the eigenfunctions of the functional process -- are employed in the next four scenarios. S2 considers a Brownian bridge as the functional process. S3 takes linear combinations of eigenfunctions associated to smaller eigenvalues to construct $\brho$. S4 and S5 collect the process and coefficients used in Section 5 of \cite{Hall2006} for $l=1$ and $l=2$, respectively, which is a finite-dimensional smooth process. S6 is example (b) in \cite{Cardot2003a}, which is not expressible as a finite combination of eigenfunctions. The remaining scenarios follow this idea: S7 and S8 with Ornstein--Uhlenbeck processes (as in Section 4.2 of \cite{Garcia-Portugues:flm}) and S9 with geometric Brownian motion. 

\begin{figure}[h!]
\centering
\includegraphics[width=\textwidth]{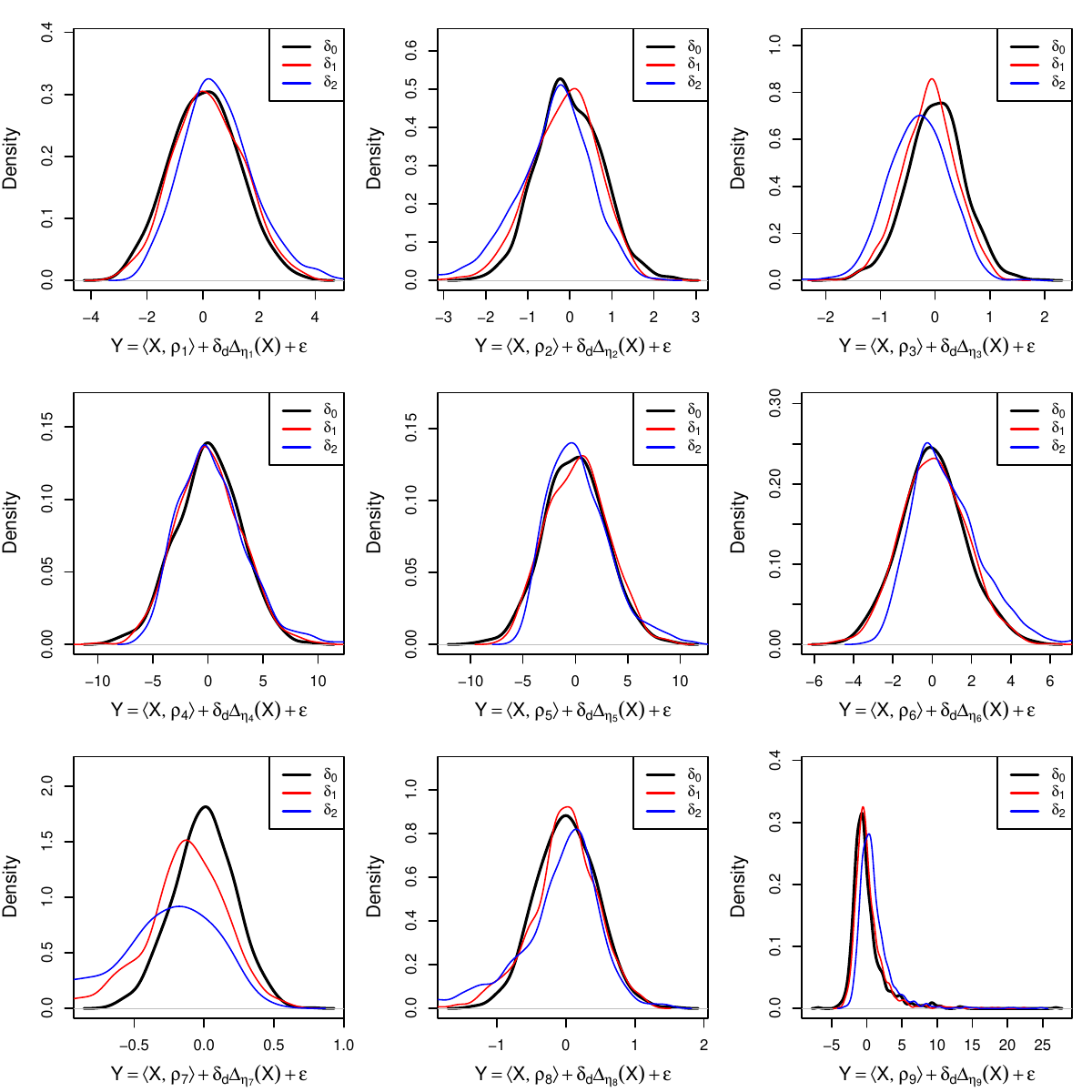}
\caption{\small Densities of the responses $Y$ under $H_{k,d}$, with $k=1,\ldots,9$ from left to right and up to down, and $d=0,1,2$.\label{fig:dens}}
\end{figure}

The deviation coefficients $\delta_d$, $d=1,2$, for the scenarios in Table \ref{tab:scenarios} were chosen by comparing the densities of the response $Y$ under the null and alternative. This comparison provides a graphical visualization of the difficulty in distinguishing between the hypotheses (Figure \ref{fig:dens}). The actual choice of the coefficients aims to make this distinction a challenging task.

\subsection{Supplementary results for the composite hypothesis}
\label{ap:fursim:comp}

Table \ref{tab:results50} shows the empirical sizes and powers for $n=50$. The numerical summaries for the data-driven $d_n$ are given in Table \ref{tab:p}.

\begin{table}[H]
\centering
\footnotesize
\setlength{\tabcolsep}{1pt}
\begin{tabular}{l|>{\centering}m{1.00cm}>{\centering}m{1.00cm}>{\centering}m{1.00cm}|>{\centering}m{1.00cm}>{\centering}m{1.00cm}>{\centering}m{1.00cm}|c}
  \toprule  \toprule
\multirow{2}{*}{$H_{k,\delta}$} & \multicolumn{7}{c}{$n=50$} \\ \cmidrule{2-8}
&  CvM$_{1}$ & CvM$_{3}$ & CvM$_{5}$ & KS$_{1}$ & KS$_{3}$ & KS$_{5}$ & PCvM  \\  \midrule
$H_{1,0}$ & $6.0$ & $4.7$ & $4.1$ & $6.2$ & $5.6$ & $4.7$ & $4.6$ \\
$H_{2,0}$ & $6.0$ & $5.2$ & $4.8$ & $6.5$ & $5.8$ & $5.5$ & $3.5$ \\
$H_{3,0}$ & $7.9$ & $6.7$ & $6.0$ & $8.6$ & $8.0$ & $7.4$ & $7.2$ \\
$H_{4,0}$ & $6.0$ & $5.1$ & $4.3$ & $6.6$ & $5.3$ & $4.7$ & $4.2$ \\
$H_{5,0}$ & $5.7$ & $4.2$ & $3.7$ & $5.8$ & $4.8$ & $4.3$ & $4.8$ \\
$H_{6,0}$ & $5.6$ & $4.3$ & $3.6$ & $5.7$ & $4.7$ & $4.4$ & $4.9$ \\
$H_{7,0}$ & $6.2$ & $5.0$ & $4.5$ & $6.6$ & $5.6$ & $5.4$ & $5.2$ \\
$H_{8,0}$ & $5.4$ & $4.2$ & $3.6$ & $5.6$ & $4.7$ & $4.4$ & $4.5$ \\
$H_{9,0}$ & $6.3$ & $5.1$ & $4.8$ & $6.1$ & $5.2$ & $4.9$ & $6.2$ \\
\midrule
$H_{1,1}$ & $29.4$ & $29.2$ & $27.7$ & $22.1$ & $21.4$ & $21.0$ & $35.1$ \\
$H_{2,1}$ & $63.2$ & $77.7$ & $79.5$ & $54.5$ & $65.2$ & $66.7$ & $84.9$ \\
$H_{3,1}$ & $65.5$ & $70.2$ & $69.6$ & $53.9$ & $58.0$ & $57.6$ & $80.5$ \\
$H_{4,1}$ & $17.0$ & $17.0$ & $16.5$ & $13.5$ & $13.4$ & $12.6$ & $20.3$ \\
$H_{5,1}$ & $23.3$ & $21.1$ & $20.3$ & $17.5$ & $16.1$ & $15.1$ & $24.7$ \\
$H_{6,1}$ & $13.1$ & $11.6$ & $10.7$ & $10.7$ & $9.6$ & $9.0$ & $14.2$ \\
$H_{7,1}$ & $94.7$ & $98.8$ & $98.8$ & $92.0$ & $96.0$ & $96.0$ & $99.0$ \\
$H_{8,1}$ & $53.6$ & $52.8$ & $52.4$ & $34.3$ & $33.8$ & $33.3$ & $55.2$ \\
$H_{9,1}$ & $8.1$ & $6.8$ & $6.5$ & $7.8$ & $7.0$ & $6.6$ & $8.7$ \\
\midrule
$H_{1,2}$ & $88.0$ & $96.7$ & $96.9$ & $81.2$ & $90.2$ & $90.3$ & $98.7$ \\
$H_{2,2}$ & $77.4$ & $95.8$ & $97.6$ & $73.8$ & $91.0$ & $93.3$ & $98.3$ \\
$H_{3,2}$ & $90.6$ & $98.5$ & $98.7$ & $87.2$ & $95.1$ & $95.5$ & $99.5$ \\
$H_{4,2}$ & $55.0$ & $68.2$ & $70.0$ & $46.1$ & $54.9$ & $56.7$ & $79.7$ \\
$H_{5,2}$ & $81.6$ & $84.6$ & $84.1$ & $70.0$ & $71.7$ & $71.1$ & $87.7$ \\
$H_{6,2}$ & $86.0$ & $95.7$ & $96.0$ & $78.1$ & $87.6$ & $88.3$ & $97.8$ \\
$H_{7,2}$ & $95.3$ & $99.0$ & $99.0$ & $92.8$ & $96.2$ & $96.3$ & $99.2$ \\
$H_{8,2}$ & $66.7$ & $67.8$ & $67.5$ & $31.8$ & $33.1$ & $33.2$ & $70.6$ \\
$H_{9,2}$ & $52.8$ & $56.2$ & $55.4$ & $46.7$ & $50.3$ & $50.5$ & $63.7$ \\
\bottomrule\bottomrule
\end{tabular}
\caption{\small Empirical sizes and powers (in percentages) of the CvM, KS, and PCvM tests with $\alpha=0.05$, sample size $n=50$, and estimation of $\brho$ by data-driven FPC ($d_n$ chosen by SICc). KS and CvM tests are shown with $1$, $3$, and $5$ projections. \label{tab:results50}}
\end{table}

\begin{table}[h]
\centering
\footnotesize
\setlength{\tabcolsep}{1.5pt}
\begin{tabular}{l|ccc||ccc||ccc}
  \toprule  \toprule
\multirow{2}{*}{$H_{k,\delta}$} & \multicolumn{3}{c||}{$n=50$} & \multicolumn{3}{c||}{$n=100$} & \multicolumn{3}{c}{$n=250$} \\ \cmidrule{2-10}
&  $\delta=0$ & $\delta=1$ & $\delta=2$ &  $\delta=0$ & $\delta=1$ & $\delta=2$ &  $\delta=0$ & $\delta=1$ & $\delta=2$ \\  \midrule
$H_{1,\delta}$ & 3.4 (0.7) & 3.4 (0.6) & 3.3 (0.5) & 3.4 (0.6) & 3.3 (0.6) & 3.2 (0.5) & 3.3 (0.5) & 3.3 (0.5) & 3.2 (0.4) \\
$H_{2,\delta}$ & 3.7 (0.9) & 3.5 (0.7) & 3.0 (0.6) & 3.7 (0.8) & 3.5 (0.7) & 3.1 (0.4) & 3.6 (0.8) & 3.4 (0.6) & 3.1 (0.3) \\
$H_{3,\delta}$ & 7.8 (1.1) & 7.6 (1.2) & 6.2 (1.9) & 8.1 (1.0) & 7.9 (0.9) & 7.5 (0.8) & 8.1 (0.9) & 7.9 (0.9) & 7.5 (0.7) \\
$H_{4,\delta}$ & 2.0 (0.7) & 2.0 (0.7) & 1.8 (0.7) & 2.3 (0.6) & 2.3 (0.6) & 2.0 (0.6) & 2.6 (0.7) & 2.6 (0.7) & 2.3 (0.6) \\
$H_{5,\delta}$ & 1.5 (0.6) & 1.4 (0.6) & 1.3 (0.5) & 1.7 (0.6) & 1.7 (0.6) & 1.5 (0.6) & 2.0 (0.3) & 2.0 (0.4) & 1.8 (0.5) \\
$H_{6,\delta}$ & 1.6 (0.8) & 1.6 (0.8) & 1.4 (0.6) & 1.9 (0.8) & 1.9 (0.8) & 1.6 (0.7) & 2.5 (1.0) & 2.5 (1.0) & 2.1 (0.7) \\
$H_{7,\delta}$ & 4.2 (0.7) & 3.2 (0.6) & 1.5 (0.7) & 4.3 (0.7) & 3.4 (0.6) & 1.8 (0.8) & 4.4 (0.7) & 3.7 (0.5) & 2.4 (0.7) \\
$H_{8,\delta}$ & 2.1 (0.4) & 1.9 (0.5) & 1.2 (0.5) & 2.2 (0.5) & 2.0 (0.4) & 1.1 (0.4) & 2.3 (0.5) & 2.0 (0.4) & 1.1 (0.4) \\
$H_{9,\delta}$ & 3.0 (1.0) & 3.0 (1.0) & 2.8 (0.9) & 3.5 (1.1) & 3.4 (1.1) & 3.2 (1.0) & 4.3 (1.0) & 4.3 (1.0) & 3.9 (1.0) \\\bottomrule\bottomrule
\end{tabular}
\caption{\small Averages of the SICc-driven $d_n$ for the different models, sample sizes, and deviations from the null hypothesis. Standard deviations are given in parentheses. \label{tab:p}}
\end{table}

\subsection{Dependence on the projection process}
\label{ap:fursim:projs}

The goodness-of-fit tests depend clearly on the way random directions $\bh$ are chosen. In order to explore its practical influence, we replicated the results in Section \ref{sec:simu} for two new processes. We have three different projecting processes in total: (\textit{i}) the data-driven process described in Section \ref{sec:testing}; (\textit{ii}) the same process, but with constant variance coefficients $\eta_j\sim\mathcal{N}(0,1)$; (\textit{iii}) an Ornstein--Uhlenbeck process with $\alpha=\tfrac{1}{2}$ and $\sigma=1$ that is completely independent from the sample. Note that process (\textit{ii}) generates more noisy random directions than (\textit{i}) since all the FPC's of the sample are equally weighted. 

\begin{figure}[h!]
\centering
\includegraphics[width=0.315\textwidth]{sizes_cvm_100_sd0.pdf}
\includegraphics[width=0.315\textwidth]{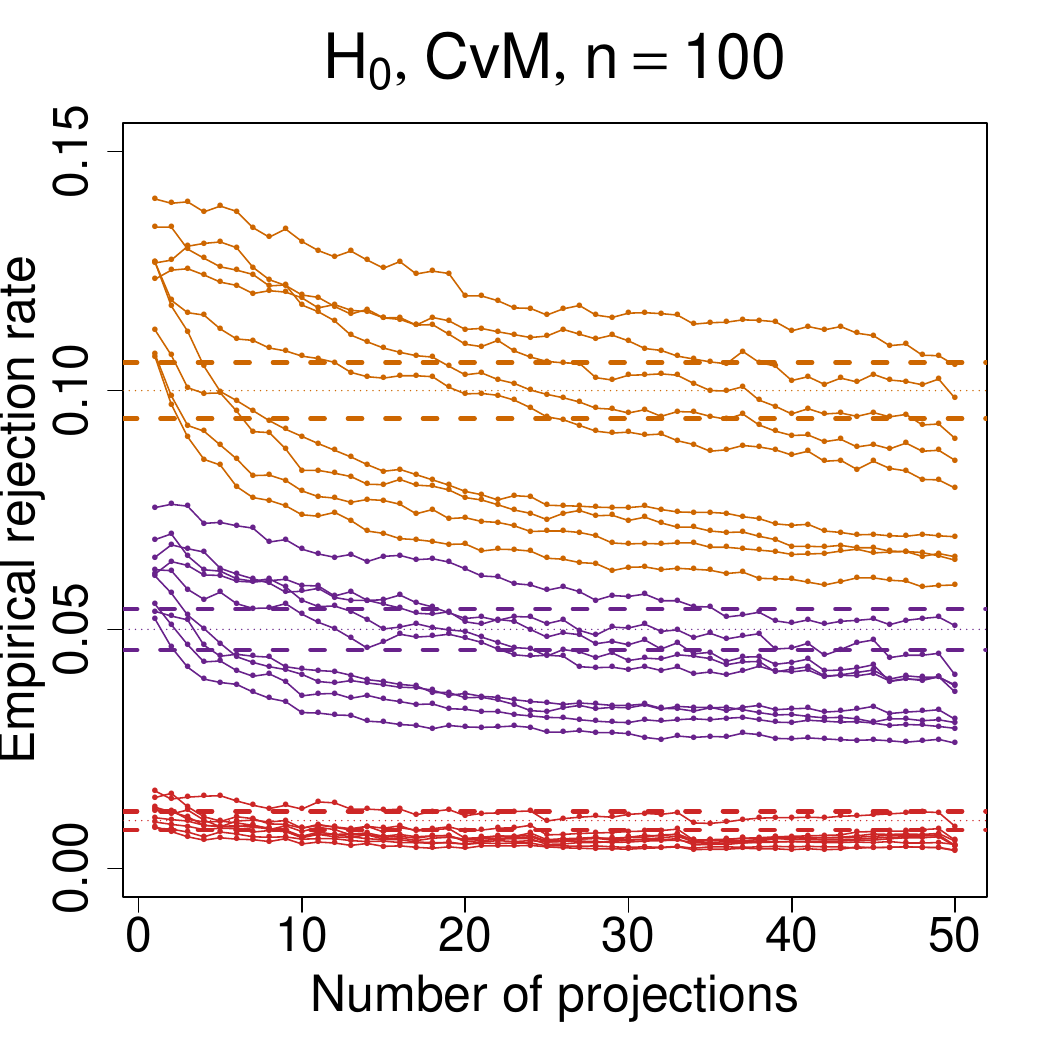}
\includegraphics[width=0.315\textwidth]{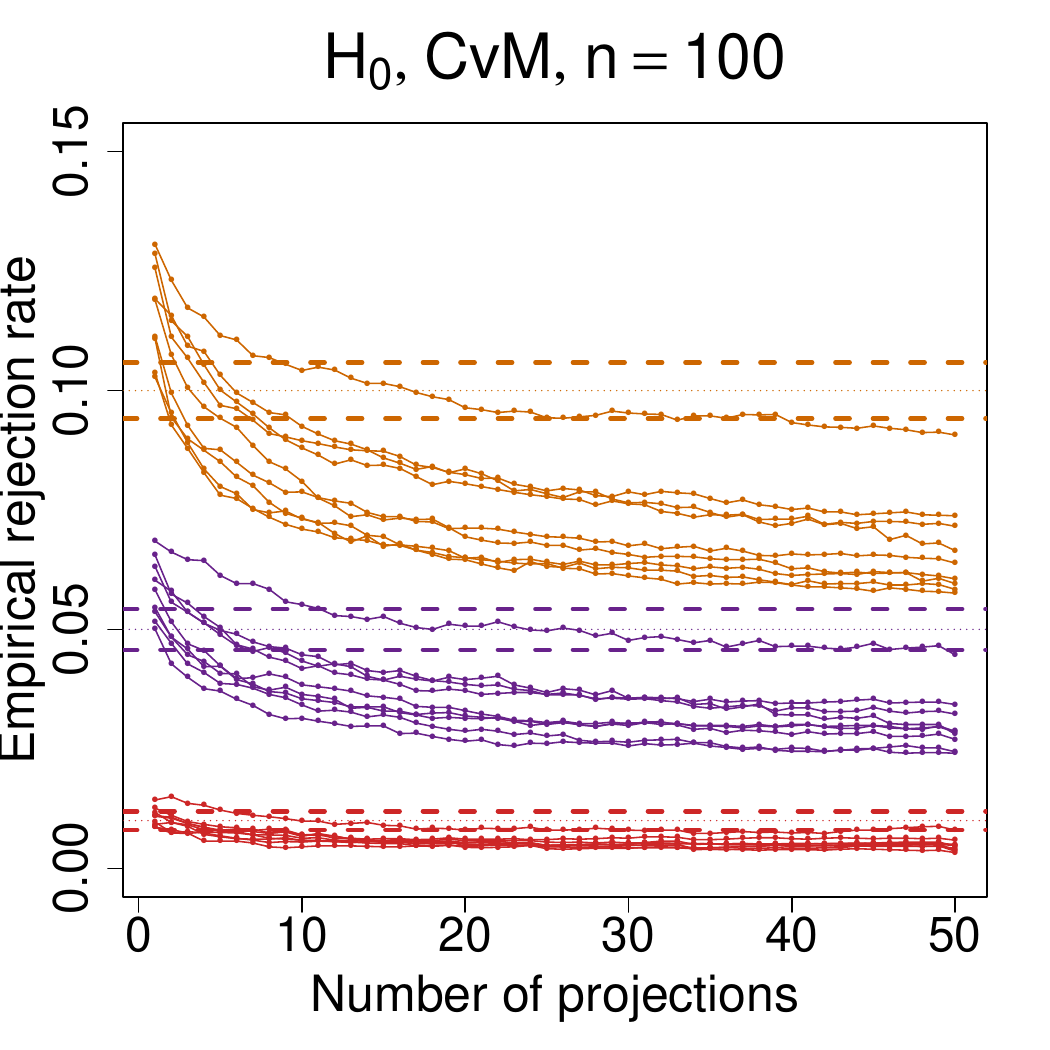}\\
\includegraphics[width=0.315\textwidth]{sizes_ks_100_sd0.pdf}
\includegraphics[width=0.315\textwidth]{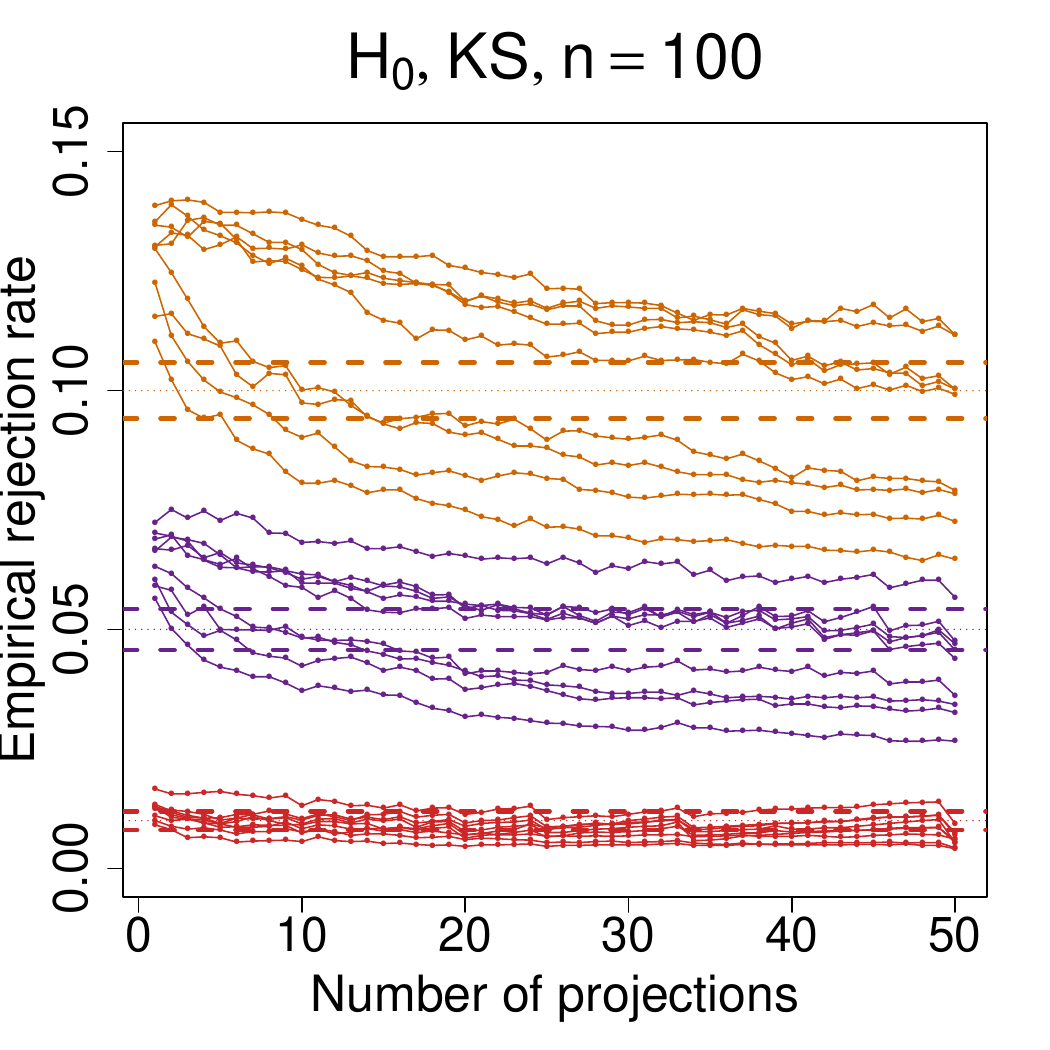}
\includegraphics[width=0.315\textwidth]{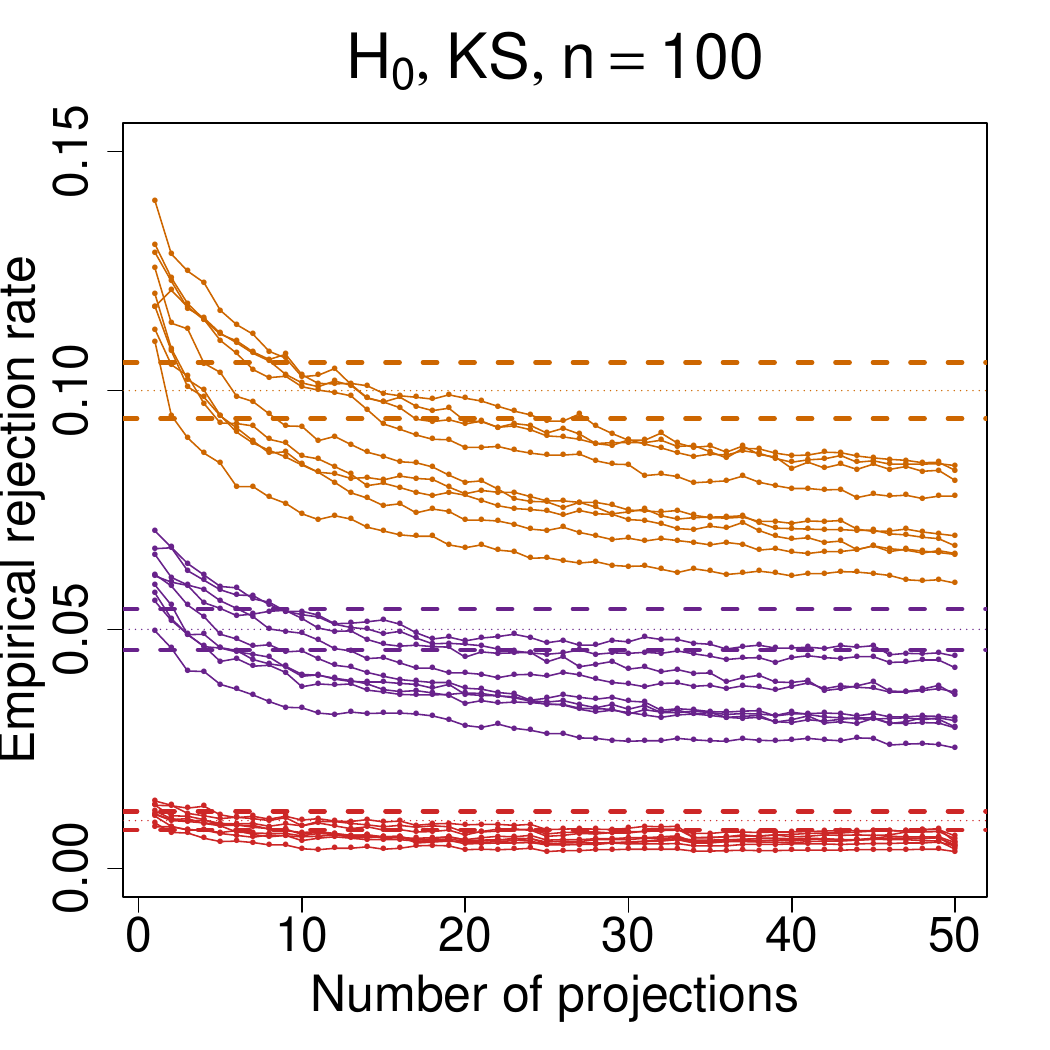}
\caption{\small Empirical sizes of the CvM (upper row) and KS (lower row) tests for the scenario S$k$, $k=1,\ldots,9$, depending on the number of projections $K=1,\ldots,50$. From left to right, columns represent the data generating processes (\textit{i}), (\textit{ii}), and (\textit{iii}). The sample size is $n=100$. The empirical sizes associated with the significance levels $\alpha=0.01,0.05,0.10$ are coded in red, purple, and orange, respectively. Dashed thick lines represent the asymptotic $95 \%$ confidence interval for the proportion $\alpha$ obtained from $M$ replicates.\label{fig:sizeproj:projs}}
\end{figure}

Figures \ref{fig:sizeproj:projs} and \ref{fig:powpro:projsj} show the empirical levels and powers of the tests based on processes (\textit{i}), (\textit{ii}), and (\textit{iii}) for $n=100$. Relatively minor changes can be observed between (\textit{i}) and (\textit{iii}), with the main features described in Section \ref{sec:simu} being consistent: L-shaped patterns in the size curves, mild decrements for the power curves, occasional bumps yielding power gains, and domination of CvM over KS. The results for both processes show no main changes, and both indicate that $K\in\{1,\ldots,5\}$ is a reasonable choice with respect to size and power. The big picture for (\textit{ii}) is similar, albeit with\nopagebreak[4] more spread and variable level curves, and power curves dominated by those of (\textit{i}) and (\textit{iii}). \\

The presented empirical results indicate that less variable random directions seem to yield better behaviour for the tests and that the data-driving process given in Section \ref{sec:testing} is a sensible alternative. However, more thorough research into the selection of the projecting process -- beyond the scope of this paper -- is required.

\begin{figure}[h!]
\centering
\includegraphics[width=0.315\textwidth]{powers_cvm_100_1_sd0.pdf}
\includegraphics[width=0.315\textwidth]{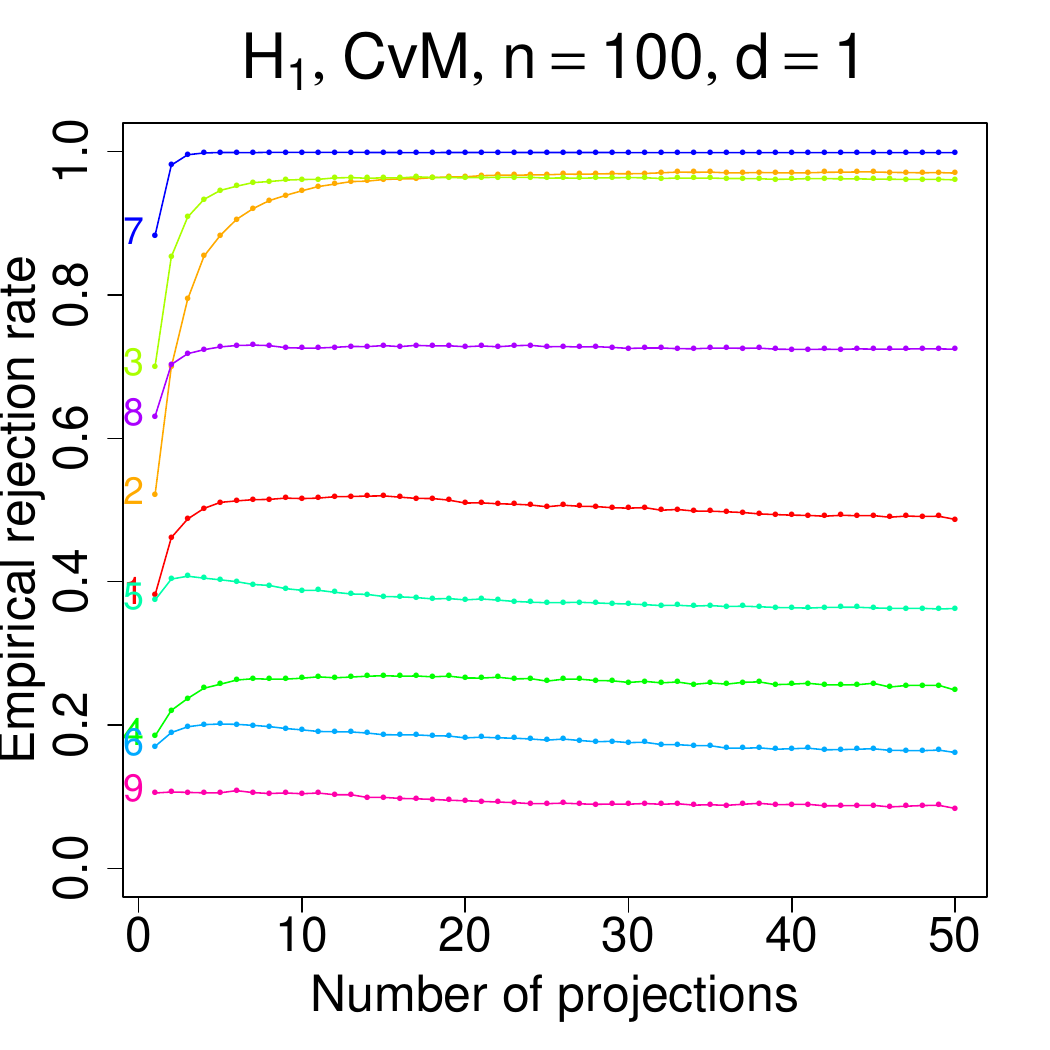}
\includegraphics[width=0.315\textwidth]{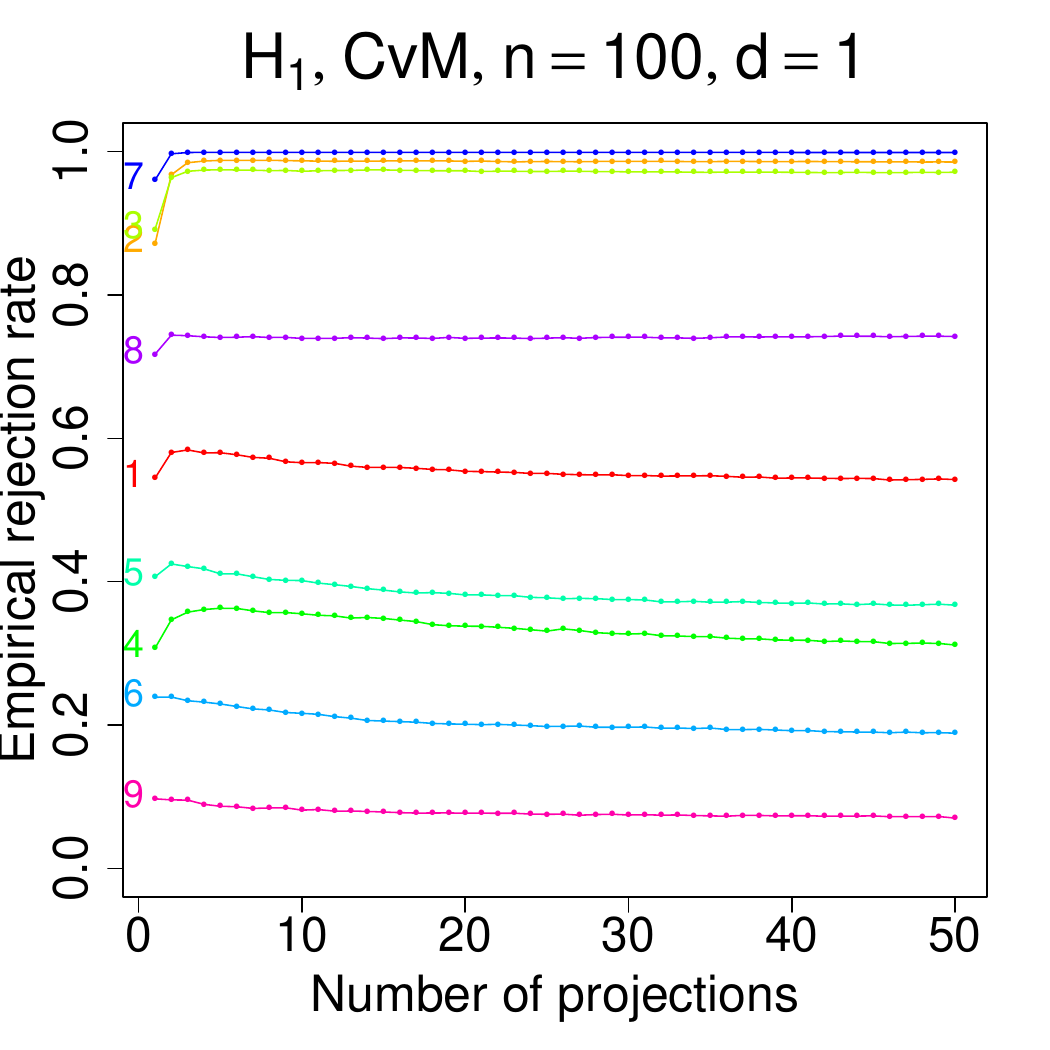}\\
\includegraphics[width=0.315\textwidth]{powers_ks_100_1_sd0.pdf}
\includegraphics[width=0.315\textwidth]{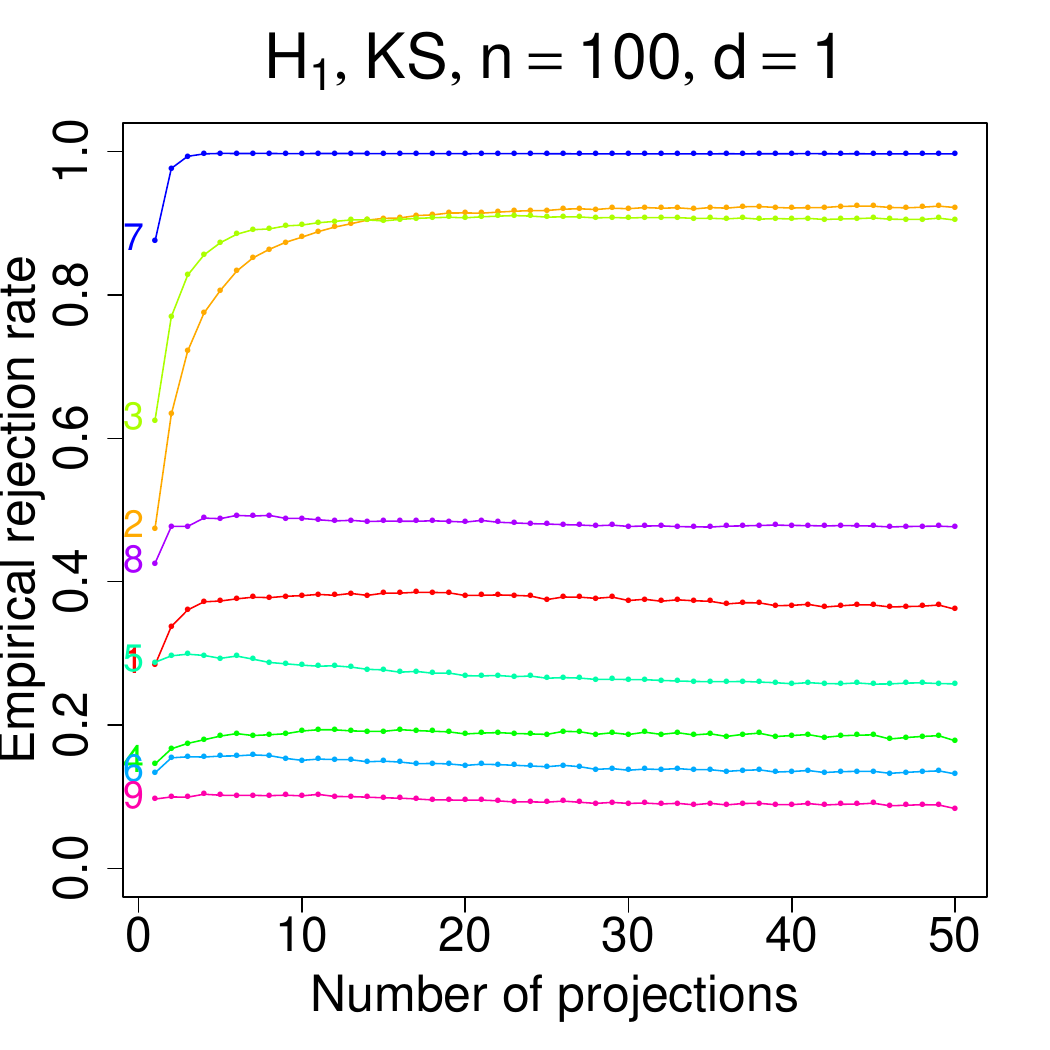}
\includegraphics[width=0.315\textwidth]{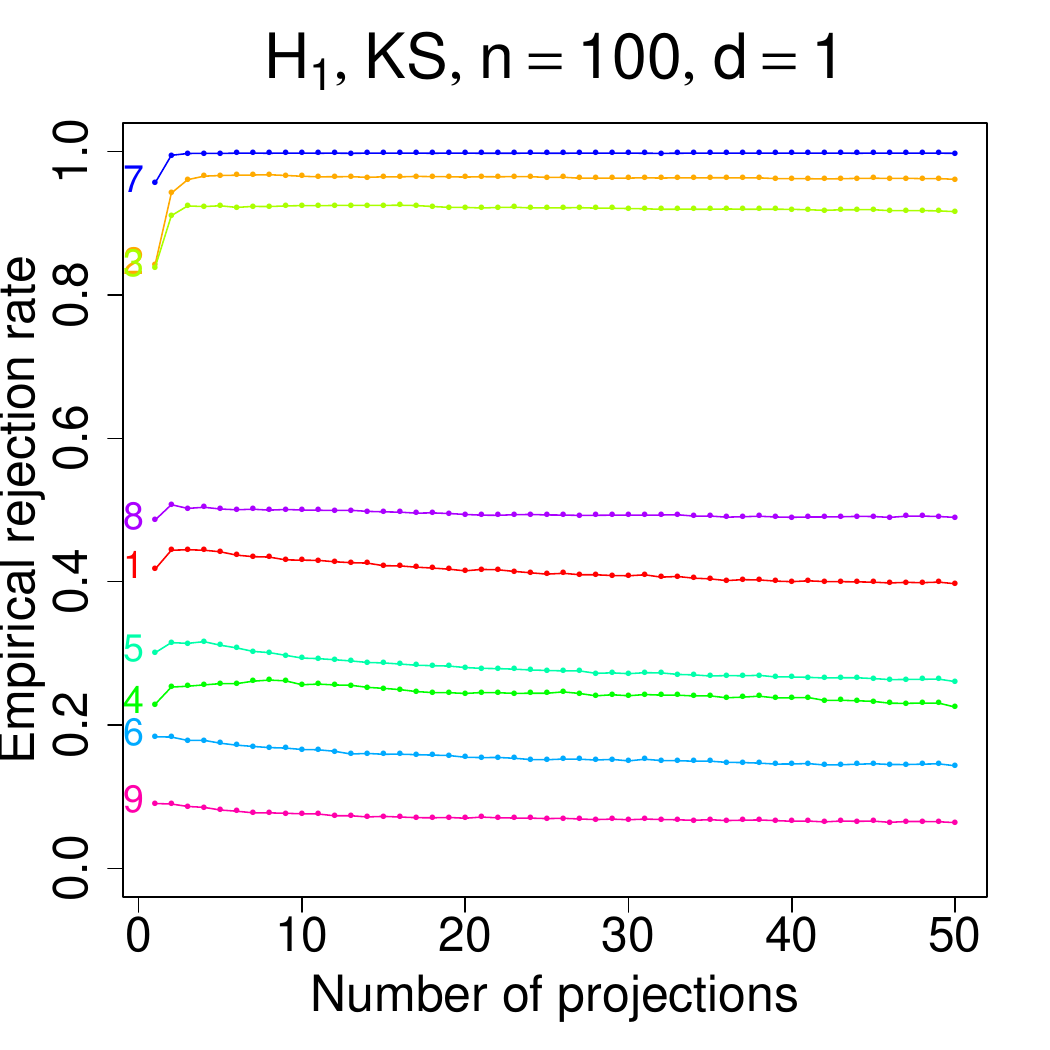}
\caption{\small Empirical powers of the CvM (first row) and KS (second two) tests for the scenario S$k$, $k=1,\ldots,9$, depending on the number of projections $K=1,\ldots,50$. From left to right, columns represent the data generating processes (\textit{i}), (\textit{ii}), and (\textit{iii}). The significance level is $\alpha=0.05$, the sample size is $n=100$, and the deviation index is $d=1$. \label{fig:powpro:projsj}}
\end{figure}

\subsection{\texorpdfstring{Remarks on discrete $p$-values and FDR correction}{Remarks on discrete p-values and FDR correction}}
\label{ap:fursim:fdr}

The FDR correction of a set of $K$ \textit{continuous} $p$-values $p_1,\ldots,p_K$, arising from $K$ hypothesis tests, results in the FDR $p$-value $p_{\mathrm{FDR},K}:=\min_{i=1,\ldots,K}\frac{K}{i}p_{(i)}$. Under the null hypotheses for all the tests, the level of the test that rejects the null hypothesis if $p_{\mathrm{FDR},K}<\alpha$ is $\alpha$ at most \citep{Benjamini2001}. When using a resampling strategy for approximating the $p$-values $p_1,\ldots,p_K$, for example, by considering $B$ bootstrap replicates, we end up with a collection of \textit{discrete} $p$-values $\hat p_1,\ldots,\hat p_K$. This has a notable influence on $p_{\mathrm{FDR},K}$, resulting in an increment of the type I error that is magnified for moderate and large $K$. \\

Under the null, $\hat p_i$ is approximately distributed as a $\mathcal{U}(\{0,\frac{1}{B},\ldots,\frac{B}{B}\})$, $i=1,\ldots,K$. If independence between $\hat p_1,\ldots,\hat p_K$ is assumed, then the rejection rate of $p_{\mathrm{FDR},K}<\alpha$ is at least $\mathbb{P}[\hat p_{(1)}=0]\approx 1-\big(\tfrac{B}{B+1}\big)^K$ \textit{no matter what significance level $\alpha$ is chosen}. For example, if $K=25$ and $B=500$, under the null hypothesis, $p_{\mathrm{FDR},K}<\alpha$ will reject the null at least $4.87\%$ of the time for any $\alpha$. For $K=5$ and $B=1000$, the lower bound for the rejection percentage drops to $0.499\%$. This simple argument illustrates the more demanding precision (larger $B$'s) required in the approximated $p$-values when\nolinebreak[4] $K$ grows. \\

In order to gain more insights into the problematic dependence of $K$ and $B$, we have conducted the following experiment, aimed at reproducing a comparable situation to our testing in practice.
\begin{enumerate}[label=\arabic{*}.]
\item Simulate $K$ discrete $p$-values independently: $\hat p_i\sim\mathcal{U}(\{0,\frac{1}{B},\ldots,\frac{B}{B}\})$ and compute $p_{\mathrm{FDR},K}$.
\item Repeat the previous step $M=10000$ times and obtain the empirical rejection rates of $p_{\mathrm{FDR},K}<\alpha$ for $\alpha=0.10,0.05,0.01$. Plot the rejection curves as a function of $K=1,\ldots,50$. Repeat this five times to account for variability.
\item Repeat the above steps for different $B$'s.
\end{enumerate}

The results of the experiment, namely the empirical rejection curves as a function of $K$ for different $\alpha$'s and $B$'s, are shown in the left panel of Figure \ref{fig:fdrpvalvsdiscretization}. A sawtooth pattern of over-rejection appears for curves with $B=500, 1000$ (yellow and light green curves) and values of $K$ larger than $7-10$, resulting in significant violations of the confidence intervals for the proportions $\alpha$ for $K$'s in the range $[10, 50]$. When $B$ is larger (dark green and blue curves), the rejection rates remain more stable and inside the confidence intervals for $K$ up to $50$. This highlights that, given the effect that both $K$ and $B$ have on the computation proficiency of the test, a reasonable compromise on the choice of $K$ and $B$ that respects the type I error is a low value for $K$, say $K\in\{1,\ldots,5\}$, and a relatively large value of $B$, such as $B\geq 1000$. The right panel of Figure \ref{fig:fdrpvalvsdiscretization} shows the resulting levels if the positive correction $\frac{\hat{p}B+1}{B+1}$ is applied for avoiding null $p$-values. The same conclusions can be extracted, the main change being under-rejections instead of over-rejections.

\begin{figure}[h!]
\centering
\includegraphics[width=0.475\textwidth]{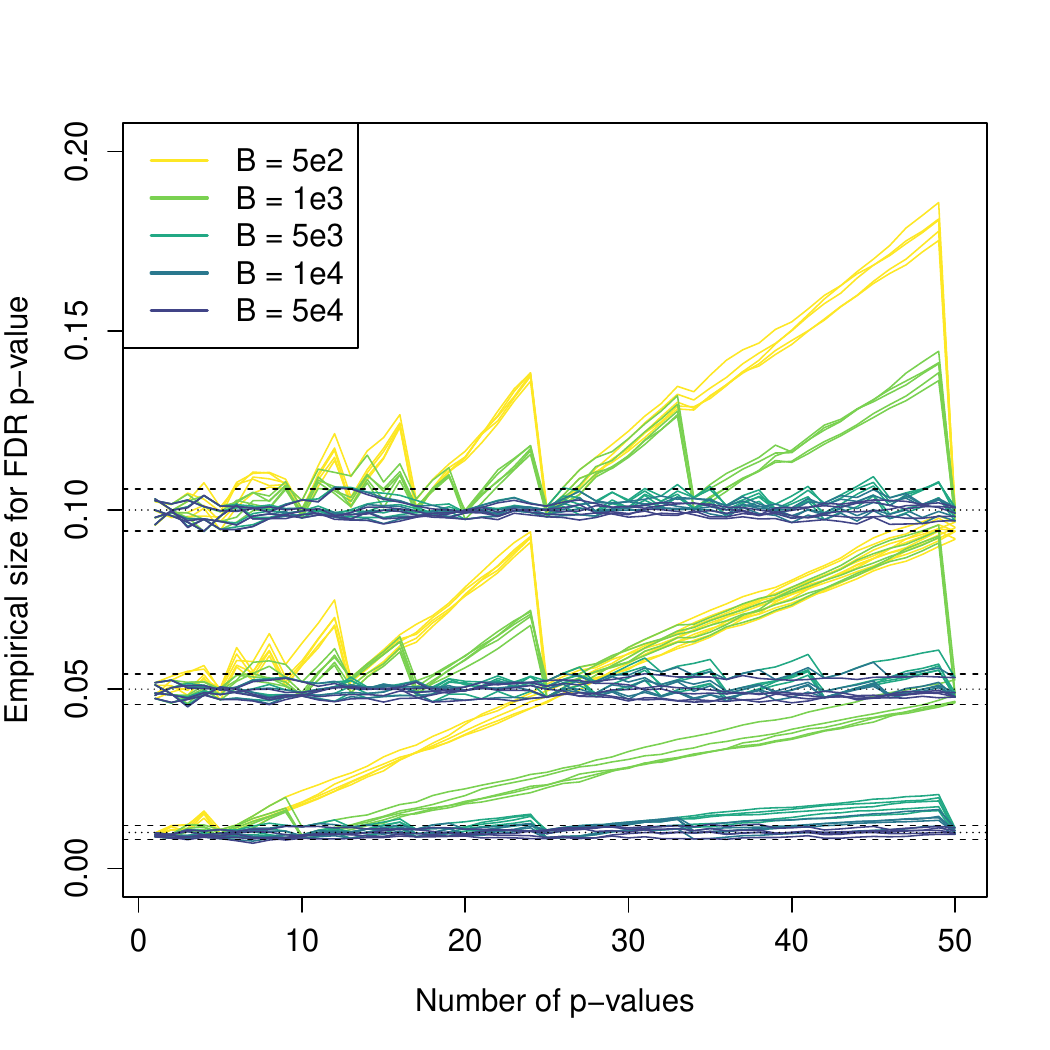}\includegraphics[width=0.475\textwidth]{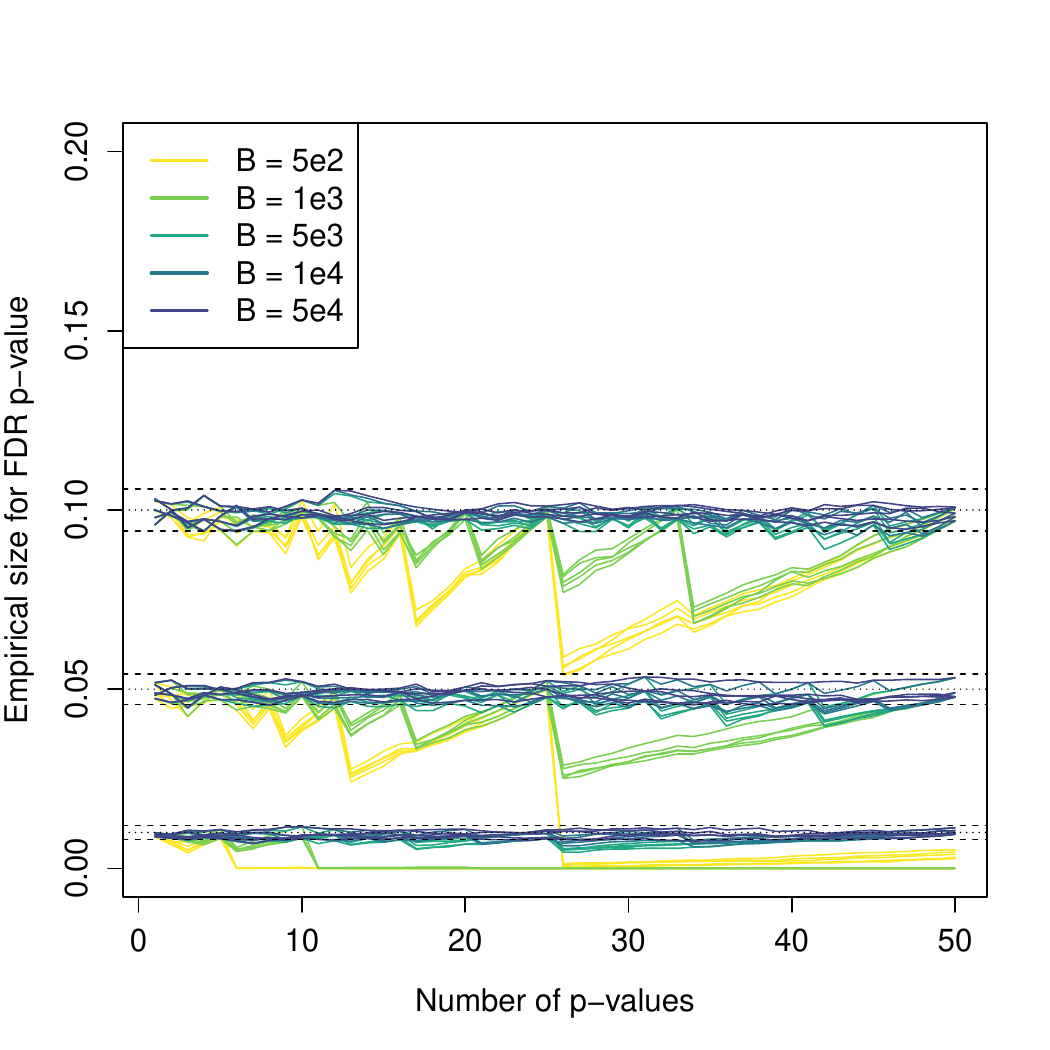}
\caption{\small Left: empirical levels of the test with rejection rule $p_{\mathrm{FDR},K}<\alpha$, as a function of $K$. Right: empirical levels employing a positive correction for the $p$-values.}
\label{fig:fdrpvalvsdiscretization}
\end{figure}


\fi

\end{document}